\pdfoutput=1

\documentclass[%
 reprint,
 superscriptaddress,
 amsmath,amssymb,
 prb,
floatfix,
longbibliography
]{revtex4-2}

\usepackage{graphicx}
\usepackage{dcolumn}
\usepackage{bm}
\usepackage{physics, mathtools}
\usepackage{comment}
\usepackage[caption=false]{subfig}

\usepackage[svgnames,dvipsnames, x11names]{xcolor}
\usepackage{soul} % -- Finally remove it
\usepackage[colorlinks,citecolor=DarkGreen,linkcolor=magenta,linktocpage]{hyperref}
\usepackage[normalem]{ulem}

\usepackage{tikz}
\usetikzlibrary{arrows, calc}
\usetikzlibrary{decorations.pathreplacing}

\usepackage{amsthm}
\theoremstyle{plain}
\newtheorem{theorem}{Theorem}[section]

\newcommand{\totS}{\mathcal{S}}
\newcommand{\totSb}{\bm{\mathcal{S}}}

\begin{document}

\preprint{APS/123-QED}

\title{Quantum many-body scars in spin models with multibody interactions}% 
\author{Kazuyuki Sanada}
\affiliation{Department of Physics, Graduate School of Science, The University of Tokyo, 7-3-1 Hongo, Bunkyo-ku, Tokyo 113-0033, Japan}

\author{Yuan Miao}
\affiliation{Galileo Galilei Institute for Theoretical Physics, INFN, Largo Enrico Fermi 2, 50125 Firenze, Italy}
\author{Hosho Katsura}%

 \affiliation{Department of Physics, Graduate School of Science, The University of Tokyo, 7-3-1 Hongo, Bunkyo-ku, Tokyo 113-0033, Japan}
\affiliation{Institute for Physics of Intelligence, The University of Tokyo, 7-3-1 Hongo, Bunkyo-ku, Tokyo 113-0033, Japan}
\affiliation{Trans-scale Quantum Science Institute, The University of Tokyo, Bunkyo-ku, Tokyo 113-0033, Japan}

\date{\today}

\begin{abstract}
We introduce and study several classes of quantum spin models with multibody interactions that exhibit quantum many-body scars. The models are constructed by two different methods: one exploiting boundary states in integrable spin chains and the other based on a variant of existing methods such as restricted spectrum generating algebras. The first method allows us to construct deformations of the Majumdar-Ghosh and Affleck-Kennedy-Lieb-Tasaki models---prototypes of frustration-free systems. With the second method, we construct a large class of spin-$1$ models involving scalar spin chirality in both one and two dimensions. Interestingly, in some cases, the models so constructed have towers of scar states of different character. For each example, we show that the scar states behave differently from thermal states by comparing their spectral and dynamical properties with those of other states. We also show that a superposition of the scar states constructed by the second method exhibits perfectly periodic revivals in the dynamics.

\end{abstract}

\maketitle

\section{Introduction}
Since the early days of quantum mechanics, thermalization of isolated quantum systems has been of great theoretical interest, as it is at the %foundation 
heart of statistical mechanics. 
Recently, thanks to the development of quantum simulators such as systems with ultracold atoms~\cite{ultracold}, superconducting circuits~\cite{superconductive}, trapped ions~\cite{trapped}, and Rydberg atoms~\cite{Bernien}, we have been able to delve into the quantum many-body dynamics in detail, leading to a better understanding of the nature of thermalization. 
Theoretically, the eigenstate thermalization hypothesis (ETH) was introduced as a plausible mechanism to explain thermalization phenomena in isolated quantum systems~\cite{von2010proof, goldstein2010long, rigol2012alternatives}, and was subsequently discussed in a number of works such as Refs.~\cite{tasaki1998quantum, deutsch1991quantum, srednicki1994chaos, horoi1995chaos, zelevinsky1996nuclear}. Roughly speaking, the ETH is a quantum counterpart of ergodicity in classical systems~\cite{venuti2019ergodicity, deutsch2018eigenstate}.
The strong version of ETH asserts that all energy eigenstates are thermal states~\footnote{In contrast, the weak ETH claims that almost all energy eigenstates are thermal, which allows %the existence of small numbers of 
a small number of exceptional eigenstates called nonthermal states. Actually, the weak ETH has been proved in some cases~\cite{biroli2010effect, iyoda2017fluctuation}.}, %whose behaviors of macroscopic values 
which are locally indistinguishable from the microcanonical average. It has been confirmed by numerical calculations that the ETH is valid for many isolated quantum systems~\cite{rigol2008thermalization, polkovnikov2011colloquium, nandkishore2015many}.

Even though the ETH has been tested and confirmed in many studies, %by many examples, 
it does not hold in some special cases~\footnote{Unfortunately, it was reported that there is no general theorem, algorithm, or systematic procedure to determine whether any given quantum many-body system thermalizes or not~\cite{shiraishi2021undecidability}.}.
For example, quantum integrable models and many-body localized systems are known to violate the strong ETH~\cite{nandkishore2015many}. Moreover, there are systems that do not have these characteristics but still have eigenstates that do not thermalize. These nonthermal states are called quantum many-body scars (QMBS)~\cite{serbyn2021quantum, regnault2022quantum, chandran2023quantum}.

{
The signatures of QMBS have been observed in recent experiments~\cite{Bernien, su2023observation, zhang2023many} and several experimental platforms to realize QMBS have been proposed~\cite{hudomal2020quantum, zhao2020quantum, desaules2021proposal, kunimi2023proposal}. Of particular interests are systems involving Rydberg atoms~\cite{Bernien}.
}
Such systems exhibit nonthermal dynamics, despite being non-integrable~\cite{Turner, choi2019emergent}. 
The theoretical understanding of QMBS has progressed rapidly in recent years, and to date, many models with exact QMBS have been known. Examples include the PXP model~\cite{lin2019exact, shiraishi2019connection, lin2020quantum}, the Affleck-Kennedy-Lieb-Tasaki (AKLT) model~\cite{Moudgalya, Mark, o2020tunnels}, the Ising- and XY-like models~\cite{schecter2019weak, Iadecola, Chattopadhyay}, the perturbed Potts model~\cite{moudgalya2020large}, and the Onsager scars~\cite{SYK} (see Ref.~\cite{regnault2022quantum} for a review). They motivated the development of systematic %In addition, there is interest in 
methods for constructing concrete models with exact QMBS~\cite{Shiraishi-Mori, mcclarty2020disorder, pakrouski2020many, pakrouski2021group, ren2021quasisymmetry, tang2021multi, wildeboer2022quantum, ren2022deformed, omiya2023fractionalization}.
{Also, the fate of exact QMBS under perturbations has been a subject of debate~\cite{Turner2, lin2020slow, gotta2023asymptotic}.}  
For a more mathematical approach, an attempt has been made to comprehensively understand QMBS using commutant algebras~\cite{moudgalya2022exhaustive, moudgalya2022hilbert}.
However, despite these developments, the overall picture is far from complete. 
Therefore, to better understand the general framework and origin of QMBS, it is important to explore different methods for constructing new models that host QMBS in a systematic manner. 

In this paper, we introduce and study several classes of spin models with multibody interactions that exhibit QMBS. To construct the models, we employ two different methods: one based on integrable boundary states~\cite{de2015one, Piroli2017WhatIA, de2018scalar, Pozsgay, piroli2019integrable, pozsgay2019integrable}, and the other using a variant of the existing methods based on restricted spectrum generating algebras~\cite{moudgalya2020eta, buvca2019non} or quasi-symmetry groups~\cite{ren2021quasisymmetry}. The first method allows one to construct an infinite family of models with a scar state. However, since this approach heavily relies on the integrability of some terms in the Hamiltonian, its application is limited to one dimension. In addition, with this method, one cannot obtain a tower of scar states with equal energy spacing. In contrast, the second method allows for the construction of models with towers of scar states. We will demonstrate that a superposition of these scar states shows perfectly periodic revivals in the dynamics. Unlike the first method, the second method is capable of constructing models in higher dimensions. We will illustrate this using a model on a triangular lattice as an example. 
It should be noted that both methods allow the models to accommodate designed inhomogeneities that do not affect QMBS. 

The paper is organized as follows. In Sec. \ref{subsec:method_of_construction}, we explain %two methods to construct the scarred models. 
the two methods in more detail. In Sec. \ref{subsec:numerical_verification}, we discuss how to distinguish QMBS from thermal states. In Sec. \ref{sec_MG+CSC}, we consider the spin-$1/2$ Majumdar-Ghosh model deformed by the spin-1/2 scalar spin chirality as an example of a scarred model constructed by the first method. 
In Sec. \ref{sec_AKLT+H3}, we show another example constructed by the same method, namely, the spin-$1$ AKLT model deformed by the third conserved quantity of the ${\rm SU}(3)$ Sutherland model. We also discuss possible generalizations to higher spins. 
In Sec. \ref{sec:AKLT+CSC}, we introduce a model consisting of the AKLT Hamiltonian and the spin-$1$ scalar spin chirality as an example of a model constructed by the second method. 
In Sec. \ref{sec_CSC}, we first construct exact zero-energy eigenstates of the spin-$1$ scalar spin chirality term. Then we show that they form towers of scar states in a class of models obtained by perturbing the scalar spin chirality by tailored disorder and discuss that they are examples of models constructed by the second method.  
We conclude with a summary and some open questions in Sec. \ref{sec_discussion}. Some technical details are relegated to the Appendices. 

\section{Methods}\label{sec:method}
\subsection{Construction of scarred models}\label{subsec:method_of_construction}
To construct models with QMBS, we adopt the following two methods (i) the method based on integrable boundary states~\cite{Piroli2017WhatIA, de2018scalar, Pozsgay, piroli2019integrable, pozsgay2019integrable}, and (ii) the method relying on a tower of states generated by some operator. First, let us describe (i), which is deeply related to quantum integrable systems. 
Let $H$ be a nearest-neighbor integrable Hamiltonian. One can construct an infinite number of conserved quantities $Q_n$ successively starting from $Q_2 \propto H$ by $Q_{n+1} = [B, Q_n]$, where $B$ is the boost operator~\cite{sklyanin1992quantum, GRABOWSKI1995299, de2019classifying}
\footnote{The boost operator also works for non-difference-form R matrices, see Ref.~\cite{de2019classifying}.}.
Each operator $Q_n$ can be written as a sum of local operators spanning at most $n$ consecutive sites. The conserved quantities can be divided into two groups by their behavior under spatial reflection. The even ones, $Q_{2n}$, are symmetric under the parity operation, whereas the odd ones $Q_{2n+1} $ are anti-symmetric.

An integrable boundary state, say $\ket{\Psi_0}$, is defined as a state that is annihilated by all odd conserved charges \cite{Piroli2017WhatIA}, i.e., 
\begin{equation}
    Q_{2k+1}\ket{\Psi_0} = 0
\end{equation}
for all $k = 1, 2,\ldots$. 
One can see that if $\ket{\Psi_0}$ is an energy eigenstate of another Hamiltonian $H_0$, %it is expected that 
then $\ket{\Psi_0}$ is an exact eigenstate of the Hamiltonian 
\begin{align}
H(t_1, t_2, \ldots, t_n) = H_0 +\sum_{k=1}^n t_k Q_{2k+1}, 
\end{align}
where $\{t_k\}^n_{k=1}$ are real numbers. If this new Hamiltonian is non-integrable and the energy of $\ket{\Psi_0}$ is in the middle of the spectrum, then $\ket{\Psi_0}$ is likely to be a scar state. Note that this method allows one to construct an enormous number of scarred models by changing parameters $\{t_k\}^n_{k=1}$. 

Next, we describe the second approach (ii). This is a variant of the existing methods~\cite{SYK, moudgalya2020eta, buvca2019non, ren2021quasisymmetry}. 
We first assume that the tower of states generated by an operator $\mathcal{Q}$, namely, $\mathcal{Q}$ $\ket{\psi}, \mathcal{Q}\ket{\psi}, \mathcal{Q}^2\ket{\psi}, \cdots, \mathcal{Q}^n \ket{\psi}$, are energy eigenstates of some Hamiltonian $H$. Then, we can create the system with QMBS by considering the property of operator $\mathcal{Q}$ other than each eigenstate. To be more specific, if there exists an operator $\Pi$ such that $\Pi\mathcal{Q}^k \ket{\psi} = \lambda_k \mathcal{Q}^k \ket{\psi}$ ($\lambda_k \in \mathbb{R}$) for all $k=0,1,2,...,n$, each $\mathcal{Q}^k \ket{\psi}$ is also an eigenstate of the new Hamiltonian $H' = H + \Pi$. However, almost all eigenstates of $H$ are no longer eigenstates of $H'$. Thus, it is highly likely that $\mathcal{Q}^k \ket{\psi}$ are the only nonthermal eigenstates of the Hamiltonian $H'$. 

\subsection{Numerical verification}\label{subsec:numerical_verification}
To discuss whether a given model has QMBS or not, we need to answer at least the following two questions: 
\begin{itemize}
\item Is the model non-integrable? 
\item Are the likely scar states nonthermal? 
\end{itemize}

We can answer the first question by checking the level-spacing distribution. It is defined as follows. Let $E_1\leq E_2\leq \cdots \leq E_N $ be the eigenenergies of a Hamiltonian in ascending order. The normalized level spacing $s_i$ is then defined as $s_i \coloneqq (E_{i+1} - E_i) / \delta$, where $\delta\coloneqq (E_N-E_1)/(N - 1)$ denotes the average over all neighboring level spacings. Then, the level-spacing distribution function $P(s)$ is defined such that $P(s)\Delta s$ is the probability of finding $s_i$ in the interval $[s, s+\Delta s]$.
It is empirically known that the level-spacing distribution follows the Poisson distribution
\begin{equation}
    P(s)_\mathrm{Poisson}=e^{-s}
\end{equation}
for integrable systems \cite{Berry_1977}, 
%and follows 
whereas for non-integrable models, it follows the Gaussian orthogonal ensemble (GOE)
\begin{equation}
    P(s)_\mathrm{GOE} = \frac{\pi}{2}se^{-\frac{\pi}{4}s^2}
\end{equation}
for a Hamiltonian with time-reversal symmetry, or it follows the Gaussian unitary ensemble (GUE) 
\begin{equation}
    P(s)_\mathrm{GUE} = \frac{32}{\pi^2}s^2e^{-\frac{4}{\pi}s^2}
\end{equation}
for a Hamiltonian without time-reversal symmetry~\cite{Berry_1981, Bohigas_1984, Szasz-Schagrin, GAUDIN1961447}.
To check which distribution the level spacing follows, we can use $r$-value other than the histogram of the level spacings $s_i$~\cite{oganesyan2007localization}. It is defined as follows: let $r_i = \min\left(s_i/s_{i+1}, s_{i+1}/{s_i} \right)$ be a ratio of neighboring level spacings, and let $\langle r \rangle$ be the $r$-value, the average of $r_i$. By calculating the $r$-value for each of the distributions, we get $\langle r_{\mathrm{Poisson}}\rangle = 2\ln 2 -1 \approx 0.386$ for Poisson, $\langle r_{\mathrm{GOE}} \rangle = 4-2\sqrt{3} \approx 0.536$ for GOE, and $\langle r_\mathrm{GUE} \rangle = \frac{2\sqrt{3}}{\pi}-\frac{1}{2} \approx 0.603$ for GUE~\cite{atas2013distribution}.

To answer the second question, we examine several physical quantities for all energy eigenstates. 
In particular, we use the entanglement entropy as a diagnostic tool to identify nonthermal states. It is defined as follows. 
Let $\ket{\psi}$ be a state of the system and let $A$ be a subsystem. 
Then the reduced density matrix of $A$ is defined as $\rho_A = \Tr_{B} (\dyad{\psi}{\psi})/\!\braket{\psi}{\psi}$, where $B$ is the complement of $A$. 
The entanglement entropy of $\ket{\psi}$ is then defined as
\begin{equation}
    S_A(\ket{\psi}) = -\Tr_{A}[\rho_A\ln\rho_A].
\end{equation}
It is known that the entanglement entropy of a thermal state obeys a volume law,
i.e., $S_A$ is proportional to the system size~\cite{mori2018thermalization}. 
On the other hand, nonthermal states have sub-volume-law entanglement entropy even if they are in the middle of the energy spectrum. 
Therefore, nonthermal states such as QMBS can be identified as low-entanglement outliers in the plot of energy versus $S_A$. In this paper, we calculate the half-system entanglement entropy for one- and two-dimensional systems (see Fig. \ref{fig_1D_sublattice} and \ref{fig:pic_of_triangle}). 

\begin{figure}[tbp]
    \centering
    \includegraphics[width=\linewidth]{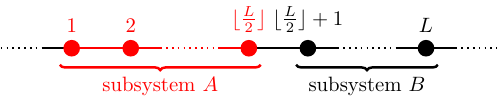}    
    \caption{The definition of subsystems $A$ and $B$ for one-dimensional systems. Note that the number of sites in %the sublattice 
    $A$ is one less than that in %the sublattice 
    $B$ when $L$ is odd.}
    \label{fig_1D_sublattice}
\end{figure}

\section{Spin-\texorpdfstring{$1/2$}{1/2} Majumdar-Ghosh model + scalar spin chirarity}\label{sec_MG+CSC}

This is one example of a scarred model constructed by method (i) in Sec. \ref{subsec:method_of_construction}.

\subsection{Hamiltonian}
In this section, we consider a one-dimensional spin-$1/2$ model with two- and three-body interactions.
The Hamiltonian of the model depends on a parameter $t \in \mathbb{R}$ and is given by
\begin{equation}\label{MG+SC}
H(t) = H_{\mathrm{MG}} + t C_{\mathrm{SC}},
\end{equation}
where 
\begin{align}
H_{\mathrm{MG}} &= \sum_{j=1}^L \left[(\bm{S}_j+\bm{S}_{j+1}+\bm{S}_{j+2})^2-\frac{3}{4}\right], \\
C_{\mathrm{SC}} &= \sum_{j=1}^L \bm{S}_j\cdot(\bm{S}_{j+1}\times\bm{S}_{j+2}) \label{CSC}, 
\end{align}
and ${\bm S}_j =(S^x_j, S^y_j, S^z_j)$ is the spin-1/2 operator acting on %the 
site $j$:
\begin{equation}
    S^x_j = \frac{1}{2}\mqty(0 & 1 \\ 1 & 0)_j, \quad S^y_j = \frac{1}{2}\mqty(0 & -\mathrm{i}\\ \mathrm{i} & 0)_j,\quad S^z_j = \frac{1}{2}\mqty(1 & 0 \\ 0 & -1)_j
\end{equation}
We %here took 
impose periodic boundary conditions and assume that the number of sites $L$ is even.
The first term $H_\mathrm{MG}$ is the Hamiltonian of the Majumdar-Ghosh model exhibiting exact dimer ground states~\cite{Majumdar, Majumdar2, Caspers}, while the second term $C_\mathrm{SC}$ is the scalar spin chirality~\cite{Wen}. Note that $C_\mathrm{SC}$ is the third conserved charge of the spin-1/2 Heisenberg model~\cite{Frahm}. 
Physically, this term appears at third order in perturbation theory starting from the ${\rm SU}(2)$ Hubbard model at half-filling in an external magnetic field~\cite{SU2scalar}. 
{We also note in passing that a similar three-spin interaction has recently been realized experimentally in Rydberg atom arrays~\cite{kim2023realization}.}

\subsection{Symmetries and non-integrability}
The Majumdar-Ghosh model has several symmetries: the Hamiltonian is invariant under time-reversal $\Theta$: ${\bm S}_j \mapsto -{\bm S}_j$, SU(2) spin rotation, translation ${\cal T}$: ${\bm S}_j \mapsto {\bm S}_{j+1}$, bond-centered inversion ${\cal I}_\mathrm{b}$: ${\bm S}_j \mapsto {\bm S}_{L-j+1}$, site-centered inversion ${\cal I}_\mathrm{s}$: ${\bm S}_j \mapsto {\bm S}_{-j}$, and spin-flip 
${\cal F}$: $\bm{S}_j \mapsto (S^x_j, -S^y_j, -S^z_j)$~\footnote{The group $\{ 1, {\cal F}\}$ is a discrete subgroup of ${\rm SU}(2)$, where ${\cal F}$ corresponds to a $\pi$ rotation around the $x$ axis.}. 
Among these symmetries, time-reversal and bond-centered- and site-centered-inversion symmetries are absent in the scalar spin chirality $C_\mathrm{SC}$. However, the combination of $\Theta$ and ${\cal I}_\mathrm{s}$ leaves $C_\mathrm{SC}$ invariant, which we call pseudo-time-reversal symmetry. 
Therefore, the entire model $H(t)$ has ${\rm SU}(2)$, translation, spin-flip, and pseudo-time-reversal symmetries. Among them, the first three are unitary symmetries and allow us to diagonalize the Hamiltonian sector by sector. For convenience we define the total spin operators as $\totS^\alpha:=\sum^L_{j=1}S^\alpha_j$ ($\alpha=x,y,z$) and write the eigenvalue of $\mathrm{SU}(2)$ Casimir operator $\totSb^2=\sum_{\alpha=x,y,z}(\totS^\alpha)^2$ as $\totS (\totS + 1)$. 
With a slight abuse of notation, we will denote the eigenvalues of the operators $\totS^z$, ${\cal T}$, and ${\cal F}$ by the same symbols. 

The Hamiltonian Eq. (\ref{MG+SC}) is non-integrable. 
This can be shown by studying the level-spacing statistics in a symmetry sector labeled by $\totS^z$, ${\cal T}$, and ${\cal F}$. 
Figure~\ref{fig:lsp-MG+SC} clearly shows that the level-spacing distribution of $H(t)$ is close to the GOE Wigner-Dyson distribution. This is consistent with the pseudo-time-reversal symmetry of the model. 
In addition, the $r$-value calculated from the histogram in Fig.~\ref{fig:lsp-MG+SC} is $\langle r \rangle \simeq 0.538$, which is close to $\langle r_{\rm GOE} \rangle$. 
Thus, we conclude that the model (\ref{MG+SC}) is non-integrable. 

{
Figure {\ref{fig:r_vs_t}} shows $\langle r \rangle$ as a function of $t$ for the model ({\ref{MG+SC}}) with different system sizes. Clearly, the results for $L=18$ and $20$ have the same trend. When $t<20$, the $r$-value $\langle r \rangle$ is close to the GOE value $0.536$, whereas $\langle r \rangle$ gets closer to the Poisson value $0.386$ as $t$ increases further. This implies that for large $t$ the whole Hamiltonian ({\ref{MG+SC}}) is dominated by the integrable part $t C_{\rm SC}$ and the system behaves more like an integrable system.} 
%\hl{Note that $r$-value in the Hamiltonian ({\ref{MG+SC}}) gets close to the value of integrable models $\expval{r}\simeq 0.386$ in large $t$ as indicated by Fig. {\ref{fig:r_vs_t}}. Thus, if $t$ becomes too large, the whole Hamiltonian ({\ref{MG+SC}}) does not become sufficiently non-integrable.} 

\begin{figure}[tbp]
\centering
\includegraphics[width=\linewidth]{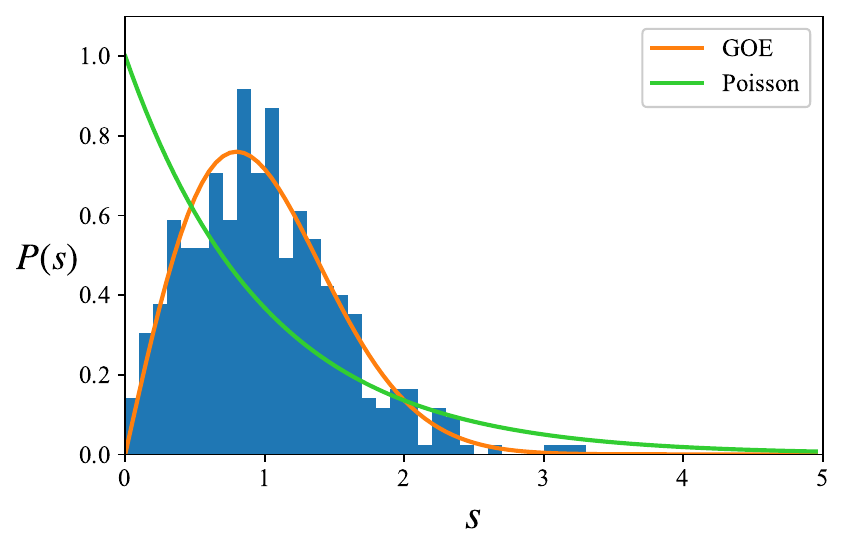}
\caption{Level-spacing statistics in the middle half of the spectrum of $H(t)$ in Eq. (\ref{MG+SC}) with $t = 8$ and $L=20$. The data are taken in the symmetry sector where $(\totS^z, \totS, {\cal T}, {\cal F})=(0, 0, 1, 1)$. %$S^z = 0$, translation $= 1$, spin-flip $= 1$. 
The curves $P(s)_\mathrm{GOE}$ (orange) and $P(s)_\mathrm{Poisson}$ (green) are shown for comparison. 
The distribution follows $P(s)_\mathrm{GOE}$. }
\label{fig:lsp-MG+SC}
\end{figure}

\begin{figure}[tbp]
    \centering
    \includegraphics[width=\linewidth]{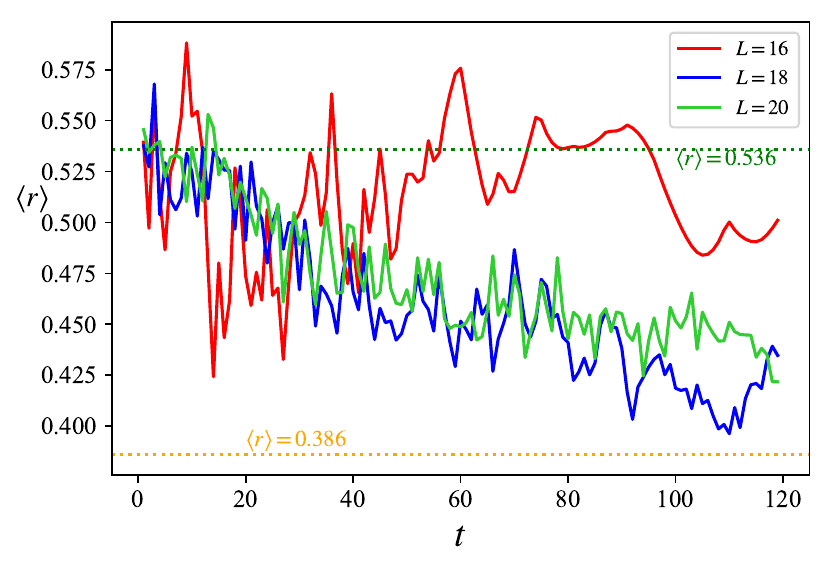}
    \caption{{The mean level-spacing ratio $\langle r \rangle$ as a function of $t$ for the Hamiltonian ({\ref{MG+SC}}) in the symmetry sector $(\totS^z, \totS, {\cal T}, {\cal F})=(0, 0, 1, 1)$ for $L=16$, $18$, and 20. 
    %The $t$-dependence of $r$-values of the Hamiltonian ({\ref{MG+SC}}) in the symmetry sector $(\totS^z, \totS, {\cal T}, {\cal F})=(0, 0, 1, 1)$ for $L=16, 18$, and 20. 
    The green and orange dotted lines indicate $\langle r_{\mathrm{GOE}} \rangle \approx 0.536$ and $\langle r_{\mathrm{Poisson}} \rangle \approx 0.386$, respectively.}}
    \label{fig:r_vs_t}
    % Note: There is a scattering chart version -> r_vs_t_MG+CSC_scatter.pdf
\end{figure}

\subsection{Scar states}
The zero-energy ground states of $H_\mathrm{MG}$ can be written as the following dimer states
\begin{align}
\ket{\Psi_1} & =\ket{\mathrm{sing}}_{1, 2}\otimes\ket{\mathrm{sing}}_{3, 4}\otimes\cdots\otimes\ket{\mathrm{sing}}_{L-1, L}, \\
\ket{\Psi_2} & = \ket{\mathrm{sing}}_{2, 3}\otimes\ket{\mathrm{sing}}_{4, 5}\otimes\cdots\otimes\ket{\mathrm{sing}}_{L, 1},
\end{align}
where $\ket{\mathrm{sing}}_{i,j} = \frac{1}{\sqrt{2}}(\ket{\uparrow\downarrow}_{i,j}-\ket{\downarrow\uparrow}_{i,j})$ denotes the normalized spin singlet formed between site $i$ and $j$. The two states are related to each other by $\ket{\Psi_2} = {\cal T}\ket{\Psi_1}$.

As discussed in \cite{Piroli2017WhatIA}, these states are integrable boundary states of the spin-$1/2$ Heisenberg XXX chain, meaning that they are annihilated by all parity-odd conserved charges of the Heisenberg Hamiltonian. Since $C_\mathrm{SC}$ is one of the parity-odd conserved charges, it is clear that $\ket{\Psi_1}$ and $\ket{\Psi_2}$ are simultaneously annihilated by both $H_\mathrm{MG}$ and $C_\mathrm{SC}$. Thus they are zero-energy eigenstates of $H(t)$ in Eq. (\ref{MG+SC}) for all $t$. 
We now argue that the states $\ket{\Psi_1}$ and $\ket{\Psi_2}$ can be thought of as QMBS. To this end, we compute the half-chain entanglement entropies ($S_A$) of all energy eigenstates {for several system sizes}.
\begin{figure*}[tbp]
\captionsetup[subfloat]{labelformat = empty}
\subfloat[][$L=14$]{
    \centering
    \includegraphics[keepaspectratio, width=0.3\linewidth]{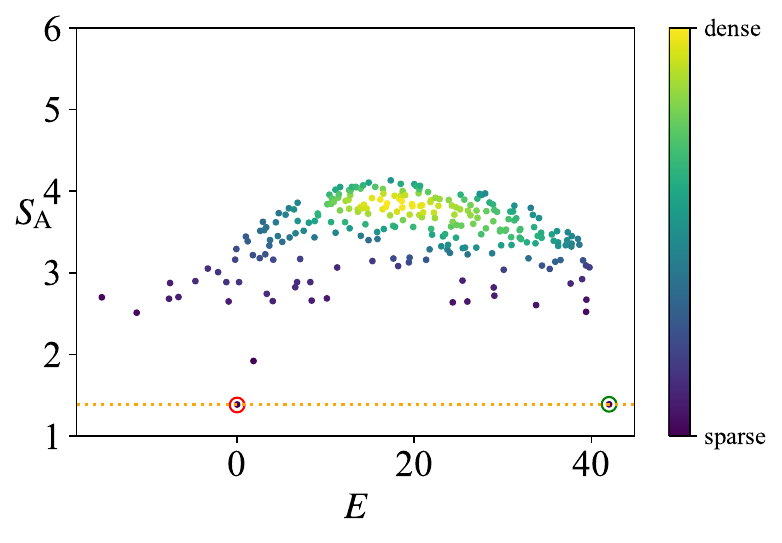}
    %\caption{$L=14$}
}\quad
\subfloat[][$L=16$]{
    \centering
    \includegraphics[keepaspectratio, width=0.3\linewidth]{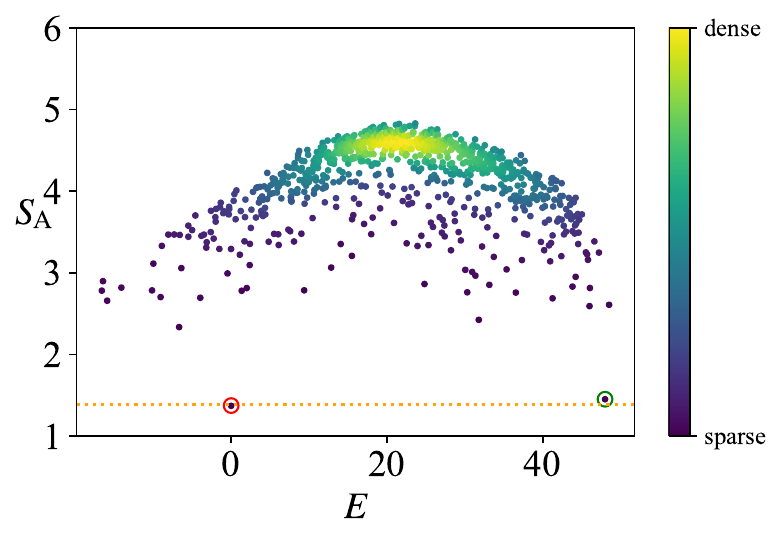}
    %\caption{$L=16$}
}\quad
\subfloat[][$L=18$]{
    \centering
    \includegraphics[keepaspectratio, width=0.3\linewidth]{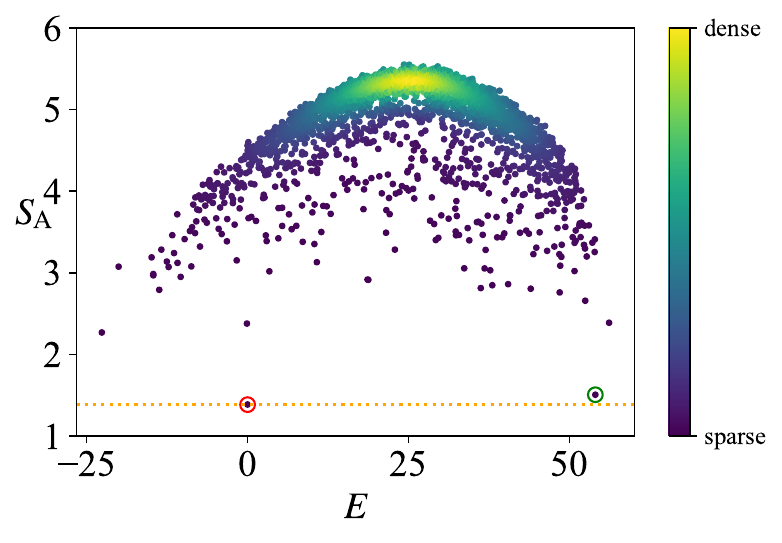}
    %\caption{$L=18$}
}
\caption{{Entanglement entropy $S_A$ in all eigenstates of $H(t)$ in Eq. (\ref{MG+SC}) with $t=8$ in the symmetry sector $(\totS^z, {\cal T})=(0,1)$ for $L=14$, $16$, and $18$. The density of data points is color coded.}  
The red and green circles indicate the dimer state $\ket{\mathrm{dimer}}$ and the ferromagnetic state $\ket{F_{L/2}}$, respectively. 
The orange dotted line indicates
$S_A = 2\ln 2 \simeq 1.386$.
}
\label{fig:MG_ent}
\end{figure*}

Figure \ref{fig:MG_ent} shows the results in the subspace spanned by translation-invariant states with zero magnetization, i.e., $(\totS^z, {\cal T})=(0,1)$. Clearly, there is a low-entanglement state distinguished from the other at zero energy. 
This state can be identified as $ \ket{\mathrm{dimer}} = (2+(-\frac{1}{2})^{\frac{L}{2}-2})^{-\frac{1}{2}}(\ket{\Psi_1}+\ket{\Psi_2})$, which is invariant under translation by one site. 
It is known that the half-chain entanglement entropy of $\ket{\mathrm{dimer}}$ becomes $S_A = 2\ln 2$ for sufficiently large $L$~\cite{Ramkarthik}. We can see that the entanglement entropy of the dimer state matches this value. 
{Furthermore, this state remains an outlier from the rest of the states with increasing $L$.}

We note that the low-entanglement state near the upper edge of the spectrum ($E=54$) is a ferromagnetic state with zero magnetization written as $\ket{F_{L/2}} = (\totS^-)^{L/2}\ket{\Uparrow}$, where $\totS^{-} := \totS^x - \mathrm{i} \totS^y$ and $\ket{\Uparrow}$ denotes the all-up state. 
The asymptotic form of the half-chain entanglement entropy of this state can be read off from Eq. (16) of \cite{popkov2005logarithmic} (see also Appendix \ref{appendix:ferro_ent_half}). The result reads
\begin{equation}\label{eq:EE_ShalfFM}
    S_A(\ket{F_{L/2}}) \approx \frac{1}{2}\ln{L} +\frac{1}{2} \ln\frac{e\pi}{8} \quad(L \gg 1).
\end{equation}
\begin{figure}
    \centering
    \includegraphics[width=0.8\linewidth]{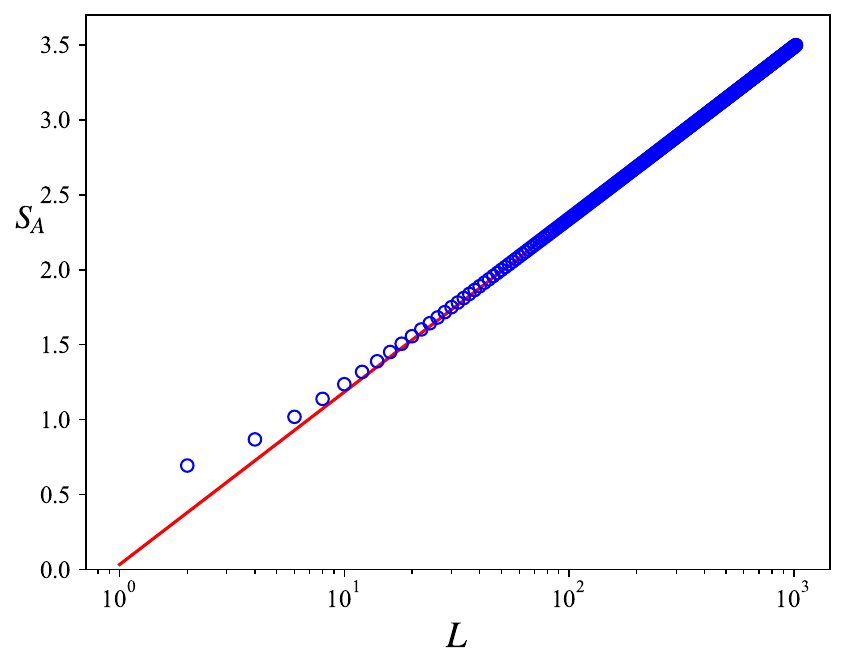}
    \caption{{Size dependence of the half-chain entanglement entropy of the ferromagnetic state $\ket{F_{L/2}}$. The red line represents the right-hand side of Eq. ({\ref{eq:EE_ShalfFM}}).}}
    \label{fig:SA_size}
\end{figure}
\noindent
{Figure {\ref{fig:SA_size}} shows the size dependence of $S_A(\ket{F_{L/2}})$. It clearly demonstrates that the entanglement entropy obeys a sub-volume law $S_A \sim \ln(L)$.}
For $L=18$, we obtain $S_A (\ket{F_{L/2}}) \approx 1.478$, which agrees with the numerical result shown by the green circle in Fig. \ref{fig:MG_ent}. It should be noted that the ferromagnetic state is not an example of a scar state because the state belongs to the subspace with maximum total spin, which is an irreducible representation of the global ${\rm SU}(2)$ symmetry of the model Eq. (\ref{MG+SC}).

We remark that since the dimer states are integrable boundary states of the spin-$1/2$ Heisenberg chain, one can construct other models involving higher-order conserved charges $Q_{2k+1}$ ($k > 1$) which have the dimer states as QMBS. 
See Refs.~\cite{grabowski1994quantum, GRABOWSKI1995299} for the explicit expressions of $Q_{2n+1}$.  

Another characteristic of QMBS is that the expectation values of physical quantities in these states do not match the microcanonical averages. Thus, we can identify QMBS by comparing the expectation value of an observable for each energy eigenstate. 
{Here, we choose the staggered Heisenberg Hamiltonian
\begin{equation}
    H_\mathrm{st} = \sum_{j=1}^L(1+(-1)^j\epsilon)\bm{S}_j\cdot\bm{S}_{j+1}
\end{equation}
as a %non-integrable 
generic observable and calculate the expectation value $\expval{H_\mathrm{st}}= \ev{H_\mathrm{st}}{\psi}$ for each normalized eigenstate $\ket{\psi}$ of $H(t)$ in Eq. (\ref{MG+SC}). Figure \ref{fig:Majumder_Ghosh_XXX} shows the numerical results for $t=8$ and $L=18$ in the symmetry sector $({\cal S}^z, {\cal T}) = (0, 1)$.}
%calculation of the expected values $\expval{H_\mathrm{XXX}}$ for eigenstates of $H(t)$ in Eq. (\ref{MG+SC}). 
{
They clearly indicate that the expectation value in the dimer state, which is calculated as 
\begin{equation}\label{eq:dimer_expXXX}
    \expval{H_\mathrm{st}}{\mathrm{dimer}} = - \frac{3L \left[
    (-\frac{1}{2})^{\frac{L}{2}}+\frac{1}{4} 
    \right]}{2+(-\frac{1}{2})^{\frac{L}{2}-2}} ,
\end{equation}
is far from those in the states near $E=0$.}

Although we have shown the results only for the symmetry sector $({\cal S}^z, {\cal T}) = (0, 1)$, we note that similar results hold for $({\cal S}^z, {\cal T}) = (0, -1)$, where $|\mathrm{dimer}^\prime \rangle \propto | \Psi_1 \rangle - | \Psi_2 \rangle$ is singled out as a scar state. 

\begin{figure}[tbp]
    \centering
    \includegraphics[width=\linewidth]{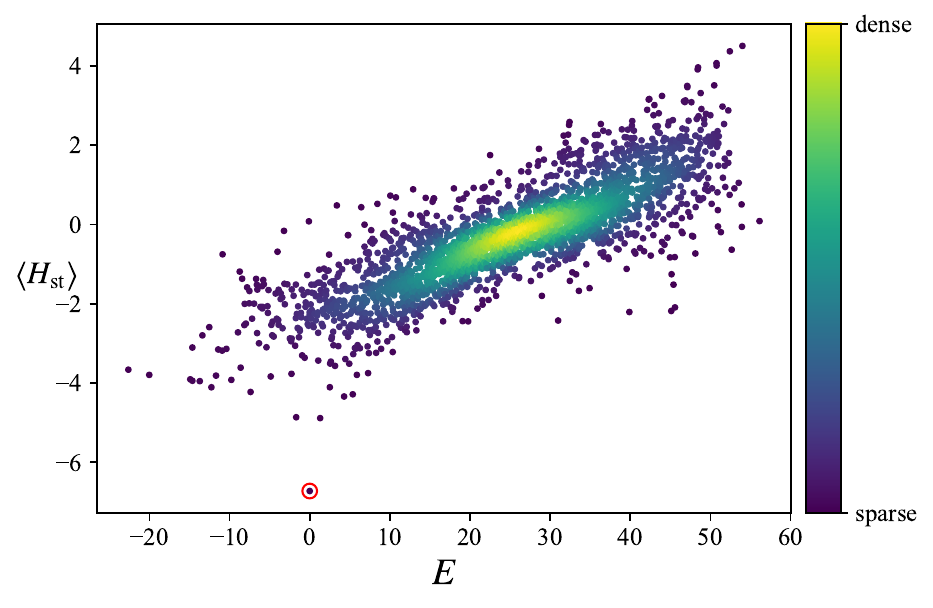}
    \caption{{The expectation values of $H_\mathrm{st}$ with $\epsilon=0.2$ in all eigenstates of $H(t)$ in Eq. (\ref{MG+SC}) with $t=8$, and $L = 18$ in the symmetry sector $(\totS^z, {\cal T})=(0,1)$. The density of data points is color coded. The red circle indicates the dimer state $\ket{\mathrm{dimer}}$ with $\expval{H_\mathrm{st}} \simeq -6.72$, which agrees with the analytical value $\expval{H_\mathrm{st}}=-1143/170$ obtained from Eq. (\ref{eq:dimer_expXXX}).}}
    %\epsilon = 0.2
    \label{fig:Majumder_Ghosh_XXX}
\end{figure}

\section{Spin-1 AKLT model + \texorpdfstring{$H_3$}{Lg}} \label{sec_AKLT+H3}

This is another example of a model with QMBS constructed by method (i) in Sec. \ref{subsec:method_of_construction}.

\subsection{Hamiltonian}
In this section, we consider a spin-$1$ model with two- and three-site interactions. Consider a spin-$1$ chain of length $L$ with periodic boundary conditions, and let ${\bm S}_j = (S^x_j, S^y_j, S^z_j)$ be the %$SU(2)$ 
operators of the spin-$1$ representation of the ${\rm SU}(2)$ algebra acting on site $j$:
\begin{align}
\label{spin1 operator}
    S^x_j&=\frac{1}{\sqrt{2}}\left(\begin{array}{ccc}
        0 & 1 & 0 \\
        1 & 0 & 1 \\
        0 & 1 & 0 \\
    \end{array}\right)_j,\,
    S^y_j = \frac{1}{\sqrt{2}}\left(\begin{array}{ccc}
        0 & -\mathrm{i} & 0 \\
        \mathrm{i} & 0 & -\mathrm{i} \\
        0 & \mathrm{i} & 0 \\
    \end{array}\right)_j,\,\nonumber\\
    S^z_j &= \left(\begin{array}{ccc}
        1 & 0 & 0 \\
        0 & 0 & 0 \\
        0 & 0 & -1 \\
    \end{array}\right)_j. 
\end{align}
As usual, we define the total spin operators by $\totS^\alpha:=\sum^L_{j=1}S^\alpha_j$ ($\alpha=x,y,z$) and write the eigenvalue of the Casimir operator $\totSb^2=\sum_{\alpha=x,y,z}(\totS^\alpha)^2$ as $\totS (\totS+1)$. 

\medskip

The Hamiltonian of the model is given by
\begin{equation}\label{AKLT+H3}
H(t) = H_{\mathrm{AKLT}}+t H_3,
\end{equation}
where
\begin{equation}\label{HAKLT}
H_\mathrm{AKLT} = \sum_{j=1}^L\left[\bm{S}_j\cdot\bm{S}_{j+1}+\frac{1}{3}(\bm{S}_j\cdot\bm{S}_{j+1})^2+\frac{2}{3}\right]
\end{equation}
is the Affleck-Kenedy-Lieb-Tasaki (AKLT)  Hamiltonian~\cite{Affleck, Affleck2, book_Tasaki} and 
\begin{equation}\label{H3}
H_3 = \sum_{j=1}^L\sum^8_{a,b,c=1}f_{abc}\lambda_j^a\lambda_{j+1}^b\lambda_{j+2}^c.
\end{equation}
is the third conserved quantity of the ${\rm SU}(3)$ Sutherland model~\cite{Sutherland, Lai_1974, Uimin_1970, Grabowski}. %and thus integrable. 
Here $\lambda^a_j$ ($a=1,...,8$) represent the Gell-Mann matrices acting on site $j$, and $f_{abc}$ are the structure constants determined by $[\lambda^a, \lambda^b] = 2\mathrm{i}f_{abc}\lambda^c$. 
The term $H_3$ is an ${\rm SU}(3)$ generalization of the scalar spin chirality. 
This can be seen by noting that the scalar spin chirality discussed in Sec.~\ref{sec_MG+CSC} can be rewritten as
\begin{equation}
    C_\mathrm{SC} = \sum^L_{j=1}\sum_{\substack{a,b,c \\ =x,y,z}} \epsilon_{abc}S^a_jS^b_{j+1}S^c_{j+2},
\end{equation}
where $\epsilon_{abc}$ is the totally anti-symmetric tensor, which is also known as the structure constants of the ${\rm SU}(2)$ algebra. 
 
There is an interesting alternative expression for $H_3$. Let $P_{i,j}$ be the permutation operator that swaps the state at site $i$ with the state at site $j$: 
\begin{equation}\label{perm}
P_{i, j}\ket{\ldots, s_i, \ldots, s_j, \ldots} = \ket{\ldots, s_j, \ldots, s_i, \ldots}. 
\end{equation}
Here, $s_i \in \{+, 0, -\}$ denotes the spin state at the site $i$. Using $P_{i,j}$,  we define the three-site ring-exchange operator as $P_{i,j,k} \coloneqq P_{j,k}P_{i,j}$. Then %it is known that 
the following relation 
%between the inner product of Gell-Mann matrices %$\lambda_i^a, \lambda_j^a$ and the operator $P_{i, j}$
holds~\cite{oh2017proposal}:
\begin{equation}
    \bm{\lambda}_i\cdot\bm{\lambda}_j = 2P_{i,j}-\frac{2}{3},
\end{equation}
where $\bm{\lambda}_i = (\lambda_i^1, \lambda_i^2, \cdots, \lambda_i^8)$ is a collection of the eight Gell-Mann operators acting on site $i$. Then, using this relation and $\sum_c f_{abc}\lambda^c_{j+2} = \frac{1}{2 \mathrm{i}}[\lambda^a_{j+2}, \lambda^b_{j+2}]$, we arrive at the alternative expression for $H_3$ in terms of $P_{i,j,k}$:
\begin{equation}\label{perm_rep}
    H_3 = -2\mathrm{i}\sum_{j=1}^L(P_{j, j+1, j+2}-P^\dagger_{j, j+1, j+2}),
\end{equation}
where $P^\dagger_{i, j, k} = P^{-1}_{i, j, k}=P_{i,j}P_{j,k}$. 
We note in passing that there are some studies on spin models containing the ${\rm SU}(3)$ scalar spin chirality term %$P_{i,j,k}$
~\cite{Pijk}. 

\subsection{Symmetries and non-integrability}
The AKLT Hamiltonian is invariant under time-reversal $\Theta$, ${\rm SU}(2)$ spin rotation, translation ${\cal T}$, bond-centered inversion ${\cal I}_\mathrm{b}$, site-centered inversion ${\cal I}_\mathrm{s}$, and spin-flip ${\cal F}$, whereas $H_3$ lacks time-reversal and inversion symmetries among them (see Appendix \ref{appendix:ferromagnetic} for a discussion of the ${\rm SU}(2)$ symmetry of $H_3$). 
However, the combined symmetry $\Theta {\cal I}_\mathrm{s}$ leaves $H_3$ invariant. Therefore, the model $H(t)$ in Eq. (\ref{AKLT+H3}) has ${\rm SU}(2)$, translation, spin-flip, and pseudo-time-reversal symmetries. 

Since the ${\rm SU}(3)$ Sutherland model is integrable, the third conserved charge $H_3$ can be considered as a quantum integrable Hamiltonian as well. On the other hand, the AKLT Hamiltonian is non-integrable. Thus, the model Eq. (\ref{AKLT+H3}), an interpolation between the two, is expected to be non-integrable. To verify this, we compute the level-spacing statistics (Fig. \ref{lsp-AKLT+H3}). The results show that the level-spacing distribution follows the GUE Wigner-Dyson distribution instead of the Poisson distribution, which provides strong evidence that this model is non-integrable. It can also be checked by the $r$-value $\langle r \rangle \simeq 0.599$, which is close to $\langle r_\mathrm{GUE}\rangle \simeq 0.603$. We remark that the model is expected to belong to the GOE class, as it has pseudo-time-reversal symmetry. The observed discrepancy may be due to the crossover between different universality classes~\cite{schierenberg2012wigner, kundu2023signatures} or finite-size effects, which are also pronounced in the PXP model~\cite{Turner, Turner2}. 

\begin{figure}[htpb]
\centering
\includegraphics[width=\linewidth]{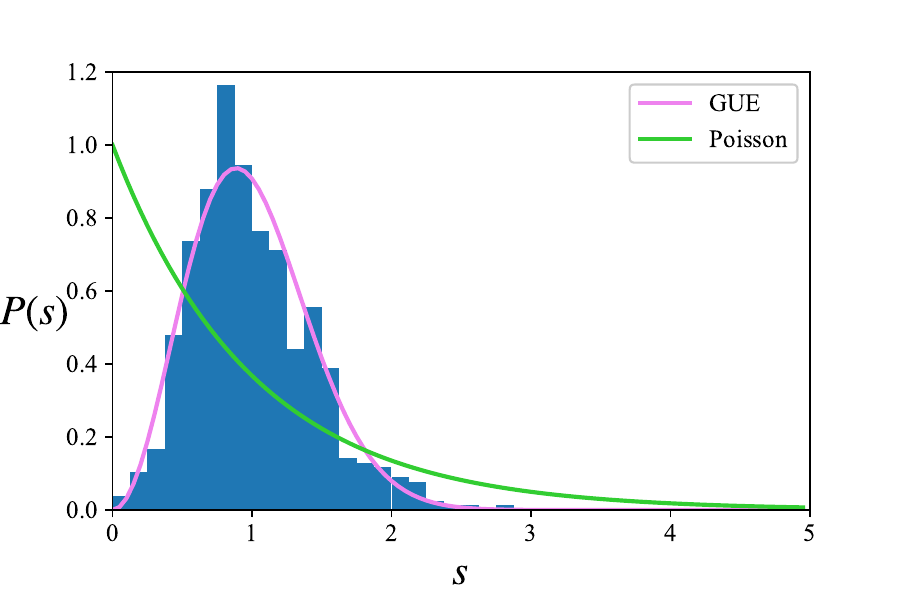}
\caption{Level spacing statistics in the middle half of the spectrum of $H(t)$ in (\ref{AKLT+H3}) with $t = 3$ and $L = 14$. The data are taken in the symmetry sector where $(\totS^z, \totS, \mathcal{T}, \mathcal{F}) = (0, 0, 1, 1)$. 
The curves $P(s)_\mathrm{GUE}$ (magenta) and $P(s)_\mathrm{Poisson}$ (green) are shown for comparison. 
The distribution follows $P(s)_\mathrm{GUE}$. 
}
\label{lsp-AKLT+H3}
\end{figure}

\subsection{Scar state}\label{scarsss}
We now argue that the ground state of the AKLT model, known as the valence-bond solid (VBS) state, can be thought of as a scar state. The VBS state can be written as a matrix product state \cite{Affleck2}:
\begin{equation}
\ket{\Psi_{\mathrm{VBS}}} = \sum_{\{s\}}\Tr[A^{s_1}A^{s_2}\cdots A^{s_L}]\ket{s_1, s_2, \ldots, s_L},
\label{eq:spin-1 VBS}
\end{equation}
where 
$s_j \in \{ +, 0, -\}$ denotes the spin state at site $j$ and
\begin{equation}
A^+ = \sqrt{\frac{2}{3}}\sigma^+, A^0 = -\sqrt{\frac{1}{3}}\sigma^z, A^- = -\sqrt{\frac{2}{3}}\sigma^-,
\end{equation}
with $\sigma^\pm$ and $\sigma^z$ being the Pauli matrices. The summation is taken over all possible spin configurations.

The VBS state is an exact ground state of $H_\mathrm{AKLT}$ with zero energy, i.e., $H_\mathrm{AKLT}\ket{\Psi_\mathrm{VBS}} = 0$. In addition, the VBS state is an integrable boundary state of the Sutherland model (see Appendix \ref{H3VBS=0} for details). Thus, $\ket{\Psi_{\mathrm{VBS}}}$ is a zero-energy eigenstate of $H(t)$ in Eq. (\ref{AKLT+H3}) for all $t$ and is likely to be a scar state.
In order to establish this, we need to consider the case of moderate $t$. 
This is because, if $t$ is close to zero, then the energy of the VBS state is near the lower edge of the spectrum. 
However, such a state cannot be thought of as QMBS, as its energy is not in the bulk of the spectrum. 
In addition, $t$ should not be too large so that the model is away from the integrable case ($H_3$). 
With these in mind, we study the model with $t=3$.

To confirm that the VBS state is indeed a scar state, we numerically compute half-chain entanglement entropies $S_A$. The results are shown in Fig. \ref{EE_HAKLT+H3}. 
It is known that %the half-chain entanglement entropy 
$S_A$ of the VBS state is~\cite{hirano2007entanglement, katsura2007exact}
\begin{align}
    S_A(\ket{\Psi_\mathrm{VBS}}) = -3\lambda_A\ln\lambda_A-\lambda_B\ln\lambda_B 
\end{align}    
with
\begin{align}
    \lambda_A &= \frac{1}{4}\frac{(1-p^{\lfloor L/2 \rfloor})(1-p^{\lceil L/2 \rceil})}{1-p^{L-1}}, \\
    \lambda_B &= \frac{1}{4}\frac{(1+3p^{\lfloor L/2 \rfloor})(1+3p^{\lceil L/2 \rceil})}{1-p^{L-1}},
\end{align}
where $p=-1/3$. $\lfloor x \rfloor$ and $\lceil x \rceil$ denote the floor and ceiling functions, respectively. In the thermodynamic limit, we obtain
\begin{equation}
    \lim_{L\to\infty}S_A(\ket{\Psi_\mathrm{VBS}}) = 2\ln 2.
\end{equation}
Figure \ref{EE_HAKLT+H3} clearly shows that the VBS state at $E=0$ with $S_A \sim 2\ln 2$ is isolated from the other states, indicating that the VBS state exhibits a different behavior from the other thermal states.

The other low-entanglement state at $E=18$ is a ferromagnetic state $\ket{F_L}$ (see Appendix \ref{energy_of_ferromagnetic}). This state is not an example of a scar state because it is a state in the subspace with maximum total spin, which is an irreducible representation of the ${\rm SU}(2)$ symmetry of the model Eq. (\ref{AKLT+H3}). 

\begin{figure}[tbp]
\centering
\includegraphics[width=\linewidth]{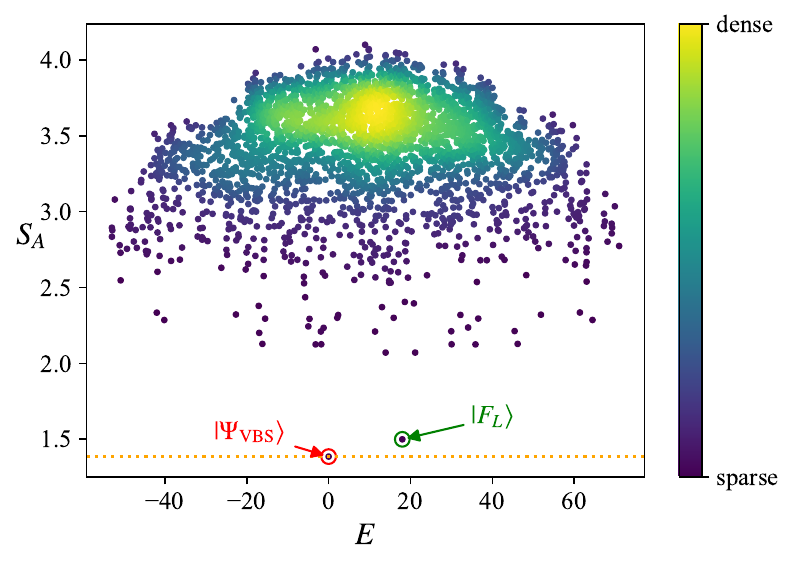}
\caption{Entanglement entropies in all eigenstates of $H(t)$ in Eq. (\ref{AKLT+H3}) for $L = 9$, $t = 3$ in the $\totS^z = 0$ sector. {The density of data points is color coded.} The points enclosed by the red and green circles indicate the VBS state $\ket{\Psi_\mathrm{VBS}}$ and ferromagnetic state $\ket{F_L}$, respectively. The orange dotted line indicates $S_A = 2 \ln 2 \simeq 1.386$.}
\label{EE_HAKLT+H3}
\end{figure}

\subsection{Inhomogeneous generalization}
In the previous model, it was necessary to increase $t$ in order to make the energy density of the VBS state (relative to the ground state) finite. 
However, the problem is that this would increase the effect of $H_3$ and make the behavior of the system more like that of an integrable system. To avoid such a situation, we consider an inhomogeneous generalization of the AKLT Hamiltonian. In this case, the VBS state is still a zero-energy eigenstate of the inhomogeneous Hamiltonian, yet locating in the middle of the spectrum. 
The Hamiltonian of the inhomogeneous model is given by
\begin{align}
 \tilde{H}(t) = \tilde{H}_{\mathrm{AKLT}} + tH_3
\label{inhomogeneous AKLT+H3}
\end{align}
with
\begin{align}
 \tilde{H}_{\mathrm{AKLT}} = \sum_{j=1}^Lc_j\left[\bm{S}_j\cdot\bm{S}_{j+1}+\frac{1}{3}(\bm{S}_j\cdot\bm{S}_{j+1})^2+\frac{2}{3}\right]. 
 \label{inhomogeneous AKLT}
\end{align}
In principle, each coefficient $c_j$ can be any real number. However, for our purpose, we choose $|c_j|\lesssim t$ in order to keep the magnitudes of the two terms ($\tilde{H}_{\mathrm{AKLT}}$ and $H_3$) comparable. 
In the following, we set $t \simeq 1$ and draw $c_j$ uniformly from the interval $[-1,1]$, in which case the model is no longer invariant under the combination of $\Theta$ and ${\cal I}_\mathrm{s}$. 

Like $H(t)$ in Eq. (\ref{AKLT+H3}), ${\tilde H}(t)$ in Eq. (\ref{inhomogeneous AKLT+H3}) is also non-integrable. As Fig. \ref{lsp-randomHAKLT+H3} shows, the level-spacing statistics of the model behave as that of the GUE. Also, the calculated $r$-value $r\simeq 0.598$ is compatible with the $r$-value of the GUE $\langle r_\mathrm{GUE} \rangle \simeq 0.603$. 

\begin{figure}[tbp]
\centering
\includegraphics[width=\linewidth]{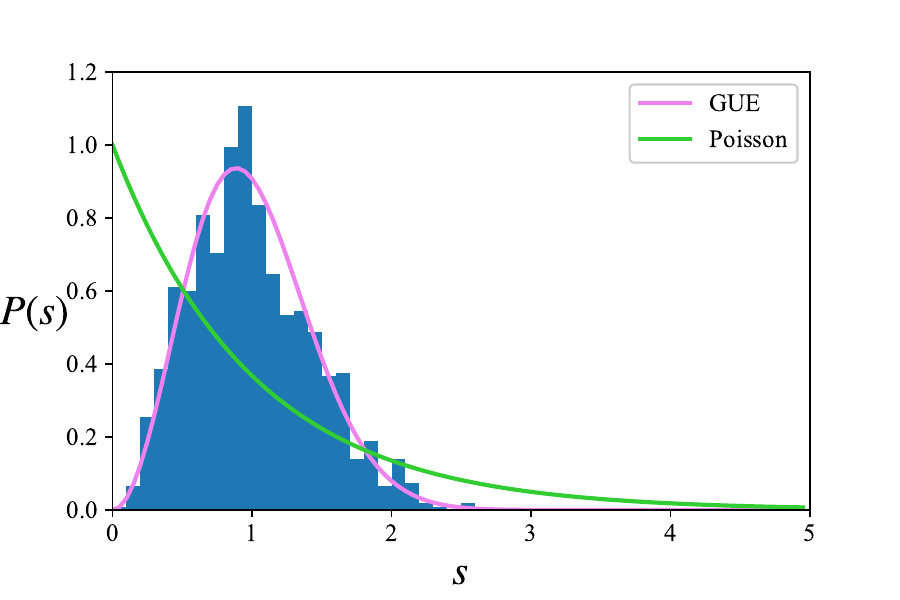}
\caption{Level-spacing statistics in the middle half of the spectrum of the inhomogeneous model ${\tilde H}(t)$ in Eq. (\ref{inhomogeneous AKLT+H3}) with $t=1$ and $L = 11$. Each $c_j$ is randomly chosen from $[-1, 1]$. 
The data are taken in the symmetry sector where $(\totS^z, \totS,{\cal F})=(0, 0, 1)$. 
The curves $P(s)_\mathrm{GUE}$ (magenta) and $P(s)_\mathrm{Poisson}$ (green) are shown for comparison. The distribution follows $P(s)_\mathrm{GUE}$. }
\label{lsp-randomHAKLT+H3}
\end{figure}

\subsubsection{Entanglement entropy}
In order to check whether the VBS state is a scar state, we compute entanglement entropies. The results are shown in Fig. \ref{EE_randomAKLT+H3}. Clearly, there are two entanglement outliers. The one at $E=0$ is the VBS state. The other one that is also far from other ordinary states is the ferromagnetic state. 
%The VBS state can be distinguished from other states. The other state that is 
Its energy is $2\sum c_j$ (see Appendix \ref{energy_of_ferromagnetic}), and it is a trivial state rather than a scar because of the SU(2) symmetry of the model (\ref{inhomogeneous AKLT+H3}).

\begin{figure}[hbtp]
\centering
\includegraphics[width=\linewidth]{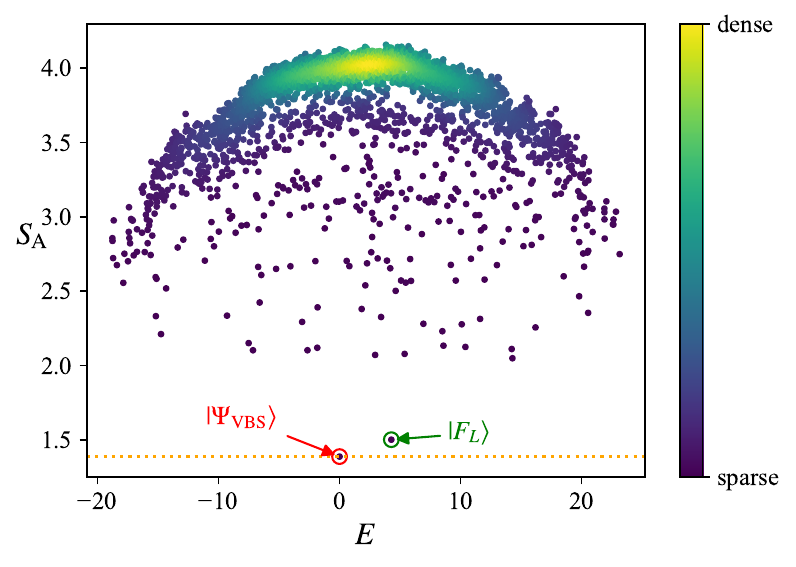}
\caption{Entanglement entropies in all eigenstates of the inhomogeneous model ${\tilde H}(t)$ in Eq. (\ref{inhomogeneous AKLT+H3}) for $L = 9$, $t= 1$ in the
$\totS^z = 0$ sector. {The density of data points is color coded.} 
Each $c_j$ is randomly chosen from $[-1, 1]$. The red and green circles indicate the VBS and the ferromagnetic states, respectively.}
\label{EE_randomAKLT+H3}
\end{figure}

\subsubsection{Other thermodynamic quantities} 
We provide further evidence that the VBS state is a scar state in this model. 
To this end, we study the expectation values of some physical observable in all energy eigenstates. The physical quantity we consider here is the AKLT Hamiltonian Eq. (\ref{HAKLT}) whose expectation value in a normalized state $\ket{\psi}$ is denoted as $\expval{H_\mathrm{AKLT}} \coloneqq \expval{H_\mathrm{AKLT}}{\psi}$.   
Figure \ref{expt_aklt} shows the numerical result for the distribution of $\expval{H_\mathrm{AKLT}}$. 
Clearly, the VBS state at $(0,0)$ can be distinguished from other eigenstates. 

\begin{figure}[tbp]
    \centering
    \includegraphics[width = \linewidth]{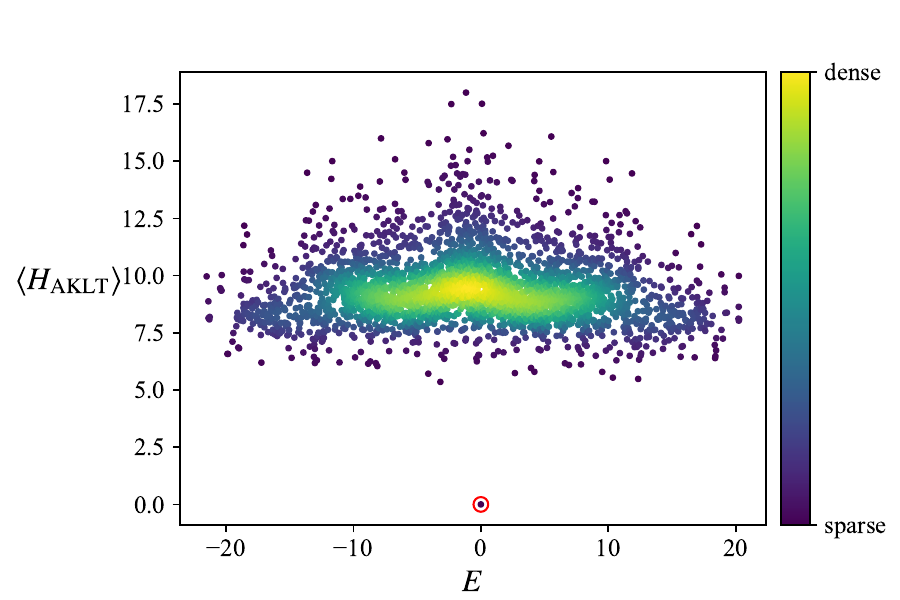}
    \caption{The expectation values of $H_\mathrm{AKLT}$ in all eigenstates of (\ref{inhomogeneous AKLT+H3}) with $t=1$, $L=9$ in the symmetry sector where $\totS^z = 0$. {The density of data points is color coded.} The red circle indicates the VBS state.}
    \label{expt_aklt}
\end{figure}

\subsection{Further generalizations}
Now we generalize the model in the previous subsection in two ways. First, we consider a generalization of the spin-$1$ AKLT model to include next-nearest-neighbor interactions. Based on the results obtained in~\cite{lange1994exact, nakano1996long}, we find that the VBS state in Eq. (\ref{eq:spin-1 VBS}) is annihilated by
\begin{align}
    \tilde{H}^{\prime}_{\mathrm{AKLT}} = \sum^L_{j=1} & d_j \biggl[
    {\bm S}_j \cdot {\bm S}_{j+1} + {\bm S}_{j+1} \cdot {\bm S}_{j+2} \nonumber \\
    & +\frac{1}{2} {\bm S}_j \cdot {\bm S}_{j+2} -\frac{1}{2} ({\bm S}_j \cdot {\bm S}_{j+2})^2 + 3
    \biggr],
\end{align}
where each coefficient $d_j$ can be any real number. This means that adding this term to the Hamiltonian $\tilde{H}(t)$ in Eq. (\ref{inhomogeneous AKLT+H3}) leaves the scar state unaffected. 
Second, we consider higher-order conserved charges $Q_{2k+1}$ ($k>1$) of the ${\rm SU}(3)$ Sutherland model, whose explicit expressions can be found in Ref.~\cite{GRABOWSKI1995299}. Since the VBS state is an integrable boundary state, it is annihilated by all $Q_{2k+1}$ (see Appendix \ref{H3VBS=0}). Thus, adding these terms with arbitrary coefficients does not affect the scar state. Combining these two generalizations leads to the following Hamiltonian
\begin{align}
   {\tilde H} (t_1, t_2, ..., t_n) = \tilde{H} (t_1) + \tilde{H}^\prime_{\mathrm{AKLT}} 
   + \sum^n_{k=2} t_k Q_{2k+1},
\end{align}
in which the VBS state survives as a scar state.

Let us finally discuss higher-spin generalizations. The spin-$1$ AKLT model can be generalized to models with ${\rm SO}(5)$ and more generally ${\rm SO}(2l+1)$ symmetry~\cite{scalapino1998so, frahm2001electronic, tu2008class}. The exact ground states of these models, which we dub ${\rm SO}(2l+1)$ VBS states, take the form of a matrix product state built from $2l+1$ gamma matrices. 
According to the general theory of integrable boundary states~\cite{pozsgay2019integrable}, the ${\rm SO}(2l+1)$ VBS state is an integrable boundary state of the ${\rm SU}(2l+1)$ Heisenberg model, meaning that the state is annihilated by all parity-odd conserved charges of the model. 
%combining the ${\rm SO}(2n+1)$ AKLT Hamiltonian and these parity-odd charges, one can construct an infinite family of models in which the ${\rm SO}(2n+1)$ VBS state is an exact zero-energy state. 
Thus, the construction of deformed models proceeds in much the same way as in the ${\rm SU}(2)$ case. We also note that the parent Hamiltonian of the ${\rm SO}(2l+1)$ VBS state can be inhomogeneous, like the one in Eq. (\ref{inhomogeneous AKLT}). We thus expect that the models constructed in this way are non-integrable for general $l$ and can be thought of as scarred models.

\section{Spin-1 AKLT model + scalar spin chirality}\label{sec:AKLT+CSC}

The model considered in this section is an example of a scarred model constructed by method (ii) with $n=0$ mentioned in Sec. \ref{subsec:method_of_construction}.

\subsection{Hamiltonian}
In this section, we consider another spin-$1$ model in which the VBS state in Eq. (\ref{eq:spin-1 VBS}) is a scar state. The Hamiltonian of the model is given by
\begin{equation}\label{AKLT+CSC}
H(t) = H_\mathrm{AKLT} + t C_\mathrm{SC},
\end{equation}
where $H_\mathrm{AKLT}$ is the AKLT Hamiltonian in Eq. (\ref{HAKLT}), and
\begin{equation}
C_\mathrm{SC} = \sum_{j=1}^L\bm{S}_j\cdot(\bm{S}_{j+1}\times\bm{S}_{j+2}),
\label{eq:spin-1 SC}
\end{equation}
is the scalar spin chirality term with ${\bm S}_j$ being the spin-1 operators in Eq. (\ref{spin1 operator}). 

This model has the same symmetries as $H(t)$ in Eq. (\ref{AKLT+H3}), i.e., SU(2) spin rotation, translation, spin-flip, and pseudo-time-reversal symmetries. (See Sec. \ref{sec_AKLT+H3} for details).  
Since the AKLT Hamiltonian is non-integrable, it is quite likely that the model Eq. (\ref{AKLT+CSC}) is not integrable either. This is indeed the case as can be seen from Fig. \ref{AKLT+CSC_levelstat}. 
Clearly, the level-spacing distribution is close to the GOE Wigner-Dyson distribution. 
To provide further evidence for this, we compute the $r$-value from the histogram and obtain $\langle r \rangle\simeq 0.536$, which agrees with $\langle r_\mathrm{GOE}\rangle \simeq 0.536$.

\begin{figure}
  \centering
  \includegraphics[width=\linewidth]{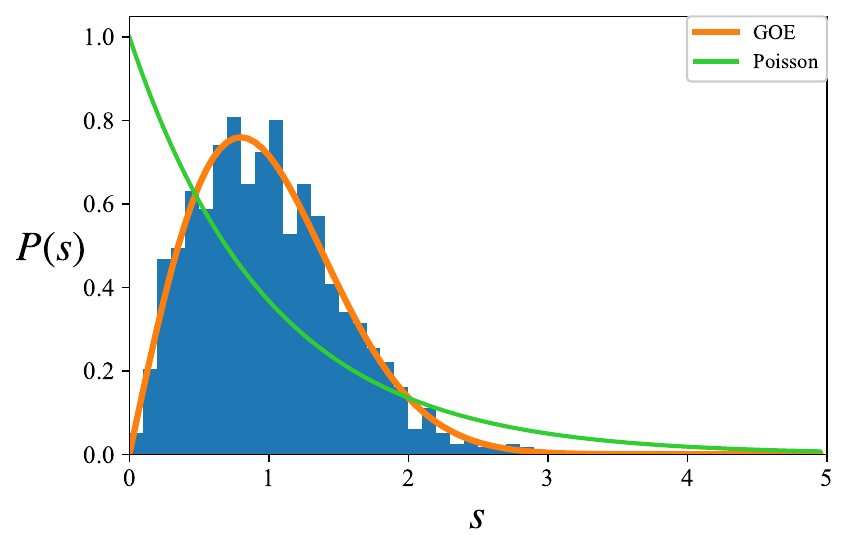}
  \caption{Level-spacing statistics in the middle half of the spectrum of the model (\ref{AKLT+CSC}) with $t=3$ and $L = 13$. 
  The data are taken in the symmetry sector where $(\totS^z, \totS,{\cal T}, {\cal F})=(0,0,1,1)$. 
  The curves $P(s)_\mathrm{GOE}$ (orange) and $P(s)_\mathrm{Poisson}$ (green) are shown for comparison. 
  The distribution follows $P(s)_\mathrm{GOE}$.}
  \label{AKLT+CSC_levelstat}
\end{figure}

\subsection{Scar state}
The VBS state $\ket{\Psi_\mathrm{VBS}}$ in Eq. (\ref{eq:spin-1 VBS}) is the zero-energy ground state of $H_\mathrm{AKLT}$. Interestingly, one can show that $\ket{\Psi_\mathrm{VBS}}$ is an eigenstate of $C_\mathrm{SC}$ with eigenvalue $0$ using its matrix product state representation (see Appendix \ref{appendix:CSCVBS=0.proof} for a proof).
Thus, $\ket{\Psi_\mathrm{VBS}}$ is a simultaneous eigenstate of $H_\mathrm{AKLT}$ and $C_\mathrm{SC}$, and is likely to be a scar state of the system. We checked it by computing half-chain entanglement entropies (Fig.~\ref{EE_AKLT+CSC}). The obtained results show that the VBS state has sufficiently low entanglement entropy compared to other states, confirming that it is indeed nonthermal. 

\begin{figure}
  \centering
  \includegraphics[width=\linewidth]{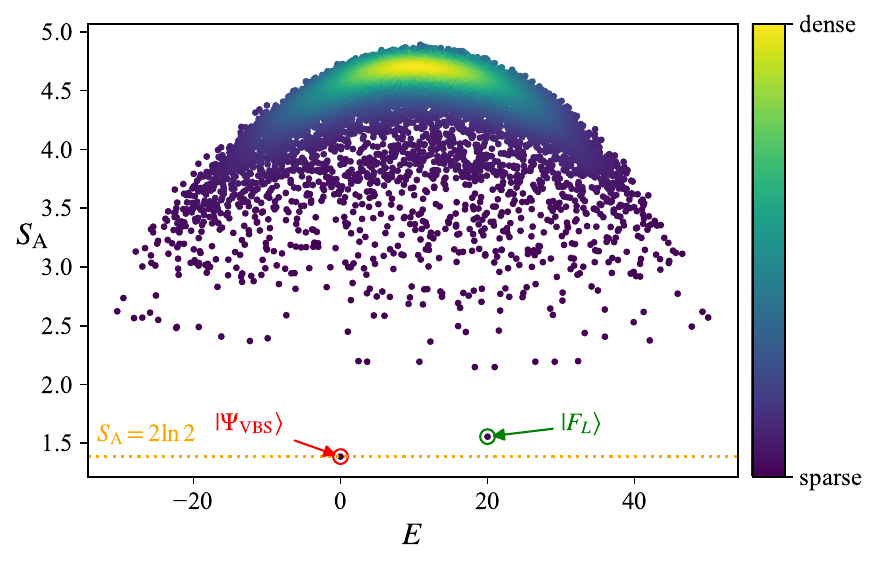}
  \caption{Entanglement entropies in all eigenstates of $H(t)$ in Eq. (\ref{AKLT+CSC}) with $t = 3$ for $L = 10$ in the $\totS^z = 0$ sector. {The density of data points is color coded.} The red and green circles indicate the VBS and ferromagnetic states, respectively. The orange dotted line indicates $S_\mathrm{A} = 2\ln 2\simeq 1.386$.}
  \label{EE_AKLT+CSC}
\end{figure}

\section{Perturbed spin-1 scalar spin chirality}\label{sec_CSC}
In this section, we consider a class of Hamiltonians consisting of the spin-1 scalar spin charity term $C_\mathrm{SC}$ in Eq. (\ref{eq:spin-1 SC}) and some other terms. They are examples of models constructed by method (ii) discussed in Sec. \ref{subsec:method_of_construction}. To provide some insight into what is special about this class of models, we have calculated the half-chain entanglement entropies in all eigenstates of the Hamiltonian 
\begin{equation}\label{SSC_only}
H_0 (h) = C_\mathrm{SC} + h \sum^L_{j=1} S^z_j,
\end{equation}
for $h=1$ and $L=8$. The results in Fig.~\ref{EE_SSC} indicate towers of low-entanglement states forming multiple arcs bridging $E = \pm 8$, but it is hard to distinguish them clearly because of degeneracies due to additional symmetries. In the following subsections, we classify these eigenstates and remove the degeneracies by introducing extra inhomogeneous terms. 

\begin{figure}
    \centering
    \includegraphics[width=\linewidth]{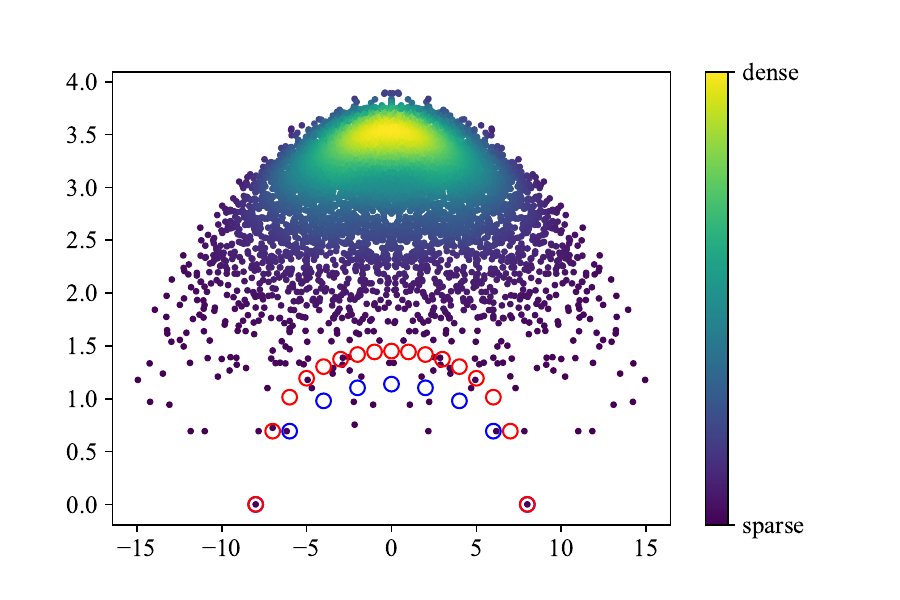}
    \caption{Entanglement entropies in all eigenstates of (\ref{SSC_only}) for $L = 8$, $h=1$. {The density of data points is color coded.} The red and blue circles indicate the positions of $\ket{{A}_{0, n}}$ in Eq. (\ref{eq:Amn}) and $\ket{B_{0, n}}$ in Eq. (\ref{eq:Bmn}), respectively.}
    \label{EE_SSC}
\end{figure}

\subsection{zero-energy states of \texorpdfstring{$C_\mathrm{SC}$}{CSC}}
Interestingly, one can explicitly construct some zero-energy states of $C_\mathrm{SC}$ by acting with certain ladder operators on reference states $\ket{\Uparrow} \coloneqq \ket{++\cdots +}$ and $\ket{\mathbf{0}} \coloneqq \ket{00\cdots 0} $, where $\ket{+}, \ket{0}$, and $\ket{-}$ are eigenstates of $S^z$ with eigenvalues $+1$, $0$, and $-1$, respectively. 
We define the ladder operators by
\begin{equation}
    \mathcal{O}_p^- = \sum_{j=1}^L e^{\mathrm{i}pj}S_j^-,\quad  \mathcal{Q}_p^- = \sum_{j=1}^L e^{\mathrm{i}pj}(S_j^-)^2, 
\end{equation}
where $S_j^- = S_j^x - \mathrm{i}S_j^y$. The subscript $p$ indicates that the operator carries momentum $p$, which takes the values $p= \frac{2\pi n}{L}$ with $n=0, 1, ..., L-1$. 
We find that the following states are zero-energy states of $C_\mathrm{SC}$:
\begin{align}
    \ket{A_{m, n}} &\coloneqq (\mathcal{O}_0^-)^m(\mathcal{O}_\pi^-)^n\ket{\Uparrow}\quad(0\leq m+n \leq 2L), \label{eq:Amn}\\
    \ket{B_{m, n}} &\coloneqq (\mathcal{O}_0^-)^m (\mathcal{Q}_0^-)^n\ket{\Uparrow} \quad(n\geq1, 1\leq m\leq 2L-2n) \label{eq:Bmn}. 
\end{align}
Note that $\ket{A_{m, n}}$ are well defined only for even $L$. To see that the above states are zero-energy states of $C_\mathrm{SC}$, it suffices to consider the case $m=0$. This is because the operator $\mathcal{O}_0^-$ is exactly the spin-lowering operator $\totS^- = \sum_{j=1}^L S^-_j$ commuting with $C_\mathrm{SC}$ due to the $\mathrm{SU}(2)$ symmetry. For convenience, we introduce the notation $\ket{\bar{A}_n} \coloneqq \ket{A_{0, n}}$, $\ket{\bar{B}_n} \coloneqq \ket{B_{0, n}}$ to denote the above states with $m=0$. One can prove that $\ket{\bar{A}_n}$ and $\ket{\bar{B}_n}$ are zero-energy eigenstates of $C_\mathrm{SC}$ by noting that $C_\mathrm{SC}$ and either $\mathcal{O}_\pi^-$ or $\mathcal{Q}_0^-$ satisfy a restricted spectrum generating algebra of order $2$ \cite{moudgalya2020eta}. See Appendix \ref{appendix:AB_proof} for the proof. 

The towers of states $\ket{A_{m, n}}$ and $\ket{B_{m, n}}$ do not exhaust the zero-energy manifold of $C_\mathrm{SC}$. In fact, there are other towers of zero-energy states generated by $\mathcal{O}_p^-$: 
\begin{align}
    \ket{+_m} &\coloneqq (\mathcal{O}_0^+)^m\ket{\mathbf{0}}, \label{eq:plus}\\
    \ket{-_m} &\coloneqq (\mathcal{O}_0^-)^m\ket{\mathbf{0}}, \label{eq:minus}\\
    \ket{K_{m, p}} &\coloneqq  (\mathcal{O}_0^-)^m \mathcal{O}_p^-\mathcal{O}_{-p}^-\ket{\Uparrow},
    \label{eq:K}
\end{align}
where $0 \leq m\leq 2L-2$ and $p = \frac{2\pi}{L}, \frac{4\pi}{L}, \cdots, \frac{2\pi}{L} ( \lfloor \frac{L}{2} \rfloor -1 )$. 
Again since $\mathcal{O}_0^-$ commutes with $C_\mathrm{SC}$, it suffices to consider the case $m=0$. It is easy to see that $\ket{+_0}=\ket{-_0}=\ket{\mathbf{0}}$ is annihilated by each local term in $C_\mathrm{SC}$, and hence $C_\mathrm{SC} \ket{\pm_m} = 0$. To see that $C_\mathrm{SC} \ket{K_{m, p}}=0$, it is convenient to rewrite the state $\ket{K_{0, p}}$ as 
\begin{align}
    \ket{K_{0,p}} = \sum^L_{n=1} e^{{\rm i}n p} \ket{\Phi_n}
\end{align}
where
\begin{align}\label{eq:K_to_phi}
    \ket{\Phi_n} = \sum_{j=1}^L S^-_jS^-_{j+n}\ket{\Uparrow}. 
\end{align}
One can show that each $\ket{\Phi_n}$ is annihilated by $C_\mathrm{SC}$. Therefore, it follows that $C_\mathrm{SC} \ket{K_{m, p}}=0$. See Appendix \ref{appendix:AB_proof} for a detailed proof.

\medskip

In this way, we have constructed a number of exact zero-energy states of $C_\mathrm{SC}$. It should be noted that they are exact eigenstates of $H_0 (h) = C_\mathrm{SC}+h {\cal S}^z$ as well because each of them is a superposition of states with fixed ${\cal S}^z$. We also remark that the obtained states in Eqs. (\ref{eq:Amn})-(\ref{eq:K}) are not orthogonal to each other. In fact, they are not even linearly independent.  
This can be seen by considering, for example, the $L=3$ site chain. In this case, $\ket{B_{1,1}}$, $\ket{B_{3,0}}$, and $\ket{+_0}$ satisfy $3 \ket{B_{1,1}}- \ket{B_{3,0}} +12 \sqrt{2} \ket{+_0}=0$, and hence linearly dependent. In Appendix \ref{appendix:AB_proof}, we derive a lower bound on the number of zero-energy states of $C_\mathrm{SC}$, which proves that the number grows exponentially with the system size. Such an exponentially large degeneracy can be a source of QMBS and Hilbert space fragmentation, as discussed in the context of geometrically frustrated systems~\cite{lee2020exact, lee2021frustration}.

In the following, we will consider $H_0 (h)$ in Eq.(\ref{SSC_only}) under tailored disorder, which is designed such that some of the obtained zero-energy states of $C_\mathrm{SC}$ remain intact. 

\subsection{Random single-ion anisotropy --- scarred \texorpdfstring{$\ket{\bar{B}_n}$}{Lg}}\label{subsec:B_n}
In this subsection, we focus on the model in which $\ket{\bar{B}_n}=(\mathcal{Q}_0^-)^n\ket{\Uparrow}, (n = 1, 2, \ldots, L)$ become scars. We consider the Hamiltonian
\begin{equation}\label{SSC+S2}
    H_1(h, \{D_j\}_j) = C_\mathrm{SC} + h\sum_{j=1}^L S^z_j + \sum_{j=1}^L D_j(S^z_j)^2, 
\end{equation}
where $D_j$ are any real numbers. In what follows, we omit the dependence of $H_1$ on $h$ and $\{D_j\}_j$ unless necessary. 

Since $H_1$ commutes with ${\cal S}^z$, one can split the Hilbert space into subspaces labeled by the eigenvalues of ${\cal S}^z$. The ${\cal S}^z = 0$ subspace is special in that it is invariant under spin flip ${\cal F}$. Thus, this subspace can be further divided into two sectors: one with ${\cal F}=1$ and the other with ${\cal F}=-1$. We have analyzed the level-spacing statistics in the sector $({\cal S}^z, {\cal F}) = (0, 1)$ and found that the distribution is close to the GUE Wigner-Dyson distribution. We also calculated the $r$-value and obtained $\langle r \rangle \simeq 0.593$, which is consistent with the GUE. 

\subsubsection{Tower of eigenstates}
The states $\ket{\bar{B}_n}$ constitute a tower of eigenstates of $H_1$. This can be seen as follows. In the previous subsection, we have already shown that each $\ket{\bar{B}_n}$ is a simultaneous eigenstate of $C_\mathrm{SC}$ and ${\cal S}^z$. Thus it remains to show that these states are eigenstates of the third term on the RHS of Eq. (\ref{SSC+S2}), which we call the $D$ term. 
To show this, we take a closer look at $\ket{\bar{B}_n}$. In the basis of $S^z_j$ eigenstates, they read
\begin{align}
    \ket{\bar{B}_0} &= \ket{\Uparrow} =\ket{++\cdots+}, \\
    &\qquad\vdots \nonumber \\
    \ket{\bar{B}_n} &= (\mathcal{Q}_0^-)^n \ket{\Uparrow} = 2^n\times n!\sum_{1\leq j_1<j_2<\cdots<j_n\leq L}\nonumber \\
    &\quad\ket{++\cdots-_{j_1}\cdots-_{j_2} \cdots-_{j_n}\cdots +}, \\
    &\qquad\vdots \nonumber \\
    \ket{\bar{B}_L} &= 2^L \times L! \ket{--\cdots-}.
\end{align} 
As one can see, each $\ket{\bar{B}_n}$ consists of sequences of $\ket{\pm}$, in which the state $\ket{0}$ never appears. Therefore, each $\ket{\bar{B}_n}$ is an eigenstate of $(S_j^z)^2$ with eigenvalue $1$ for all $j$, implying that $\ket{\bar{B}_n}$ is an eigenstate of the $D$ term with eigenvalue $\sum_j D_j$.

One can calculate the half-chain entanglement entropy of $\ket{\bar{B}_n}$ in the same way as the ferromagnetic states with spin-$1/2$. The result reads  
\begin{equation}\label{eq:SA_Bn}
    S_A(\ket{\bar{B}_n}) = -\sum_{k=0}^{n}\frac{\binom{L/2}{k}\binom{L/2}{n-k}}{ \binom{L}{n}}\ln \frac{\binom{L/2}{k}\binom{L/2}{n-k}}{\binom{L}{n}}.
\end{equation}
(see Appendix \ref{appendix:ferro_ent_half} for details). The state $\ket{\bar{B}_{L/2}}$ has the largest entanglement entropy in $\{\ket{\bar{B}_n}\}_{n = 0, 1, \ldots, L}$, and the asymptotic form of $S_A (\ket{\bar{B}_{L/2}})$ for $L \gg 1$ is given by
\begin{equation}
    S_A(\ket{\bar{B}_{L/2}}) \approx \frac{1}{2}\left(\ln \frac{\pi L}{8}+1\right),
\end{equation}
which obeys a sub-volume law. 

Figure \ref{EE_SSC+S2_full} shows the half-chain entanglement entropy as a function of energy for $H_1$ with $L = 8$, $h = 1$, and $D_j$ randomly chosen from $[-1,1]$. As one can see, the states $\ket{\Bar{B}_n}$ form a tower of low-entanglement states. They are, however, not well separated from other states due to the presence of other low-entanglement states. 
The obtained result also suggests that the energy $E=0$ is highly degenerate even in the presence of disorder in $D_j$. Since we obtain superpositions of degenerate eigenstates in numerical diagonalization, the entanglement entropy of the state $\ket{\Bar{B}_{L/2}}$ at this energy is obscured. This may be the reason why a data point is missing in the red circle at $E=0$.

In order to resolve the degeneracy, we now consider the sector with fixed quantum numbers $(\mathcal{S}^z, \mathcal{F})=(0,1)$. Figure \ref{EE_SSC+S2_Sz0} shows the entanglement entropies of the eigenstates of $H_1$ in this symmetry sector. Clearly, the state $\ket{\bar{B}_{L/2}}$ is isolated from the other states, indicating its nonthermal nature. 
In addition to this state, there are two other entropy outliers: $\ket{\psi_s}=\ket{+-+-\cdots}+\ket{-+-+\cdots}$ and $\ket{\mathbf{0}}$. These states are zero-energy states of $C_\mathrm{SC}$ that remain intact under the influence of the $D$ term. However, $\ket{\mathbf{0}}$ cannot be thought of as a scar. This is because the projection onto this state, namely $\mathcal{P}=\prod_{j=1}^L(1-(S_j^z)^2)$, commutes with the Hamiltonian $H_1$, which simply means that $\ket{\mathbf{0}}$ is the state that is uniquely specified by the eigenvalue $1$ of $\mathcal{P}$. 
Thus, out of the three entropy outliers, only $\ket{\bar{B}_{L/2}}$ and $\ket{\psi_s}$ are identified as QMBS. 
Figure \ref{fig:h1_aklt} shows the expectation values of a local observable for all eigenstates of $H_1$ in the sector with $(\mathcal{S}^z, \mathcal{F})=(0,1)$. As an observable, we consider $H_\mathrm{AKLT}$ in Eq. (\ref{HAKLT}). As we can see, the expectation values in these three states are isolated from other thermal states, which implies a violation of strong ETH. 

\begin{figure}
    \centering
    \includegraphics[width = \linewidth]{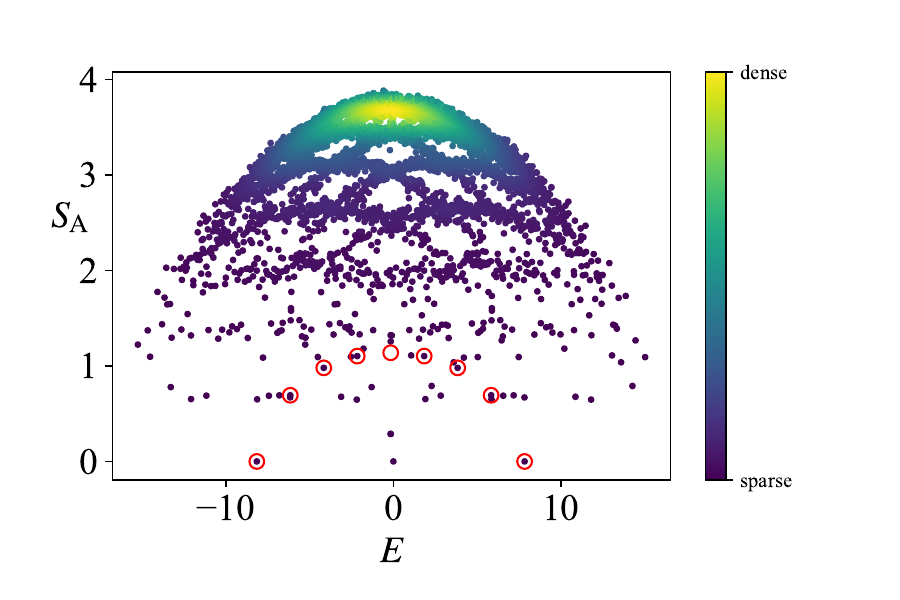}
    \caption{Entanglement entropies in all eigenstates of $H_1(h, \{D_j\})$ in Eq. (\ref{SSC+S2}) for $L = 8$, $h = 1$, $D_j\in [-1, 1]$. {The density of data points is color coded.} The red circle corresponds to $\ket{\bar{B}_n}$ $(n = 0, 1, \ldots, L)$.}
    \label{EE_SSC+S2_full}
\end{figure}
\begin{figure}[tbh]
    \centering
    \includegraphics[width=\linewidth]{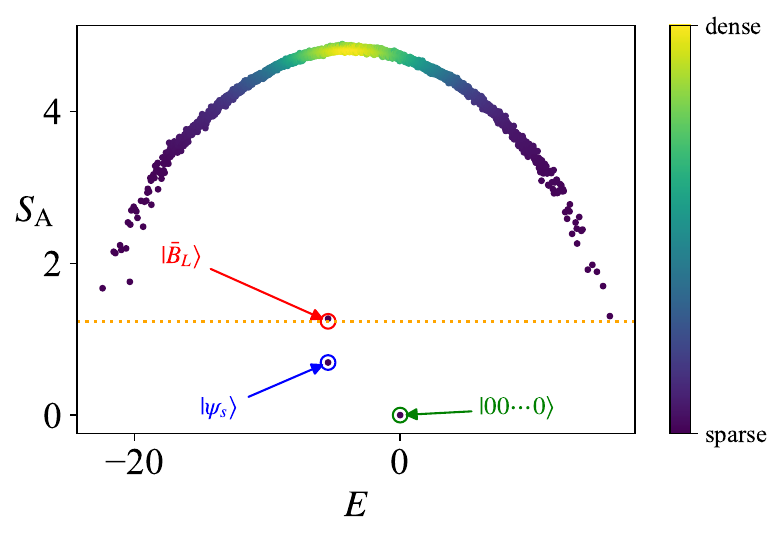}
    \caption{Entanglement entropies in all eigenstates of $H_1(h, \{D_j\})$ in Eq.~(\ref{SSC+S2}) for $L = 10$, $h = 1$ in the symmetry sector $(\totS^z, \mathcal{F}) = (0, 1)$. Each $D_j$ is randomly chosen from $[-5, 5]$. {The density of data points is color coded.} The entanglement entropy of each state obeys a volume-law except for $\ket{\bar{B}_{L/2}}$, the symmetric state $\ket{\psi_s} = \ket{+-+-\cdots} + \ket{-+-+\cdots}$, and the trivial state $\ket{\mathbf{0}}=\ket{00\cdots 0}$. The orange dotted line indicates $S_A = 1.236$, which is obtained from Eq. (\ref{eq:SA_Bn}) for $L = 10$. }
    \label{EE_SSC+S2_Sz0}
\end{figure}
\begin{figure}[tbh]
    \centering
    \includegraphics[width=\linewidth]{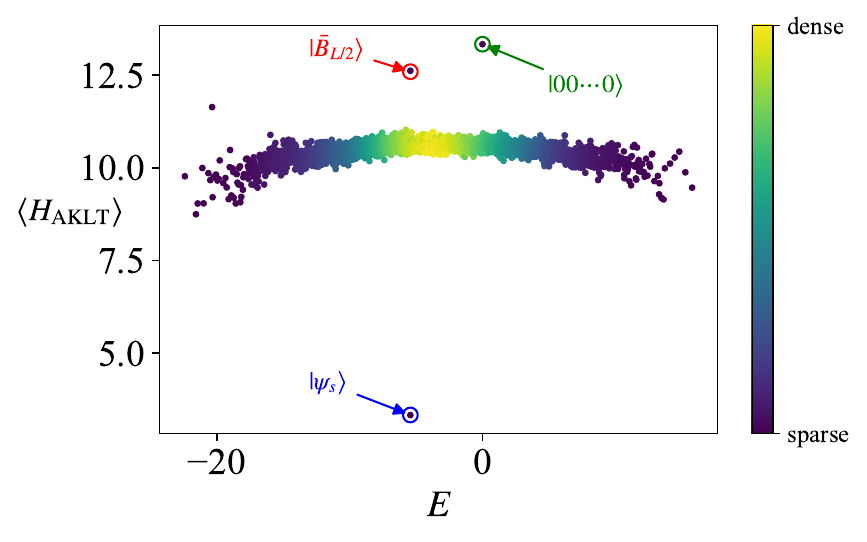}
    \caption{$\expval{H_\mathrm{AKLT}}$ in all eigenstates of $H_1(h, \{D_j\})$ in Eq.~(\ref{SSC+S2}) for $L = 10$, $h = 1$ in the symmetry sector $(\totS^z, \mathcal{F}) = (0, 1)$. Each $D_j$ is randomly chosen from $[-5, 5]$. {The density of data points is color coded.} The red, blue, and green circles indicate $\ket{\bar{B}_{L/2}}$, $\ket{\psi_s}$, and $\ket{00\cdots0}$, respectively.}
    \label{fig:h1_aklt}
\end{figure}

\subsubsection{Dynamics}
%The dynamics describes the characteristic of the scars more precisely. 
To illustrate the nonthermal features of scarred states, we study the quench dynamics of the system. 
The initial states we consider are coherent states of $\mathcal{Q}_0^-$, namely superpositions of $\ket{\bar{B}_n}$ defined as
\begin{equation}
    \ket{\beta}\coloneqq C_L^{-1}\exp(\beta\mathcal{Q}_0^-)\ket{\Uparrow} = C_{L}^{-1}\sum_{n=0}^L\frac{\beta^n}{n!}\ket{\bar{B}_n},
\end{equation}
where $\beta \in \mathbb{C}$ and $C_L \coloneqq (1+4\abs{\beta}^2)^{\frac{L}{2}}$ is a normalization factor such that $\langle \beta|\beta\rangle=1$. 
Under time evolution by the Hamiltonian (\ref{SSC+S2}), the initial state $\ket{\beta}$ evolves into
\begin{equation}
    \ket{\beta(t)} \coloneqq e^{-\mathrm{i}H_1 t}\ket{\beta}
\end{equation}
at time $t$. Since the states $\ket{\bar{B}_n}$ are common eigenstates of the $D$ term with eigenvalue ${\cal D} =\sum_j D_j$, we can rewrite it as
\begin{align}
    \ket{\beta(t)} = e^{-\mathrm{i} \mathcal{D} t} e^{-\mathrm{i}h \mathcal{S}^z t}\ket{\beta}.
\end{align}

We first consider the fidelity between initial and time-evolved states. For an arbitrary initial state $\ket{\phi(0)}$, it is defined by
\begin{equation}\label{eq:fidelity}
    \mathcal{F}(t) = \abs{\braket{\phi(0)}{\phi(t)}},
\end{equation}
where $\ket{\phi(t)} = e^{-\mathrm{i}H_1t}\ket{\phi(0)}$. For the coherent states $\ket{\beta}$, we can calculate the fidelity as
\begin{align}\label{eq:fidelityB}
    \mathcal{F}(t) &= C_L^{-2}\abs{\sum_{m, n=0}^L \frac{(\beta^*)^m\beta^n}{m!n!}\mel{\bar{B}_m}{e^{-\mathrm{i}h \mathcal{S}^z t}}{\bar{B}_n}} \nonumber\\
    &= C_L^{-2} \abs{\sum_{n=0}^L\frac{|\beta|^{2n}}{(n!)^2}e^{2\mathrm{i}hnt}\braket{\bar{B}_n}{\bar{B}_n}} \nonumber\\
    &= C_L^{-2} \abs{\sum_{n=0}^L(4|\beta|^2e^{2\mathrm{i}ht})^n\binom{L}{n}} \nonumber\\
    &= \abs{\frac{1+4|\beta|^2e^{2\mathrm{i}ht}}{1+4\abs{\beta}^2}}^L.
\end{align}
Clearly, it is a periodic function with period $T=\pi/h$, exhibiting perfect revivals, i.e., $\mathcal{F}(t)=1$ at $t=nT$ ($n \in \mathbb{N}$), irrespective of the system size. We show in Fig. \ref{fig:F(t)_b} the numerical results of the fidelity dynamics with several initial states. As we can see, the coherent states show perfectly periodic revivals, indicating that they never thermalize. This is in stark contrast to the fidelity of a generic state, which decays rapidly to zero. We remark that a perfect revival of the initial state after a time of at most ${\cal O} ({\rm poly}(L))$, in general, implies the existence of QMBS~\cite{alhambra2020revivals}.  

\begin{figure}[tbp]
    \centering
    \includegraphics[width = \linewidth]{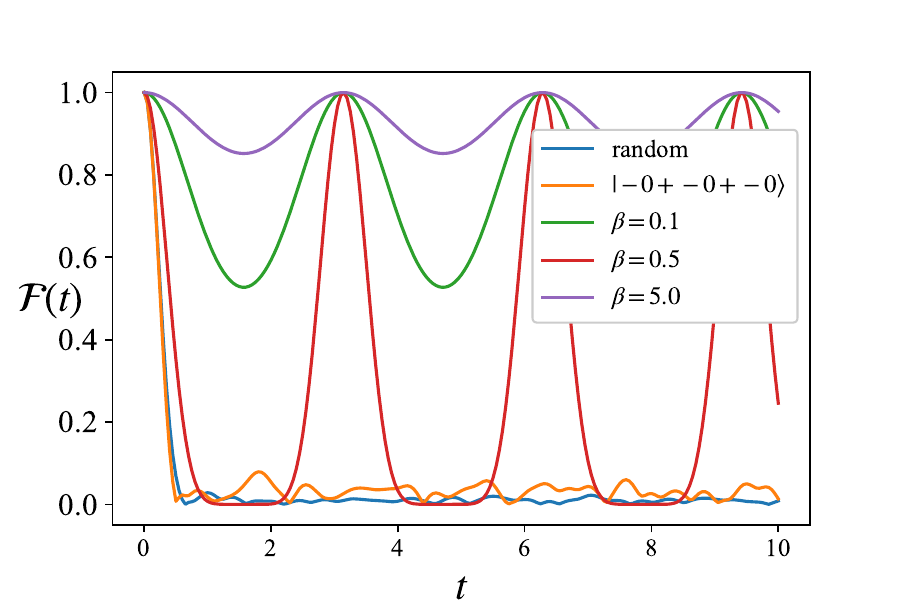}
    \caption{The dynamics of the fidelity with $L=8$, $h = 1$, and $D_j$ $(j = 1, 2, \ldots, L)$ chosen randomly from $[-1, 1]$. Perfectly periodic revivals can be seen when the initial state is a coherent state, whereas for other generic states the fidelity decays rapidly to zero.}
    \label{fig:F(t)_b}
\end{figure}

We next examine the time evolution of the half-chain entanglement entropy for several initial states. 
In the following, we consider the case of even $L$. The half-chain entanglement entropy of the coherent state $\ket{\beta(t)}$ does not evolve in time. In fact, it is always $0$. This can be seen by noting that $\ket{\beta(t)}$ is just a product of two states:
\begin{align}
    \ket{\beta(t)} & = e^{-\mathrm{i} \mathcal{D} t} e^{-\mathrm{i}h \mathcal{S}_A^z t} \ket{\beta}_A \otimes e^{-\mathrm{i}h \mathcal{S}_B^z t} \ket{\beta}_B, 
\end{align}
where $\mathcal{S}^z_{A (B)} = \sum_{j \in A (B)} S^z_j$ and 
\begin{align}
    \ket{\beta}_{A (B)} := C_{L/2}^{-1} \exp (\beta \sum_{j \in A (B)} (S^-_j)^2) \ket{\Uparrow}_{A(B)}
\end{align}
with $\ket{\Uparrow}_{A(B)}=\otimes_{j \in A (B)}\ket{+}_j$.  
Let us slightly generalize the initial state by considering a superposition of coherent states with different $\beta$, namely, $\sum_{i=1}^n c_i\ket{\beta_i}$. 
Since $\ket{\beta_i}_L = \ket{\beta_i}_A \otimes \ket{\beta_i}_B$, the entanglement entropy is obtained as~\cite{shi2006classical, katsura2010entanglement}
\begin{equation}
    S_A = -\frac{\Tr [M^2\ln M^2] }{\Tr [M^2] }+\ln{\Tr [M^2]},
\end{equation}
where the matrix elements of $M$ are defined as
\begin{equation}
    M_{i,j} = c_i^*c_j\braket{\beta_i}{\beta_j}_A = \frac{c_i^*c_j(1+4\beta_i^*\beta_j)^N}{(1+4\abs{\beta_i}^2)^{\frac{N}{2}}(1+4\abs{\beta_j}^2)^\frac{N}{2}}. 
\end{equation}
It is then clear that the entanglement entropy for this class of states is constant in time.

Figure \ref{fig:ent_dynamics_b} shows the time evolution of the half-chain entanglement entropy $S_A$ for several initial states. Clearly, the coherent states and their superposition do not gain entanglement. By contrast, $S_A$ of the product state $\ket{-0+-0+-0}$ grows rapidly and saturates near the Page value \cite{page1993average} of a random state
\begin{equation}\label{eq:pagevalue}
    S_\mathrm{Page} = \frac{L}{2}\ln 3 - \frac{1}{2}.
\end{equation}
\begin{figure}
    \centering
    \includegraphics[width=\linewidth]{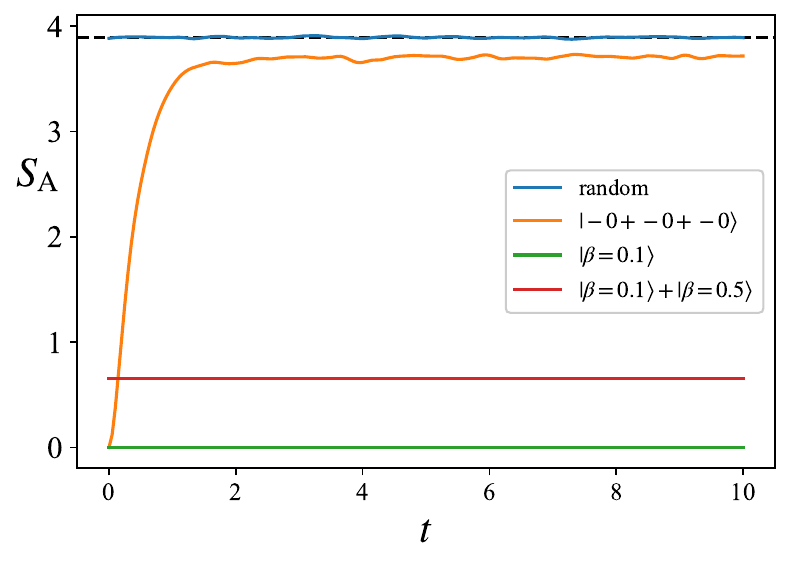}
    \caption{Dynamics of the half-chain entanglement entropies with the same setup as Fig. \ref{fig:F(t)_b}. The dashed line indicates the Page value $S_\mathrm{page}$ (Eq. (\ref{eq:pagevalue})). Initial coherent states have constant entanglement entropy, but that of $\ket{-0+-0+-0}$ rapidly grows and saturates near $S_\mathrm{Page}$.}
    \label{fig:ent_dynamics_b}
\end{figure}

\subsection{Scarred \texorpdfstring{$\ket{\bar{A}_n}$}{Lg}}
\subsubsection{Tower of eigenstates}
As we have seen in the previous subsection, the key to finding a suitable perturbation is to find an operator that acts on a set of target states as a constant. To find such an operator for $\ket{\bar{A}_n}$, let us take a closer look at these states. 
The operator $\mathcal{O}_\pi^-$ that generates $\ket{\bar{A}_n}$ is invariant under translation by two sites. Therefore, considering its action on the two neighboring sites may suggest a suitable %additional 
operator. The operator $\mathcal{O}_\pi^-$ acts as $S^-_j - S^-_{j+1}$ at the two neighboring sites $(j, j+1)$, and as a result, we get the states listed in Table \ref{A_table} by repeatedly applying it to the state $\ket{++}_{j, j+1}$.
\begin{table}[tbhp]
    \centering
    \begin{tabular}{c|c|c|c}
        $n$ & state & $S^z_j+S^z_{j+1}$ & $P_{j, j+1}$  \\
        \hline
        $0$ & $\ket{++}$ & $2$ & $1$ \\
        $1$ & $\ket{0+}-\ket{+0}$ & $1$ & $-1$ \\
        $2$ & $\ket{-+}-2\ket{00}+\ket{+-}$ & $0$ & $1$ \\
        $3$ & $\ket{0-}-\ket{-0}$ & $-1$ & $-1$ \\
        $4$ & $\ket{--}$ & $-2$ & $1$
    \end{tabular}
    \caption{The explicit form of $(S_j^--S_{j+1}^-)^n\ket{++}_{j, j+1}$ up to constant factors. The third and fourth columns indicate the eigenvalues of the corresponding operators for each state.}
    \label{A_table}
\end{table}

Each state in the table is a simultaneous eigenstate of $S_j^z+S_{j+1}^z$ and $P_{j, j+1}$, the permutation operator between site $j$ and $j+1$ (see Eq. (\ref{perm})). From this result, we see that $(-1)^{S_j^z+S_{j+1}^z}P_{j, j+1}=1$ holds in the subspace spanned by these states.  
%is a conserved quantity. 
Therefore, each $\ket{\bar{A}_n}$ is an eigenstate of the following Hamiltonian:
\begin{equation}\label{perm_ham}
\begin{split}
    H_2(h, \{D_j\}_j) = & C_\mathrm{SC} + h\sum_{j=1}^LS_j^z  \\
    & + \sum_{j=1}^L D_j(-1)^{S_j^z+S_{j+1}^z}P_{j, j+1},
\end{split}
\end{equation}
where $D_j$ are any real numbers. %is a constant randomly selected and 
{One can, in principle, construct a more complicated Hamiltonian involving more than two-spin interactions using the same strategy.}

In what follows, we assume that the number of sites $L$ is even and omit the dependence of $H_2$ on $h$ and ${D_j}$ unless necessary. Interestingly, the states $\ket{\bar{B}_n}$ are also eigenstates of $H_2$ since $\ket{\bar{B}_n}$ is totally symmetric, i.e., $P_{i, j}\ket{\bar{B}_n} = \ket{\bar{B}_n}$ for any $i, j$ and $S_j^z+S_{j+1}^z$ is even for any $j$. 
We can see from Fig. \ref{fig:perm_full} that the states $\ket{\bar{A}_n}$ behave as QMBS in this system. On the other hand, the data points for the states $\ket{\bar{B}_n}$ are mostly missing due to the degeneracies. 

Since the last term in Eq. (\ref{perm_ham}) does not break U$(1)$ symmetry associated with $\totS^z$, we can divide the Hilbert space into subspaces according to the eigenvalues of $\totS^z$. 
The $\totS^z = 0$ subspace can be further decomposed into two sectors with opposite ${\cal F}$. We have analyzed the level-spacing statistics in the sector $({\cal S}^z, {\cal F}) = (0, 1)$ and found that the distribution is close to the GUE Wigner-Dyson distribution. We also calculated the $r$-value and obtained $\langle r \rangle \simeq 0.593$, which is consistent with the GUE. 

\begin{figure*}[tb]
\captionsetup[subfloat]{}
    \begin{center}
    %\captionsetup[subfigure]{singlelinecheck=false,skip=0pt,position=top}
    \subfloat[][]{
        %\caption{}
        \includegraphics[width=0.45\linewidth]{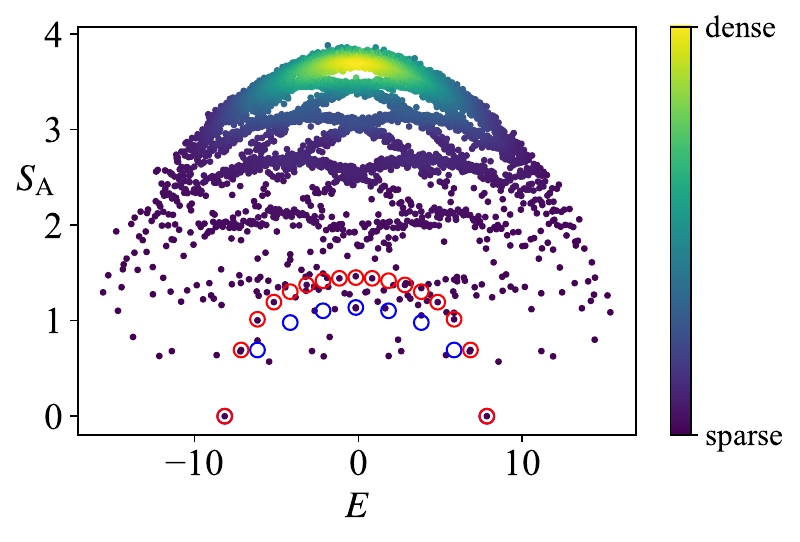}
        \label{fig:perm_full}
        }
    \quad
    \subfloat[][]{
        %\caption{}
        \includegraphics[width=0.45\linewidth]{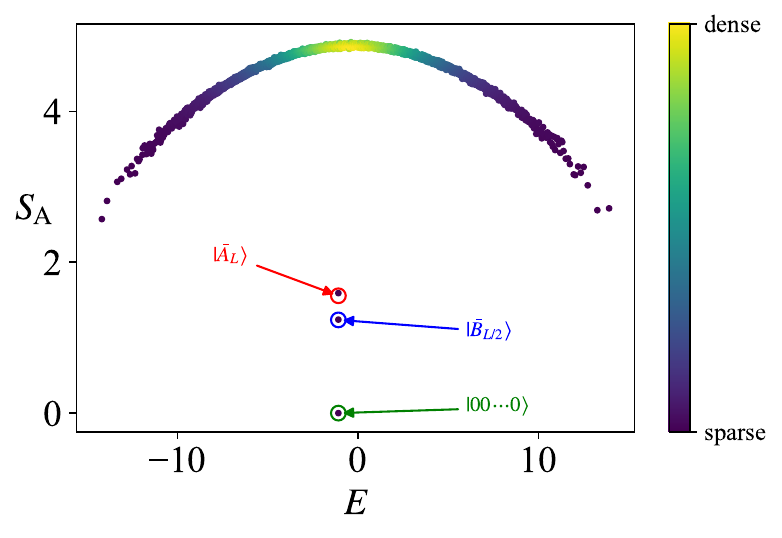}
        \label{fig:perm_0}
    }
    \end{center}
    \vspace*{-4mm}
    \caption{(a) Entanglement entropies in all eigenstates in Eq. (\ref{perm_ham}) with $L=10, h=1$. Each $D_j$ is randomly chosen from $[-1, 1]$. {The density of data points is color coded.} 
    The states $\ket{\bar{A}_n}$ (red circles) and $\ket{\bar{B}_n}$ (blue circles) have relatively low entanglement entropy. Note that although $\ket{\bar{A}_n}$ and $\ket{\bar{B}_n}$ are eigenstates of $H_2(h, \{D_j\}_j)$, some circles do not overlap with corresponding data points because of degeneracies. 
    (b) Entanglement entropies in all eigenstates of inhomogeneous model Eq. (\ref{perm_ham}) for $L = 10, h=1$ in the symmetry sector $\totS^z = 0$. Each $D_j$ is randomly chosen from $[-1, 1]$. {The density of data points is color coded.} 
    The red, blue and green circles indicate $\ket{\bar{A}_L}$, $\ket{\bar{B}_{L/2}}$, and $\ket{000\cdots}$, respectively.}
    \label{fig:perm}
\end{figure*}

Figure \ref{fig:perm_0} shows the entanglement entropies of eigenstates of $H_2$ in the sector $\totS^z = 0$. In the figure, both $\ket{\bar{A}_L}$ and $\ket{\bar{B}_{L/2}}$ can be identified as entanglement outliers, which leads us to the conclusion that $\ket{\bar{B}_n}$ are also QMBS for $H_2$.
It should be noted that another entanglement outlier, i.e., 
$\ket{\mathbf{0}}=\ket{000\cdots}$ in Fig. \ref{fig:perm_0}, cannot be thought of as a scar. This is because the projection onto this state, namely $\mathcal{P}=\prod_{j=1}^L(1-(S_j^z)^2)$, commutes with $H_2$, which simply means that this state is uniquely specified by the eigenvalue $1$ of $\mathcal{P}$. 

\subsubsection{Dynamics}
Similarly to the previous subsection, we introduce a coherent state of ${\cal O}_{\pi}$, namely a superposition of $\ket{\bar{A}_n}$ defined by
\begin{equation}
    \ket{\alpha} = \tilde{C}^{-1}_L \exp (\alpha\mathcal{O}_\pi^-) \ket{\Uparrow} = \tilde{C}^{-1}_L \sum_{n=0}^{2L} \frac{\alpha^n}{n!}\ket{\bar{A}_n},
\end{equation}
where $\tilde{C}_L = (1+\abs{\alpha}^2)^L$ is the normalization constant. 
%Then, we can define the fidelity in the same way as in Section \ref{subsec:B_n} (Eq. (\ref{eq:fidelity})). 
When the initial state is the coherent state $\ket{\alpha}$, the fidelity defined in Eq. (\ref{eq:fidelity})) can be computed as
\begin{align}\label{eq:fidelityA}
    \mathcal{F}(t) &= \tilde{C}^{-2}_L\abs{\sum_{m,n=0}^{2L}\frac{\alpha^{*m}\alpha^n}{m!n!}\mel{\bar{A}_m}{e^{-\mathrm{i}H_2t}}{\bar{A}_n} } \nonumber \\
    &= \tilde{C}^{-2}_L\abs{\sum_{n=0}^{2L}\frac{\abs{\alpha}^{2n}}{(n!)^2}e^{\mathrm{i}hnt}\braket{\bar{A}_n}} \nonumber \\
    &= \tilde{C}^{-2}_L\abs{\sum_{n=0}^{2L}(\abs{\alpha}^2e^{\mathrm{i}ht})^n\binom{2L}{n}} \nonumber \\
    &= \abs{\frac{1+\abs{\alpha}^2e^{\mathrm{i}ht}}{1+\abs{\alpha}^2}}^{2L}.
\end{align}
Thus, it attains the maximum fidelity $\mathcal{F}(t) = 1$ periodically with period $T = 2\pi/ h$. Figure \ref{fig:fidelity_A} shows the numerical results of the fidelity dynamics with several initial states. Here we set $h=1$. 
Clearly, the coherent states attain $\mathcal{F}(t) = 1$ periodically with period $2\pi$, whereas the fidelities of the other states decay rapidly to zero. 
We also calculated the time evolution of the entanglement entropy for these states and obtained a result similar to that shown in Fig. \ref{fig:ent_dynamics_b}. See Appendix \ref{appendix:dynamics} for the dynamics from a more complex initial state. 

\begin{figure}[bp]
    \centering
    \includegraphics[width = \linewidth]{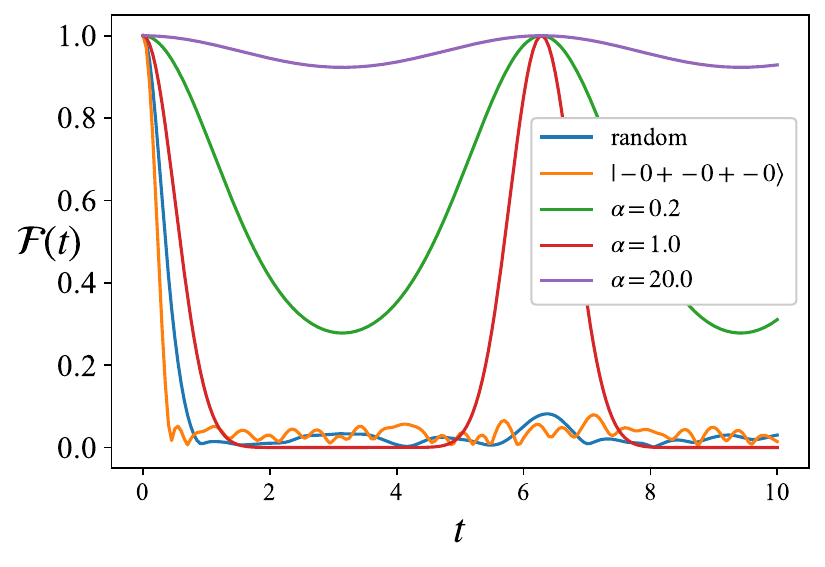}
    \caption{The dynamics of the fidelity with $h = 1$, $L = 8$, and $D_j$ $(j = 1, 2, \ldots, L)$ chosen randomly from $[-1, 1]$. The fidelity shows perfectly periodic revivals when the initial state is a coherent state, whereas it decays rapidly to zero for other generic states.}
    \label{fig:fidelity_A}
\end{figure}

\subsection{Two-dimensional model}
In the same way as before, we can construct two-dimensional models with QMBS. Here we consider the generalization of the model $H_1$ in Eq. (\ref{SSC+S2}) on a triangular lattice with periodic boundary conditions (Fig. \ref{fig:pic_of_triangle}). Let $\Lambda$ be the triangular lattice. The Hamiltonian is 
\begin{equation}\label{eq:2-dim}
    H^{\rm 2d}_1 (h, \{D_j\}_j) =  C^{\rm 2d}_\mathrm{SC} + h\sum_{j \in \Lambda} S^z_j+\sum_{j \in \Lambda} D_j(S^z_j)^2,
\end{equation}
where
\begin{equation}
    C^{\rm 2d}_\mathrm{SC} = \sum_{\triangle / \bigtriangledown} \bm{S}_j\cdot(\bm{S}_k\times\bm{S}_l), \label{eq_triangularCSC}
\end{equation}
and the summation is over all triangles. The subscripts $j, k$, and $l$ are in the clockwise (counterclockwise) order in each upward (downward) triangle.
\begin{figure}[pb]
    \centering
    \includegraphics[width=0.65\linewidth]{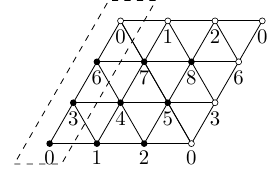}
    \caption{An example of the triangular lattice. %with periodic boundary conditions. 
    The integers denote the site indices, and the black and white sites with the same index are identified by the periodic boundary conditions. %the sites represented by white dot are identified as the corresponding black sites. 
    We take the subsystem $A$ to be the set of sites enclosed by the dashed lines, which is used to calculate the entanglement entropy}
    \label{fig:pic_of_triangle}
\end{figure}

Analogously to the states $\ket{\bar{B}_n}$, one can define the following states:
\begin{equation}
    \ket{\Psi_n} = (\mathcal{Q}_0^-)^n\ket{\Uparrow} = \left(\sum_{j \in \Lambda} (S_j^-)^2\right)^n\ket{\Uparrow}.
\end{equation}
Since $\mathcal{Q}_0^-$ can be decomposed into $\mathcal{Q}_0^- = Q_X + Q_{\Lambda \backslash X}$ with $Q_X := \sum_{i\in X}(S^-_i)^2$ and $Q_{\Lambda \backslash X} := \mathcal{Q}_0^- - Q_X$ for any $X=\{ j, k, l \}$ forming an upward or downward triangle, we obtain
\begin{align}
    \ket{\Psi_n} &= \sum_{p=0}^n\binom{n}{p} (Q_X)^p\ket{+++}_{X}\nonumber\\ &\qquad\otimes (Q_{\Lambda \backslash X})^{n-p}\ket{++\cdots+}_{\Lambda \backslash X}.
\end{align}
Then, one can show that $\bm{S}_j\cdot(\bm{S}_k\times\bm{S}_l)(Q_X)^p\ket{+++}_{X}=0$ for any $X$ and $p$ in the same way as in the one-dimensional case. Therefore, $\ket{\Psi_n}$ is an eigenstate of $C^{\rm 2d}_\mathrm{SC}$ with eigenvalue $0$. Furthermore, since $D_j(S_j)^2$ acts only on one site, $D_j(S_j)^2\ket{\Psi_n} = D_j\ket{\Psi_n}$ can be shown in the same way as in the one-dimensional case. Therefore, each $\ket{\Psi_n}$ is an eigenstate of $H^{\rm 2d}_1$ in Eq. (\ref{eq:2-dim}) with eigenvalue $h(|\Lambda|-2n)+\sum_{j=1}^L D_j$, where by $|\Lambda|$ we denote the total number of sites in $\Lambda$.

To see whether $\ket{\Psi_n}$ are QMBS, we calculate entanglement entropies %$S_A$ 
in all eigenstates of $H^{\rm 2d}_1$ in the symmetry sector with ${\cal S}^z=1$. In the calculation, we take the subsystem $A$ to be $\{0, 3, 6\}$ (see Fig. \ref{fig:pic_of_triangle} for the site labels). As shown in Fig. \ref{fig:ent_triangular}, the state $\ket{\Psi_{n=4}}$ has significantly lower entanglement entropy than the other states, indicating that this state is a scar state. 
We have also checked numerically that the entanglement entropy of $\ket{\Psi_n}$ is extremely low regardless of the choice of the subsystem $A$. 

\begin{figure}[pt]
    \centering
    \includegraphics[width = \linewidth]{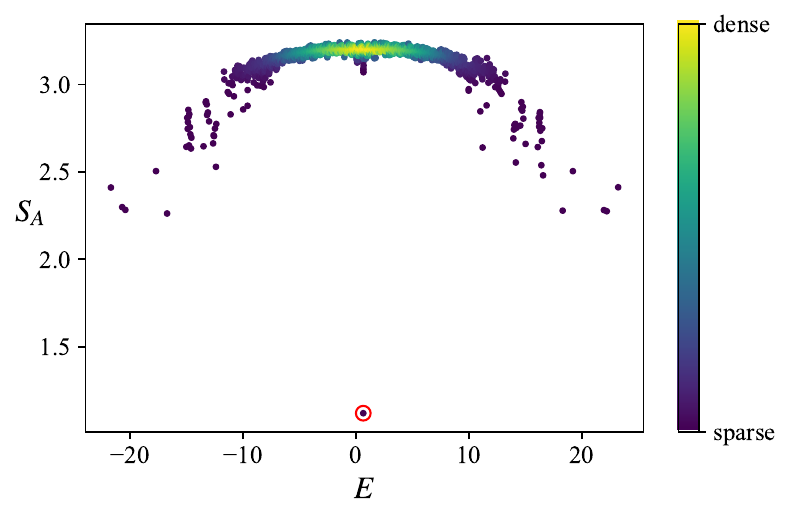}
    \caption{Entanglement entropies in all eigenstates of $H^{\rm 2d}_1 (h, \{D_j \})$ in Eq. (\ref{eq:2-dim}) for $L = 9, h = 1$ in the symmetry sector $\totS^z = 1$. Each $D_j$ is randomly chosen from $[-1, 1]$. {The density of data points is color coded.} The red circle indicate the scar state $\ket{\Psi_{n=4}}$.}
    \label{fig:ent_triangular}
\end{figure}

\section{Discussion}\label{sec_discussion}
We have constructed several examples of quantum spin models with two- and three-body interactions that exhibit QMBS using the two different methods: one based on integrable boundary states and the other focusing on towers of states generated by the raising and lowering operators. 
We demonstrated that the QMBS in the models behave differently from thermal states by comparing their spectral, dynamical, and entanglement properties with those of other typical states. 
The methods presented in this work can be used to systematically construct other models with QMBS. For example, using another combination of conserved charges $\{Q_{2k+1}\}_{k=1, 2, \ldots}$ in the first method, one can construct a family of new models in which an integrable boundary state is an exact zero-energy state. The second method discussed in Sec. \ref{sec:method} is a variant of existing methods, and in particular, it is similar to that of Tang {\it et al}.~\cite{tang2021multi}, where the authors constructed models with towers of scar states generated by irreducible tensor operators.  
However, their main focus is on spin-$1/2$ systems, whereas our work is primarily concerned with spin-$1$ systems.
Moreover, we demonstrated that our method allows for the construction of a two-dimensional model with QMBS. 

In future research, it would be interesting to construct non-integrable models that have multiple integrable boundary states as QMBS. To this end, we need to find a non-integrable Hamiltonian in which these integrable boundary states become degenerate. 
Since the idea of integrable boundary states can be traced back to those of integrable quantum field theories~\cite{ghoshal1994boundary}, it would also be interesting to construct non-integrable quantum field theories with QMBS by extending our method. In this regard, we note that QMBS in continuous models have also been discussed in previous studies~\cite{schindler2022exact, martin2022scar, liska2022holographic, cotler2022}. 

{
Another direction worth investigating is to apply our methods to open quantum and periodically driven systems. Previous studies have shown that there are some open or driven systems that fail to thermalize at late times~\cite{iemini2018boundary, buvca2019non, dutta2021out, tindall2021analytical, buvca2022algebraic}, where algebraic approaches were also widely used. Finally, it would also be interesting to extend the notion of integrable boundary states to such systems and construct models exhibiting nonthermalizing dynamics. In this respect, dissipative systems described by integrable Lindblad superoperators~\cite{ziolkowska2020yang, de2021constructing, de2023hidden} and integrable Floquet systems~\cite{gritsev2017integrable, vanicat2018integrable, lotkov2022floquet} may serve as a good starting point for constructing concrete examples.}

%As another application of our method, we can consider constructing open or driven systems with quantum many-body scars. Some previous research reports that there exist states that show extraordinary thermalization in open quantum systems~\cite{iemini2018boundary, dutta2021out} or driven systems~\cite{tindall2021analytical,dutta2021out}. Since we can construct an integrable conserved charge from an integrable Lindblad superoperator~\cite{de2021constructing}, we can apply our idea of constructing QMBS from integrable boundary states to open systems. Similarly, for driven systems, we can consider QMBS emerged from integrable Floquet dynamics~\cite{gritsev2017integrable, lotkov2022floquet}.

%Last but not least, with rapid advances in experimental techniques in cold atom systems, it would not be entirely unrealistic to expect the realization of somewhat artificial models discussed in Sec.~\ref{sec_CSC}, such as the one consisting solely of scalar spin chirality terms. It will thus be interesting to speculate on the consequences of QMBS found in this work in experimental settings. In this respect, the fate of exact QMBS under perturbations is also worth further investigation~\cite{Turner2, lin2020slow, gotta2023asymptotic}.  

\section*{Acknowledgement}
We thank Eric Vernier for valuable discussions and for allowing us to include his proof of Theorem \ref{thm:f1}, and Kensuke Tamura for useful discussions. The numerical calculations of the entanglement entropy and the level spacing statistics were performed using QuSpin~\cite{quspin1, quspin2}. 
H. K. was supported by MEXT KAKENHI Grant-in Aid for Transformative Research Areas A ``Extreme Universe" No. JP21H05191.
H. K. was also supported in part by JSPS KAKENHI Grant No. JP18K03445, JP23H01093, and the Inamori Foundation. K. S. acknowledges the support of the Forefront Physics and Mathematics Program to Drive Transformation. 
Y. M. acknowledges the support from the GGI BOOST fellowship.

\newpage
\appendix % A P P E N D I X
\begin{widetext}
\section{Entanglement entropy of the ferromagnetic states}\label{appendix:ferro_ent_half}
We consider the general spin-$\sigma$ case ($\sigma = 1/2, 1, 3/2, \ldots$). The basis states $\ket{\zeta}$ $(\zeta = \sigma, \sigma-1, \ldots, -\sigma+1, -\sigma)$ on each site are defined such that $S^z\ket{\zeta} = \zeta\ket{\zeta}$ and $S^\pm \ket{\zeta} = \sqrt{\sigma (\sigma+1) -\zeta (\zeta \pm 1)} \ket{\zeta \pm 1}$. We denote the fully polarized state, i.e., $\ket{\sigma}^{\otimes L}$ by $\ket{\Uparrow}_L$. 
Then the ferromagnetic states as the $\mathrm{SU}(2)$ descendants of $\ket{\Uparrow}_L$ are defined as
\begin{align}
\ket{F_n}_L =\frac{1}{\sqrt{\mathcal{N}_\sigma(n, L)}} (\totS^-)^n\ket{\Uparrow}_L, 
\quad n=0,1,...,2 \sigma L
\end{align}
where $\mathcal{N}_\sigma(n, L)$ is a normalization constant such that $\braket{F_n}_L = 1$. One can show that
\begin{equation}\label{eq:norm_s1}
    \mathcal{N}_\sigma(n, L) = \frac{n!(2\sigma L)!}{(2\sigma L-n)!}
\end{equation}
for all $\sigma$, $L$, and $n$ using %the commutation relation 
$[\totS^+, \totS^-] = 2\totS^z$ and the mathematical induction on $n$. The proof goes as follows. 

\medskip

First, it is obvious that $\norm{(\totS^-)^0\ket{\Uparrow}_L}^2 = \norm{\ket{\Uparrow}_L}^2 = 1$. On the other hand, we obtain $\frac{0!2\sigma L!}{(2\sigma L-0)!} = 1$. Thus, Eq. (\ref{eq:norm_s1}) is valid for $n = 0$. Next, we assume Eq. (\ref{eq:norm_s1}) is true for $n=k$. 
Then we can calculate $\mathcal{N}_\sigma(k+1, L)$ as
\begin{align}
    \mathcal{N}_\sigma(k+1, L) &= \expval{(\totS^+)^{k+1}(\totS^-)^{k+1}}{\Uparrow}_L \\
    &= \expval{(\totS^+)^k(\totS^- \totS^+ +2\totS^z)(S^-)^k}{\Uparrow}_L \\
    &= 2(\sigma L-k)\mathcal{N}_\sigma(k, L) + \expval{(\totS^+)^k \totS^-(\totS^- \totS^+ +2\totS^z)(\totS^-)^{k-1}}{\Uparrow}_L \\
    &= \cdots\nonumber\\
    &= 2\{(\sigma L-k)+(\sigma L-k+1)+\cdots+\sigma L\}\mathcal{N}_\sigma(k, L) \\
    %&= (2\sigma L-k)(k+1) \frac{k!(2\sigma L)!}{(2\sigma L-k)!} \\
    &= \frac{(k+1)!(2\sigma L)!}{(2\sigma L-(k+1))!}
\end{align}
Thus, Eq. (\ref{eq:norm_s1}) is true for $n=k+1$. Hence Eq. (\ref{eq:norm_s1}) is true for all $n = 0, 1, \ldots, 2\sigma L$.

\medskip

To compute the entanglement entropy, we divide the whole chain into two subsystems, $A$ and $B$, with lengths $L_A$ and $L_B$, respectively. Since $\totS^- = \sum_{j=1}^LS_j^-$ is the sum of single-site operators, we can rewrite $\totS^-$ as $\totS^- = \totS^-_A+\totS_B^-$, where $\totS^-_A = \sum_{j\in A}S^-_j$ and $\totS^-_B = \sum_{j\in B}S^-_j$. %We assume the subsystem $A$ as the left half of the chain (Fig. \ref{fig_1D_sublattice}). 
Then, we see that $\ket{F_n}_L$ can be written in the Schmidt decomposition form:
\begin{align}\label{eq:FnEE}
    \ket{F_n}_L &= \frac{1}{\sqrt{\mathcal{N}_\sigma(n, L)}}(\totS_A^- + \totS_B^-)^n\ket{\Uparrow}_L \\
    &= \sum_{k=0}^n\frac{1}{\sqrt{\mathcal{N}_\sigma(n, L)}}\binom{n}{k}(\totS_A^-)^k (\totS_B^-)^{n-k}\ket{\Uparrow}_L \\
    &= \sum_{k=0}^n\sqrt{\frac{\mathcal{N}_\sigma(k, L_A)\mathcal{N}_\sigma(n-k, L_B) }{\mathcal{N}_\sigma(n, L)}}\binom{n}{k}\ket{F_k}_{L_A}\otimes \ket{F_{n-k}}_{L_B},
\end{align}
from which the entanglement entropy is obtained as
\begin{align}
    S_A(\ket{F_n}_L) &= -\sum_{k=0}^{n}\frac{\mathcal{N}_\sigma(k, L_A)\mathcal{N}_\sigma(n-k, L_B) }{\mathcal{N}_\sigma(n, L)}\binom{n}{k}^2 \ln \frac{\mathcal{N}_\sigma(k, L_A)\mathcal{N}_\sigma(n-k, L_B) }{\mathcal{N}_\sigma(n, L)}\binom{n}{k}^2\\
    &= -\sum_{k=0}^{n} \binom{2\sigma L}{n}^{-1}\binom{2\sigma L_A}{k}\binom{2\sigma L_B}{n-k} \ln \binom{2\sigma L}{n}^{-1}\binom{2\sigma L_A}{k}\binom{2\sigma L_B}{n-k}.
\end{align}

We now discuss the asymptotic behavior of $S_A(\ket{F_n}_L)$ for large $L$. In particular, we consider the case where $L$ is even, $L_A=L_B=L/2$, and $n = \sigma L$. 
In this case, one can rewrite Eq. (\ref{eq:FnEE}) as
\begin{align}
    S_A(\ket{F_{\sigma L}}_{L}) &= -\sum_{k=0}^{\sigma L} \binom{2\sigma L}{\sigma L}^{-1}\binom{\sigma L}{k}^2 \ln{\binom{2\sigma L}{\sigma L}^{-1}\binom{\sigma L}{k}^2} \\
    &= \binom{2\sigma L}{\sigma L}^{-1}\sum_{k=0}^{\sigma L}\left[\binom{\sigma L}{k}^2\ln \binom{2\sigma L}{\sigma L} - \binom{\sigma L}{k}^2 \ln \binom{\sigma L}{k}^2\right]. \label{eq:ent_ferro}
\end{align}
Using the Vandermonde identity, we can rewrite the first term in the bracket as
\begin{equation}\label{a16}
    \sum_{k=0}^{\sigma L}\binom{\sigma L}{k}^2\ln\binom{2\sigma L}{\sigma L}=\binom{2\sigma L}{\sigma L}\ln\binom{2\sigma L}{\sigma L}.
\end{equation}
On the other hand, we cannot directly compute the second term for finite $L$. However, in the large-$L$ limit, the dominant contribution to the summation comes from large $k$. Thus, we can apply the de Moivre-Laplace theorem and evaluate $\binom{n}{k}$ as
\begin{equation}
    \binom{n}{k} = \frac{2^n}{\sqrt{\frac{1}{2}n\pi}}\exp\left[-\frac{2(k-\frac{n}{2})^2}{n}\right]\times\left(1+\order{\frac{1}{\sqrt{n}}}\right)
\end{equation}
in the limit $n\to \infty$~\cite{ELIOT1990109, uspensky1937introduction, popkov2005logarithmic, diener2005higher}. 
Using this asymptotic expansion and the Vandermonde identity, we have
\begin{equation}\label{eq:ferro_asymp1}
    \sum_{k=0}^{\sigma L}\binom{\sigma L}{k}^2\ln\binom{\sigma L}{k}^2 = 
    \binom{2\sigma L}{\sigma L} \left[
    2 \sigma L\ln2 -2 \ln\sqrt{\frac{\pi\sigma L}{2}} - \frac{\sigma L}{2\sigma L -1} +\order{\frac{1}{\sqrt{L}}}
    \right].
\end{equation}
Substituting this and Eq. (\ref{a16}) into Eq. (\ref{eq:ent_ferro}), we find
\begin{align}
    S_A(\ket{F_{\sigma L}}_{L}) = \ln \binom{2\sigma L}{\sigma L} + \frac{\sigma L}{2\sigma L-1} - 2\left[\sigma L\ln2 - \ln \sqrt{\frac{\pi\sigma L}{2}}\right] + \order{\frac{1}{\sqrt{L}}}.
\end{align}
Finally, using Stirling's formula, $\ln n! = n\ln n - n + \frac{1}{2}\ln n + \frac{1}{2}\ln 2\pi+ \order{n^{-1}}$, we get
%to evaluate the order of its magnitude. We obtain
\begin{equation}
    \ln\binom{2\sigma L}{\sigma L} = 2\sigma L\ln 2 -\frac{1}{2} \ln \sigma L -\frac{1}{2}\ln \pi + \order{\frac{1}{L}},
\end{equation}
and hence we obtain the asymptotic form of Eq. (\ref{eq:ent_ferro}) as
\begin{equation}
    S_A(\ket{F_{\sigma L}}_L) = \frac{1}{2}\ln \sigma L + \frac{1}{2}\left(\ln \frac{\pi}{4}+1\right)+\order{\frac{1}{\sqrt{L}}}\qquad (L \gg 1).
\end{equation}

\section{SU(2) symmetry of \texorpdfstring{$H_3$}{Lg}}\label{appendix:ferromagnetic}
In this section, we provide a proof that $H_3$ in (\ref{H3}) has global ${\rm SU}(2)$ symmetry. 
As in the main text, we define the total spin operators by $\totS^\alpha = \sum_{j=1}^L S^\alpha_j$ ($\alpha= x, y, z$), where $S^\alpha_j$ is the spin-1 operator at site $j$ in Eq.~\eqref{spin1 operator}. They are the generators of the global ${\rm SU}(2)$, i.e., for any $U\in {\rm SU}(2)$ there exists $\{ \theta_\alpha \}$ such that $U = \exp(i\sum_{\alpha\in\{x, y, z\}}\theta_\alpha \totS^\alpha)$. Any $U$ can be decomposed into $U = 1 + X$, where $X$ is a polynomial in $\{\totS^x, \totS^y, \totS^z\}$. 
Thus, to prove that $[H_3, U]=0$ for all $U \in {\rm SU}(2)$, it suffices to show $[H_3, \totS^\alpha] = 0$ for all $\alpha\in\{x, y, z\}$. 

For simplicity, we introduce the SU(3) generators $T^a\coloneqq \lambda^a/2$ satisfying $[T^a, T^b] = \mathrm{i}f_{abc}T^c$, in terms of which $H_3/8$ is written as
\begin{equation}
    \frac{1}{8}H_3 = \sum_{j = 1}^Lf_{abc}T^a_jT^b_{j+1}T^c_{j+2},
\end{equation}
where the summation over repeated indices $a$, $b$, and $c$ is implied. 
Due to the tracelessness of the SU (2) generators, $S_j^x, S_j^y$, and $S_j^z$ can be written as linear combinations of the SU(3) generators. By $\Tr [T^a T^b]=\frac{1}{2}\delta_{ab}$~\cite{georgi2000lie}, we have $S^\alpha_j = 2\Tr_j [ S^\alpha_j T^u_j ] T^u_j$. Now we calculate the commutator 
\begin{align}
    \frac{1}{8}\left[H_3, \sum_{k=1}^LT^u_k\right] &= \left[\sum_{j=1}^L f_{abc} T^a_jT^b_{j+1}T^c_{j+2}, \sum_{k=1}^L T^u_k\right] \\
    &= \sum_{j=1}^Lf_{abc} \left( [T^a_j, T^u_j]T^b_{j+1}T^c_{j+2}+T^a_j[T^b_{j+1}, T^u_{j+1}]T^c_{j+2}+T^a_jT^b_{j+1}[T^c_{j+2}, T^u_{j+2}] \right) \\
    &= \mathrm{i} \sum_{j=1}^Lf_{abc} \left( f_{auv}T^v_jT^b_{j+1}T^c_{j+2}+f_{buv}T^a_jT^v_{j+1}T^c_{j+2}+f_{cuv}T^a_jT^b_{j+1}T^v_{j+2} \right) \\
    &= \mathrm{i} \sum_{j=1}^L (\underbrace{f_{vbc}f_{vua}+f_{avc}f_{vub}+f_{abv}f_{vuc}}_{(*)})\, T^a_jT^b_{j+1}T^c_{j+2}.
\end{align}
Then it follows from the Jacobi identity that $(*) = 0$, which yields the desired result $[H_3, \totS^\alpha] = 0$ ($\alpha= x, y, z$), i.e., the global SU(2) symmetry of $H_3$.

\section{Proof that \texorpdfstring{$\ket{\Psi_{\mathrm{VBS}}}$}{VBS} is an integrable boundary state}\label{H3VBS=0}

In this section, we prove that the %AKLT ground 
VBS state $\ket{\Psi_{\mathrm{VBS}}}$ in Eq. (\ref{eq:spin-1 VBS}) is an integrable boundary state of the Sutherland model. The Hamiltonian of the model is defined as
\begin{equation}
    H_\mathrm{S} = - \sum_{i=1}^L \big( P_{i,i+1} -1 \big) ,
\end{equation}
commuting with $H_3$ in Sec. \ref{sec_AKLT+H3}. The quantum integrability of the model can be summarized by the transfer matrix $T(\lambda)$,
\begin{equation}
    T(\lambda) = \mathrm{Tr}_a \left( \prod_{j=1}^L R_{a,j} (\lambda) \right) ,
    \label{eq:Tlambda}
\end{equation}
with the R matrix 
\begin{equation}
    R_{a,j} (\lambda) = \frac{1}{\lambda + \text{i} } \big( \lambda + \text{i} P_{a,j} \big)
    \label{eq:RmatSutherland}
\end{equation}
satisfying the celebrated Yang--Baxter equation. The sub-index $a$ stands for the $3$-dimensional auxiliary space, which is traced over in Eq. \eqref{eq:Tlambda}, resulting in an operator acting only on the physical Hilbert space. From the Yang--Baxter equation, it is easy to deduce that $T(\lambda)$ is in involution,
\begin{equation}
    [T(\lambda) , T(\mu) ] = 0 , \quad \forall \lambda, \mu \in \mathbb{C} .
\end{equation}
Moreover, we have
\begin{equation}
    H_\mathrm{S} = - \text{i} \left. \frac{\partial}{\partial \lambda} \log T(\lambda) \right|_{\lambda=0} , \quad H_3 = -2\text{i} \left. \frac{\partial^2}{\partial \lambda^2} \log T(\lambda) \right|_{\lambda=0}  .
\end{equation}

As shown in \cite{Piroli2017WhatIA}, for an integrable boundary state $|\Psi_0 \rangle$ with even system size $L$, 
\begin{equation}
    Q_{2k+1} |\Psi_0 \rangle = 0 \quad (k=1,2,...) \quad \Leftrightarrow \quad T(\lambda) |\Psi_0 \rangle = \mathcal{I}\, T(\lambda)\, \mathcal{I}\, |\Psi_0 \rangle , 
    \label{eq:intboundarycond}
\end{equation}
where the parity (spatial inversion) operator
\begin{equation}
    \mathcal{I}  = \prod_{j=1}^{L/2} P_{j,L-j+1} .
\end{equation}
Our aim is to show that the VBS state $\ket{\Psi_{\mathrm{VBS}}}$ satisfies the condition \eqref{eq:intboundarycond} with the Sutherland model transfer matrix. To begin with, for any similarity transformation with local density
\begin{equation}
    U = \prod^L_{j=1} u_j ,
\end{equation}
the transfer matrix $T(\lambda)$ commutes with it, i.e.,
\begin{equation}
    U T(\lambda) U^{-1} = \mathrm{Tr}_a \left( \prod_{j=1}^L u_j R_{a,j} (\lambda) u_j^{-1} \right) = \mathrm{Tr}_a \left( \prod_{j=1}^L u_a^{-1} R_{a,j} (\lambda) u_a \right) = T(\lambda) ,
\end{equation}
using \eqref{eq:RmatSutherland}. We choose the similarity transformation to be
\begin{equation}
u_j = \frac{1}{\sqrt{2}}\begin{pmatrix}
1 & \mathrm{i} & 0 \\
0 & 0 & -\sqrt{2} \\
-1 & \mathrm{i} & 0
\end{pmatrix}_j ,
\end{equation}
such that after acting with $U^{-1}$ on $\ket{\Psi_{\mathrm{VBS}}}$, we get \cite{Piroli2017WhatIA}
\begin{equation}
U^{-1}\ket{\Psi_{\mathrm{VBS}}} = 3^{-L/2} \ket{\Psi_0}, \quad \ket{\Psi_0} = \sum_{ \{s\} } \Tr \left[B_{s_1} B_{s_2} \cdots B_{s_L} \right] | s_1, s_2, \cdots s_L \rangle ,
\end{equation}
where
\begin{equation}
    B_+ = \sigma^x , \quad B_0 = \sigma^y , \quad B_- = \sigma^z .
\end{equation}

Acting on $\ket{\Psi_0}$ with the parity operator, we have
\begin{equation}\label{eq:PPsi}
    \mathcal{I} \ket{\Psi_0} = \sum_{ \{s\} } \Tr \left[B_{s_1}^T B_{s_2}^T \cdots B_{s_L}^T \right] | s_1, s_2, \cdots s_L \rangle = \sum_{ \{s\} } (-1)^{n_0} \Tr \left[B_{s_1} B_{s_2} \cdots B_{s_L} \right] | s_1, s_2, \cdots s_L \rangle ,
\end{equation}
where ${}^T$ denotes transpose and $n_0$ counts the number of spin $0$ in state $| s_1, s_2, \cdots s_L \rangle$. Since the transfer matrix $T(\lambda)$ is a matrix product operator, we can express the state
\begin{equation}\label{eq:TTrB}
    T(\lambda) \Tr \left[B_{s_1} B_{s_2} \cdots B_{s_L} \right] | s_1, s_2, \cdots s_L \rangle = \Tr \left[C_{s_1} C_{s_2} \cdots C_{s_L} \right] | s_1, s_2, \cdots s_L \rangle ,
\end{equation}
where the matrices $C_s$ ($s=+,0,-$) are $6$-dimensional. 
In addition, one can show that there exists a similarity transformation $V$ such that
\begin{equation}
    C_{\pm}^T = V  C_{\pm} V^{-1} , \quad C_{0}^T = - V  C_{0} V^{-1} ,
\end{equation}
i.e.
\begin{equation}\label{eq:PTrC}
\begin{split}
    \mathcal{I} \Tr \left[C_{s_1} C_{s_2} \cdots C_{s_L} \right] | s_1, s_2, \cdots s_L \rangle & = \Tr \left[C_{s_1}^T C_{s_2}^T \cdots C_{s_L}^T \right] | s_1, s_2, \cdots s_L \rangle \\ 
    & = (-1)^{n_0} \Tr \left[C_{s_1} C_{s_2} \cdots C_{s_L} \right] | s_1, s_2, \cdots s_L \rangle .
\end{split}
\end{equation}
We are now ready to show that $|\Psi_0 \rangle$ is an integrable boundary state. Using Eqs. (\ref{eq:PPsi}), (\ref{eq:TTrB}), and (\ref{eq:PTrC}), we have 
\begin{equation}
    \begin{split}
        \mathcal{I}\, T(\lambda)\, \mathcal{I}\, |\Psi_0 \rangle & = \mathcal{I} \sum_{ \{s\} } (-1)^{n_0}  \Tr \left[C_{s_1} C_{s_2} \cdots C_{s_L} \right] | s_1, s_2, \cdots s_L \rangle \\
        & = \sum_{ \{s\} } (-1)^{2 n_0}  \Tr \left[C_{s_1} C_{s_2} \cdots C_{s_L} \right] | s_1, s_2, \cdots s_L \rangle \\ 
        & = T(\lambda) |\Psi_0 \rangle .
    \end{split}
\end{equation}
From this, we find 
\begin{equation}
    \mathcal{I}\, T(\lambda)\, \mathcal{I}\, | \Psi_{\mathrm{VBS}} \rangle = 3^{-L/2} U\, \mathcal{I}\, T(\lambda)\, \mathcal{I}\, | \Psi_0 \rangle = 3^{-L/2} U T(\lambda) | \Psi_0 \rangle = T(\lambda) | \Psi_{\mathrm{VBS}} \rangle ,
\end{equation}
which shows that the VBS state is an integrable boundary state of the Sutherland model.

\section{Eigenenergy of ferromagnetic states}
\label{energy_of_ferromagnetic}
In this section, we compute the energy of the ferromagnetic states, which are eigenstates of the Hamiltonian $H(t)$ in Eq. (\ref{AKLT+H3}). The ferromagnetic states are defined as $\ket{F_n} = (\totS^-)^n\ket{F_0}$, where $\ket{F_0} = \ket{++\cdots+}$ and $0 \le n \le 2L$. 
They are the same as $\ket{F_n}_L$ with $\sigma=1$ in Appendix \ref{appendix:ferro_ent_half} up to a normalization factor.
Because $H_{\mathrm{AKLT}}$ and $H_3$ commute with $\totS^-$ (see Appendix \ref{appendix:ferromagnetic}), all ferromagnetic states $\ket{F_n}$ have the same energy. Thus it suffices to consider the energy of $\ket{F_0}$.

Since $\ket{F_0}$ is invariant under any permutation, we have $P_{j, j+1, j+2} \ket{F_0} =  P^\dagger_{j, j+1, j+2} \ket{F_0} = \ket{F_0}$. Thus we obtain
\begin{equation}
    H_3\ket{F_0} = \sum_{j=1}^L(P_{j, j+1, j+2} - P^\dagger_{j, j+1, j+2})\ket{F_0} = 0,
\end{equation}
and hence $H_3\ket{F_n} = H_3\ket{F_0} = 0$ for all $n$. Moreover, since $(\bm{S}_j\cdot\bm{S}_{j+1})\ket{F_0} = \ket{F_0}$ for all $j$, we obtain
\begin{equation}
    H_\mathrm{AKLT}\ket{F_0} = \sum_{j=1}^L\left[1+\frac{1}{3}+\frac{2}{3}\right]\ket{F_0} = 2L\ket{F_0}.
\end{equation}
Therefore, the ferromagnetic states are eigenstates of $H(t)$ with eigenvalue $2L$. A similar calculation shows that the eigenenergy of the ferromagnetic states in the inhomogeneous model (\ref{inhomogeneous AKLT+H3}) is $2\sum_j c_j$.

\section{Proof of \texorpdfstring{$C_\mathrm{SC}\ket{\Psi_\mathrm{VBS}} = 0$}{Lg}}\label{appendix:CSCVBS=0.proof}
In this section, we show 
\begin{equation}
    C_\mathrm{SC}\ket{\Psi_\mathrm{VBS}} = 0.
\end{equation}
First, we rewrite $\ket{\Psi_\mathrm{VBS}}$ as 
\begin{equation}
    \ket{\Psi_\mathrm{VBS}} = 3^{-L/2}\Tr[ {\sf A}_1 {\sf A}_2 \cdots {\sf A}_L ],\quad {\sf A}_j = \mqty(\ket{0}_j & -\sqrt{2}\ket{+}_j \\ \sqrt{2}\ket{-}_j & -\ket{0}_j).
\end{equation}
Next, we introduce a convenient representation of $C_\mathrm{SC}$:
\begin{equation}\label{eq:expand_CSC}
    \bm{S}_j\cdot(\bm{S}_{j+1}\times \bm{S}_{j+2})= \frac{\mathrm{i}}{2}\tau_{\alpha\beta\gamma}S_j^\alpha S_{j+1}^\beta S_{j+2}^\gamma,
\end{equation}
where $\alpha, \beta, \gamma\in\{+, -, z\}$ and $\tau$ is the totally anti-symmetric tensor with $\tau_{+-z}=1$. 
Then, by acting on $\ket{\Psi_\mathrm{VBS}}$ with each term, we obtain
\begin{equation}
    \tau_{\alpha\beta\gamma}S_j^\alpha S_{j+1}^\beta S_{j+2}^\gamma\ket{\Psi_\mathrm{VBS}} = 3^{-L/2}\Tr [ {\sf A}_1 \cdots {\sf A}_{j-1} {\sf B}_{j, j+1, j+2}^{\alpha\beta\gamma} {\sf A}_{j+3}\cdots ], 
\end{equation}
where
\begin{align}
    {\sf B}_{j, j+1, j+2}^{+-z}&=\sqrt{2}\mqty(\ket{+0-} & -\ket{-+-} \\ 2\ket{00-} - \ket{+--} & -\sqrt{2}\ket{0-+}), \\
    {\sf B}_{j, j+1, j+2}^{+z-}&=\sqrt{2}\mqty(0 & -\ket{++-} \\ -\ket{+--} & \sqrt{2}(\ket{+-0}-\ket{0+-})), \\
    {\sf B}_{j, j+1, j+2}^{-z+}&=\sqrt{2}\mqty(\sqrt{2}(\ket{0-+}-\ket{-+0}) & \ket{-++} \\ \ket{--+} & 0), \\
    {\sf B}_{j, j+1, j+2}^{-+z}&=\sqrt{2}\mqty(\sqrt{2}\ket{0+-} & \ket{-++} - 2\ket{00+} \\ \ket{-+-} & -\sqrt{2}\ket{-0+}), \\
    {\sf B}_{j, j+1, j+2}^{z+-}&= \sqrt{2}\mqty(-\sqrt{2}\ket{+0-} & 2\ket{+00} - \ket{++-} \\ -\ket{-+-} & \sqrt{2}\ket{-+0}), \\
    {\sf B}_{j, j+1, j+2}^{z-+}&= \sqrt{2}\mqty(-\sqrt{2}\ket{+-0} & \ket{+-+} \\ \ket{--+} - 2\ket{-00} & \sqrt{2}\ket{-0+}),
\end{align}
and ${\sf B}_{j, j+1, j+2}^{\alpha\beta\gamma} = 0$ for any other choice of $(\alpha, \beta, \gamma)$. Therefore, we obtain
\begin{equation}
    \bm{S}_j\cdot(\bm{S}_{j+1}\times \bm{S}_{j+2})\ket{\Psi_\mathrm{VBS}} = 3^{-L/2}\Tr [ {\sf A}_1 \cdots {\sf A}_{j-1} {\sf B}_{j, j+1, j+2} {\sf A}_{j+3}\cdots ],
\end{equation}
where
\begin{align}
    {\sf B}_{j, j+1, j+2} &= \frac{\mathrm{i}}{2}\sum_{\alpha, \beta, \gamma\in \{+, -, z\}} {\sf B}_{j, j+1, j+2}^{\alpha\beta\gamma} \nonumber\\
    &= \mathrm{i}\mqty(\ket{0-+}-\ket{-+0}+\ket{0+-}-\ket{+-0} & \sqrt{2}(\ket{+00}-\ket{00+}+\ket{-++}-\ket{++-}) \\ \sqrt{2}(\ket{00-} - \ket{-00} + \ket{--+} - \ket{+--}) & \ket{+-0} + \ket{-+0} -\ket{0-+}-\ket{0+-}).
\end{align}
Then, we can decompose ${\sf B}_{j, j+1, j+2}$ into
\begin{equation}
    {\sf B}_{j, j+1, j+2} = {\sf A}_j {\sf C}_{j+1, j+2} - {\sf C}_{j, j+1} {\sf A}_{j+2},
\end{equation}
where
\begin{equation}
    {\sf C}_{j, j+1} = \mathrm{i}\mqty(\ket{+-}+\ket{-+}-\ket{00} & 0 \\ 0 & \ket{+-}+\ket{-+}-\ket{00})
\end{equation}
Therefore, we have
\begin{equation}
    C_\mathrm{SC}\ket{\Psi_\mathrm{VBS}} = 3^{-L/2}\sum_{j=1}^L \left\{
    \Tr [ {\sf A}_1 \cdots {\sf A}_{j-1} {\sf A}_j {\sf C}_{j+1, j+2} {\sf A}_{j+3} \cdots {\sf A}_L] - \Tr [ {\sf A}_1 \cdots {\sf A}_{j-1} {\sf C}_{j, j+1} {\sf A}_{j+2} {\sf A}_{j+3} \cdots {\sf A}_L ]
    \right\} = 0.
\end{equation}

\section{Zero energy states of \texorpdfstring{$C_\mathrm{SC}$}{Lg}}\label{appendix:AB_proof}
In this section, we prove that $\ket{\bar{A}_n}$, $\ket{\bar{B}_n}$, and $\ket{K_{0,p}}$ are annihilated by $C_\mathrm{SC}$. We also derive a lower bound on the number of zero-energy states of $C_\mathrm{SC}$. 

\subsection{\texorpdfstring{$C_\mathrm{SC}\ket{\bar{A}_n}=0$}{Lg}}
To prove $C_\mathrm{SC}\ket{\bar{A}_n} = 0$, we first prove the following: 
\begin{theorem}\label{thm:f1}
Consider $C_\mathrm{SC}$ with an even number of sites $L$. Then
%When the number of sites $L$ is even, 
the following relations hold:
\begin{align}
    & C_\mathrm{SC}\mathcal{O}^-_\pi\ket{\Uparrow} = 0, \label{a1}\\
    & C_\mathrm{SC}(\mathcal{O}^-_\pi)^2\ket{\Uparrow} = 0, \label{a2}\\
    & [\,\mathcal{O}^-_\pi, [\, \mathcal{O}^-_\pi, [\, \mathcal{O}^-_\pi, C_\mathrm{SC}\,] \,] \,] = 0. \label{a3}
\end{align}
\end{theorem}
\begin{proof}

We first prove Eq. (\ref{a1}). We can write $\mathcal{O}^-_\pi\ket{\Uparrow}$ in the $S^z$ basis as
\begin{equation}
    \mathcal{O}^-_\pi\ket{\Uparrow} = \sqrt{2}\sum_{j=1}^L(-1)^j \ket{j},
\end{equation}
where we have used the short-hand notation $\ket{j}\coloneqq\ket{+\cdots + 0_j +\cdots+}$. 
Then we find $C_\mathrm{SC}\ket{j} = \mathrm{i}(\ket{j-2}-\ket{j+2})$, which yields
\begin{equation}
    C_\mathrm{SC}\mathcal{O}^-_\pi\ket{\Uparrow} = \sqrt{2}\mathrm{i}\sum_{j=1}^L(-1)^j(\ket{j-2}-\ket{j+2}) =0.
\end{equation}
We can also see this as follows. First note that $C_\mathrm{SC}\mathcal{O}^-_\pi\ket{\Uparrow}$ is odd under the site-centered inversion ${\cal I}_\mathrm{s}$. We then note that it is invariant under translation by two sites ${\cal T}^2$. However, there is no single-magnon state (a linear combination of $\ket{j}$) that is compatible with these constraints. Thus, $C_\mathrm{SC}\mathcal{O}^-_\pi\ket{\Uparrow}$ must vanish identically. 

Next, we consider Eq. (\ref{a2}). By acting on $\ket{\Uparrow}$ with $\mathcal{O}_\pi^-$ twice, we get
\begin{equation}
    (\mathcal{O}_\pi^-)^2\ket{\Uparrow} = 2\left(2\sum_{1\leq j<k\leq L}(-1)^{j+k}\ket{j, k}+\sum_{j=1}^L\ket{\bar{j}}\right),
\end{equation}
where $\ket{j, k} \coloneqq \ket{+\cdots + 0_j + \cdots + 0_k + \cdots +}$ and $\ket{\bar{j}}\coloneqq \ket{+\cdots + -_j +\cdots+}$. We can rewrite the first term as
\begin{equation}\label{eq:ket|j, k>}
    \sum_{j< k}(-1)^{j+k}\ket{j, k} = 
        \sum_{j=1}^{L}\left(\sum_{r=1}^{L/2-1}(-1)^r\ket{j, j+r} + (-1)^{\frac{L}{2}}\ket{j, j+\frac{L}{2}}\right)
    .
\end{equation}
We now examine the action of $C_\mathrm{SC}$ on $\ket{j, j+r}$. For $r\geq 3$, we get 
\begin{equation}\label{eq:r>3}
    C_\mathrm{SC}\ket{j, j+r} = \mathrm{i}(\ket{j-2, j+r}-\ket{j+2, j+r}+\ket{j, j+r-2}-\ket{j, j+r+2})\quad(r\geq 3).
\end{equation}
Therefore, we obtain
\begin{equation}\label{eq:r>3 2}
    C_\mathrm{SC}\sum_{j=1}^L\ket{j, j+r} = 0 \quad (r\geq 3).
\end{equation}
Next, we consider the case with $r = 2$. In this case, the action of $C_\mathrm{SC}$ is
\begin{align}
    &\quad C_\mathrm{SC}\ket{j, j+2} \nonumber\\ 
    &= \mathrm{i}(\ket{j-2, j+2} - 2\ket{j-1, j+2} + \ket{j+1, j+2} + \ket{\overline{j}} - \ket{\overline{j+2}} + 2\ket{j, j+3} - \ket{j, j+1}  -\ket{j, j+4}), 
\end{align} 
which yields
\begin{equation}\label{eq:r=2}
    C_\mathrm{SC}\sum_{j=1}^L\ket{j, j+2} = 0.
\end{equation}
Similarly, for $\ket{j, j+1}$, we obtain
\begin{align}
    C_\mathrm{SC}\ket{j, j+1}
    = \mathrm{i}(\ket{j-2, j+1}-\ket{j-1, j+1}+2\ket{\overline{j+1}}-2\ket{\overline{j}}+\ket{j, j+2}-\ket{j, j+3}). 
\end{align}
Thus, we have
\begin{equation}\label{eq:r=1}
    \sum_{j=1}^LC_\mathrm{SC}\ket{j, j+1} = 0.
\end{equation}
We next examine $\ket{\overline{j}}$. Acting with $C_\mathrm{SC}$ on $\ket{\overline{j}}$, we obtain
\begin{align}
    C_\mathrm{SC}\ket{\overline{j}} 
    = \mathrm{i}(\ket{j-2, j} - 2\ket{j-1, j} + 2\ket{j, j+1} + \ket{j, j+2}).
\end{align}
Thus, we have
\begin{equation}\label{eq:CSCketj=0}
    \sum_{j=1}^LC_\mathrm{SC}\ket{\bar{j}} = 0.
\end{equation}
Putting this all together, we get  
\begin{equation}
    C_\mathrm{SC}(\mathcal{O}_\pi^-)^2\ket{\Uparrow} = 0.
\end{equation}

Finally, we consider Eq. (\ref{a3}). Since $C_\mathrm{SC}$ is a sum of terms of the form $S_j^\alpha S_{j+1}^\beta S_{j+2}^\gamma$ $(\{\alpha, \beta, \gamma\} = \{+, -, z\})$ and $\mathcal{O}^-_\pi$ is a linear combination of $S^-_k$, the nested commutator $[\, \mathcal{O}^-_\pi, [\, \mathcal{O}^-_\pi, [\, \mathcal{O}^-_\pi, C_\mathrm{SC}\, ] \,] \,]$ is a sum of $S_j^-S_{j+1}^-S_{j+2}^-$. Let $O^{\alpha\beta\gamma}_j \coloneqq [\, \mathcal{O}^-_\pi, [\, \mathcal{O}^-_\pi, [\, \mathcal{O}^-_\pi, S_j^\alpha S_{j+1}^\beta S_{j+2}^\gamma\, ] \, ] \,]$ for $\{\alpha, \beta, \gamma\} = \{+, -, z\}$, we obtain
\begin{align}
    O^{+-z}_j &= 3[\, \mathcal{O}^-_\pi, [\, \mathcal{O}^-_\pi, S^+_j \, ] \,] \, S^-_{j+1}\, [\, \mathcal{O}^-_\pi, S_{j+2}^z \,] \nonumber\\
    &= 3\cdot (-2)\cdot(-1)^{2j}S_j^-\cdot S_{j+1}^-\cdot (-1)^{j+2}S_{j+2}^- \\
    &= -6(-1)^{j+2}S_j^-S_{j+1}^-S_{j+2}^-, \\
    O^{-+z}_j &= 3S_j^- [\, \mathcal{O}^-_\pi, [\, \mathcal{O}^-_\pi, S^+_{j+1}\, ] \,] \, [\, \mathcal{O}^-_\pi, S_{j+2}^z \,] \nonumber\\
    &= 3\cdot  S_j^-\cdot(-2)\cdot(-1)^{2(j+1)} S_{j+1}^-\cdot (-1)^{j+2}S_{j+2}^- \\
    &= -6(-1)^{j+2}S_j^-S_{j+1}^-S_{j+2}^-. 
\end{align}
Hence, $O^{+-z}_j = O^{-+z}_j$. Similarly, one can show that $O^{-z+}_j=O^{+z-}_j$ and $O^{z+-}_j = O^{z-+}_j$. Therefore, we have
\begin{equation}
    [\, \mathcal{O}^-_\pi, [\, \mathcal{O}^-_\pi, [\, \mathcal{O}^-_\pi, C_\mathrm{SC} \, ] \, ] \, ] = \frac{\mathrm{i}}{2}\sum_{j=1}^L\sum_{\alpha, \beta, \gamma}\tau_{\alpha\beta\gamma}O^{\alpha\beta\gamma}_j = 0.
\end{equation}
\end{proof}

From Theorem \ref{thm:f1}, $C_\mathrm{SC}\ket{\bar{A}_n} = 0$ follows immediately. 

\subsection{\texorpdfstring{$C_\mathrm{SC}\ket{\bar{B}_n}=0$}{Lg}}
Next we prove that $C_\mathrm{SC}\ket{\bar{B}_n} = 0$. To this end, we first prove the following:
\begin{theorem}\label{thm:f2}
The following relations are true.
\begin{align}
    &[\, \mathcal{Q}^-_0, C_\mathrm{SC} \,]\ket{\Uparrow} = 0, \label{b1}\\
    &[\, \mathcal{Q}^-_0, [\, \mathcal{Q}^-_0, C_\mathrm{SC} \,] \,]\ket{\Uparrow} = 0, \label{b2}\\
    &[\, \mathcal{Q}^-_0, [\, \mathcal{Q}^-_0, [\, \mathcal{Q}^-_0, C_\mathrm{SC} \,] \,] \,] = 0. \label{b3}
\end{align}
\end{theorem}
\begin{proof}

We consider the coherent state $\ket{\beta}$, which can be written as 
\begin{equation}
    \ket{\beta} \propto e^{\beta\mathcal{Q}_0^-}\ket{\Uparrow} =  \prod_{j=1}^L(1+\beta(S_j^-)^2)\ket{\Uparrow} \eqqcolon \bigotimes_{j=1}^L\ket{\psi_\beta}_j ,
\end{equation}
where $\ket{\psi_\beta}_j = \ket{+}_j+2\beta\ket{-}_j$. One can prove that $\ket{\beta}$ is annihilated by $C_\mathrm{SC}$. This can be seen as follows. First note that $S^+_j\ket{\psi_\beta}_j = 2\sqrt{2}\beta\ket{0}_j$ and $S^-_j\ket{\psi_\beta}_j = \sqrt{2}\ket{0}_j$, and hence $(S^+_j S^-_k -S^-_j S^+_k) \ket{\psi_\beta}_j\otimes\ket{\psi_\beta}_{k}=0$. Next, we note that each summand of $C_\mathrm{SC}$ can be cast in the form:
\begin{align}
    \bm{S}_j\cdot(\bm{S}_{j+1}\times\bm{S}_{j+2}) 
    =\frac{\mathrm{i}}{2} \left\{
    (S^+_j S^-_{j+1} - S^-_j S^+_{j+1}) S^z_{j+2} + (S^+_{j+1} S^-_{j+2} - S^-_{j+1} S^+_{j+2}) S^z_j + (S^+_{j+2} S^-_j - S^-_{j+2} S^+_j) S^z_{j+1}
    \right\}.
\end{align}
From this, it is clear that each summand annihilates $\ket{\beta}$, and hence  $C_\mathrm{SC}e^{\beta\mathcal{Q}_0^-}\ket{\Uparrow} = 0$. 
Acting with $e^{-\beta\mathcal{Q}_0^-}$ from the left on both sides of this equation and expanding it by the Baker-Campbell-Hausdorff formula, we have
\begin{equation}
    \left(C_\mathrm{SC} -\beta [\, \mathcal{Q}^-_0, C_\mathrm{SC} \,] + \frac{\beta^2}{2}[\, \mathcal{Q}^-_0, [\, \mathcal{Q}^-_0, C_\mathrm{SC} \,] \,] + \cdots\right)\ket{\Uparrow} = 0,
\end{equation}
which proves Eqs. (\ref{b1}) and (\ref{b2}) since $\beta\in\mathbb{C}$ can be taken arbitrarily. 

Finally, we show Eq. (\ref{b3}). Let $Q_j^{\alpha\beta\gamma}=[\, \mathcal{Q}^-_0, [\, \mathcal{Q}^-_0, [\, \mathcal{Q}^-_0,S_j^\alpha S_{j+1}^\beta S_{j+2}^\gamma \,] \,] \,]$ $(\alpha, \beta, \gamma = +, -, \mathrm{or}\,\, z)$. Then, we obtain
\begin{align}
    Q_j^{+-z}&=3[\, \mathcal{Q}^-_0, [\, \mathcal{Q}^-_0, S^+_j \,] \,]S^-_{j+1}[\, \mathcal{Q}^-_0, S_{j+2}^z \,]\nonumber\\
    &= 3\cdot(-8(S_j^-)^3)\cdot S_{j+1}^-\cdot 2(S_{j+2}^-)^2 =0.
\end{align}
In the same way, one can show $Q_j^{-+z}=Q_j^{-z+}=Q_j^{+z-}=Q_j^{z+-}=Q_j^{z-+}=0$. Therefore, 
\begin{equation}
    [\, \mathcal{Q}^-_0, [\, \mathcal{Q}^-_0, [\, \mathcal{Q}^-_0,C_\mathrm{SC} \,] \,] \,]=\frac{\mathrm{i}}{2}\sum_{j=1}^L\sum_{\alpha, \beta, \gamma}\tau_{\alpha\beta\gamma}Q_j^{\alpha\beta\gamma}=0.
\end{equation}

\end{proof}

From Theorem \ref{thm:f2}, $C_\mathrm{SC}\ket{\bar{B}_n} = 0$ follows immediately. 

\subsection{\texorpdfstring{$C_\mathrm{SC}\ket{K_{0, p}}=0$}{Lg}}
Here we prove that $C_\mathrm{SC}\ket{K_{0, p}}=0$. From Eq. (\ref{eq:K_to_phi}), each $\ket{K_{0, p}}$ can be expressed as a linear combination of $\ket{\Phi_n} = \sum_{j=1}^L\ket{j, j+n}$, where $n = 0, 1, \ldots, \lfloor L / 2 \rfloor$ and $\ket{j, j}=\ket{\bar{j}}$. However, it has already been shown by Eqs. (\ref{eq:r>3 2}), (\ref{eq:r=2}), (\ref{eq:r=1}), and (\ref{eq:CSCketj=0}) that these states are annihilated by $C_\mathrm{SC}$. Therefore, $\ket{K_{0, p}}$ are zero-energy states of $C_\mathrm{SC}$.

\subsection{Lower bound on the number of zero-energy states}
In this subsection, we derive a lower bound on the number of zero-energy states of $C_\mathrm{SC}$. We follow the argument in Ref. \cite{Turner2}, where the authors obtained a lower bound on the number of zero-energy states of the PXP model. The key point is that the site-centered inversion ${\cal I}_\mathrm{s}$ anticommutes with the Hamiltonian $C_\mathrm{SC}$ in Eq. (\ref{eq:spin-1 SC}), i.e., ${\cal I}_\mathrm{s} C_\mathrm{SC} = - C_\mathrm{SC} {\cal I}_\mathrm{s}$. 

Let ${\cal H}$ be the Hilbert space of a spin-$1$ chain of length $L$. This Hilbert space can be decomposed as ${\cal H} = {\cal K}_\mathrm{e} \oplus {\cal K}_\mathrm{o}$, where ${\cal K}_\mathrm{e} = \{ \ket{\psi} \in {\cal H} \,|\, {\cal I}_\mathrm{s} \ket{\psi} = \ket{\psi} \}$ and ${\cal K}_\mathrm{o} = \{ \ket{\psi} \in {\cal H} \,|\, {\cal I}_\mathrm{s} \ket{\psi} = -\ket{\psi} \}$. 
It follows from $\{ {\cal I}_\mathrm{s}, C_\mathrm{SC} \}=0$ that if $\ket{\psi} \in {\cal K}_\mathrm{e/o}$ then $C_\mathrm{SC} \ket{\psi} \in {\cal K}_\mathrm{o/e}$. Therefore, $C_\mathrm{SC}$ can be written in block-matrix form as
\begin{align}
    C_\mathrm{SC} = 
    \left(
    \begin{matrix}
    O & D^\dagger_\mathrm{SC} \\
    D_\mathrm{SC} & O
    \end{matrix}
    \right).
\end{align}
Here the operator $D_\mathrm{SC}$ can be regarded as a linear map from ${\cal K}_\mathrm{e}$ to ${\cal K}_\mathrm{o}$. Let ${\rm Im}\, D_\mathrm{SC}$ and ${\rm Ker}\, D_\mathrm{SC}$ be the image and kernel of $D_\mathrm{SC}$, respectively. It is clear that if $\ket{\psi} \in {\rm Ker}\,D_\mathrm{SC}$, then $\ket{\psi}$ is annihilated by $C_\mathrm{SC}$. Thus, the dimension of ${\rm Ker}\,D_\mathrm{SC}$ gives a lower bound on the number of zero-energy states. We now apply the rank-nullity theorem to estimate ${\rm dim}\, {\rm Ker}\, D_\mathrm{SC}$. The theorem implies that
\begin{align}
    {\rm dim}\, {\rm Im}\, D_\mathrm{SC} + {\rm dim}\, {\rm Ker}\, D_\mathrm{SC} = {\rm dim}\, {\cal K}_\mathrm{e}, 
\end{align}
Since ${\rm dim}\, {\rm Im}\, D_\mathrm{SC} \le {\rm dim}\, {\cal K}_\mathrm{o}$, we have 
\begin{align}\label{eq:rank-nullity}
    {\rm dim}\, {\rm Ker}\, {\cal D}_\mathrm{SC} \ge {\rm dim }\, {\cal K}_\mathrm{e} - {\rm dim }\, {\cal K}_\mathrm{o},
\end{align}
which gives a lower bound on the number of zero-energy states. %of $C_\mathrm{SC}$.

Before deriving a general expression for the RHS of Eq. (\ref{eq:rank-nullity}), let us consider a simple example that illustrates the strategy. For $L=3$, the Hilbert space ${\cal H}$ is spanned by $27$ states. Consider the inversion about site $2$. Then ${\cal K}_\mathrm{e}$ is spanned by the states of the forms $\ket{s_1, s_2, s_1}$ and $\ket{s_1, s_2, s_3} + \ket{s_3, s_2, s_1}$ ($s_1 < s_3$). The number of these states amounts to $9+9=18$. On the other hand, ${\cal K}_\mathrm{o}$ is spanned by the states of the form $\ket{s_1, s_2, s_3} - \ket{s_3, s_2, s_1}$ ($s_1 < s_3$), the number of which amounts to $9$. Thus, ${\rm dim }\, {\cal K}_\mathrm{e} - {\rm dim }\, {\cal K}_\mathrm{o}=9$. 

The above example clearly illustrates that the difference between the dimensions of even and odd subspaces counts the number of product states invariant under ${\cal I}_\mathrm{s}$. Let ${\cal N}_L$ be the number of such states for the $L$-site system. Let ${\cal Z}_L$ be the exact number of zero-energy states of $C_\mathrm{SC}$. It is easy to see that ${\cal N}_L = 3^{\frac{L+1}{2}}$ for $L$ odd and ${\cal N}_L = 3^{\frac{L+2}{2}}$ for $L$ even. These results can be summarized as ${\cal Z}_L \ge {\cal N}_L = 3^{\lfloor \frac{L}{2} \rfloor +1}$, which proves that ${\cal Z}_L$ grows exponentially with the system size. Table \ref{tab:comparison} shows the comparison between ${\cal Z}_L$ obtained by exact diagonalization and the bound ${\cal N}_L$. Clearly, ${\cal Z}_L$ grows more rapidly than ${\cal N}_L$. We expect that a better lower bound can be obtained by considering other symmetries of the Hamiltonian, but leave this possibility for future work. We note in passing that a lower bound on ${\cal Z}_L$ for general spin quantum number $\sigma$ can also be derived in a similar manner; the result is ${\cal Z}_L \ge (2\sigma+1)^{\lfloor \frac{L}{2} \rfloor +1}$. 
\begin{table}[]
    \centering
    \begin{tabular}{c|ccccccccc}
    \hline\hline
        $L$ & $3$ & $4$ & $5$ & $6$ & $7$ & $8$ & $9$ & $10$ & $11$\\
        \hline
        ${\cal Z}_L$ & $11$ & $35$ & $45$ & $127$ & $141$ & $435$ & $473$ & $1451$ & $1553$ \\
        ${\cal N}_L$ & $9$ & $27$ & $27$ & $81$ & $81$ & $243$ & $243$ & $729$ & $729$\\
        \hline\hline 
    \end{tabular}
    \caption{The number of zero-energy states (${\cal Z}_L$) and the bound (${\cal N}_L$) up to $L=11$ sites.}
    \label{tab:comparison}
\end{table}

\section{Time evolution of a superposition of \texorpdfstring{$\ket{\Bar{A}_n}$}{Lg} and \texorpdfstring{$\ket{\Bar{B}_n}$}{Lg}}\label{appendix:dynamics}
We have discussed the dynamics of the coherent states $\ket{\alpha}$ and $\ket{\beta}$ in Sec. \ref{sec_CSC}.  
In this appendix, we consider the dynamics from a more complex initial state. As we have seen, the system with the Hamiltonian $H_2$ in Eq. (\ref{perm_ham}) has two types of scar states: $\ket{\Bar{A}_n}$ and $\ket{\Bar{B}_n}$. We will show that their superpositions exhibit more complex dynamics than those in the main text. 

To be specific, let us consider the following superposition of $\ket{\alpha}$ and $\ket{\beta}$:
\begin{equation}
    \ket{\xi} =\frac{1}{Z} (u\ket{\alpha} + v\ket{\beta}),
\end{equation}
where $u, v \in \mathbb{C}$ are arbitrary constants and $Z$ is the normalization constant. The fidelity between $\ket{\xi}$ and the time evolved state $\ket{\xi (t)} = e^{-\mathrm{i} H_2 t} \ket{\xi}$ can be expressed as
\begin{equation}\label{eq:fidelity_mix}
    \mathcal{F}(t) = |\braket{\xi}{\xi (t)}| =\frac{1}{Z^2} \left|\abs{u}^2\braket{\alpha}{\alpha(t)} + u^*v\braket{\alpha}{\beta(t)} + uv^*\braket{\beta}{\alpha(t)} + \abs{v}^2\braket{\beta}{\beta(t)}\right|,
\end{equation}
where $\ket{\alpha (t)} = e^{-\mathrm{i} H_2 t} \ket{\alpha}$ and $\ket{\beta (t)} = e^{-\mathrm{i} H_2 t} \ket{\beta}$. To get a more explicit expression for $\mathcal{F}(t)$, let us compute the overlaps. Along the same lines as in Eqs. (\ref{eq:fidelityB}, \ref{eq:fidelityA}), one can calculate the first and fourth overlaps in Eq. (\ref{eq:fidelity_mix}) as
\begin{equation}\label{eq:overlap1}
    \braket{\alpha}{\alpha(t)} = e^{-\mathrm{i}t (hL+\mathcal{D})} \left(\frac{1+\abs{\alpha}^2e^{\mathrm{i}ht}}{1+\abs{\alpha}^2}\right)^{2L}, \quad
    \braket{\beta}{\beta(t)} = e^{-\mathrm{i}t (hL+\mathcal{D})} \left(\frac{1+4|\beta|^2e^{2\mathrm{i}ht}}{1+4\abs{\beta}^2}\right)^L,
\end{equation}
where $\mathcal{D} =\sum^L_{j=1} D_j$ and we have used the fact that $\ket{\bar{A}_n}$ and $\ket{\Bar{B}_n}$ are eigenstates of $H_2$ with eigenvalues $h(L-n) + \mathcal{D}$ and $h(L-2n) + \mathcal{D}$, respectively. Next, let us compute the second and third overlaps in Eq. (\ref{eq:fidelity_mix}). To this end, we consider the overlap between $\ket{\bar{A}_m}$ and $\ket{\Bar{B}_n}$. Since they are eigenstates of $\mathcal{S}^z$ with eigenvalues $L-m$ and $L-2n$, respectively, it is easy to see that $\braket{\bar{A}_m}{\bar{B}_n}\propto \delta_{m, 2n}$. 
The overlap for $m=2n$ is calculated as
\begin{align}
    \braket{\bar{A}_{2n}}{\bar{B}_n} &= \bra{\Uparrow}(\mathcal{O}^+_\pi)^{2n} (\mathcal{Q}^-_0)^n\ket{\Uparrow} \nonumber\\
    &= \frac{(2n)!}{2^n}\bra{\Uparrow}\left(\sum_{1\leq j_1<\cdots <j_n\leq L}(S^+_{j_1})^2(S^+_{j_2})^2\cdots(S^+_{j_n})^2\right) n!\left(\sum_{1\leq l_1<\cdots<l_n\leq L}(S^-_{l_1})^2(S^-_{l_2})^2\cdots(S^-_{l_n})^2\right)\ket{\Uparrow} \nonumber\\
    &= 2^n\cdot (2n)!\cdot n! \cdot \binom{L}{n},
\end{align}
from which we obtain
\begin{align}\label{eq:overlap2}
    \braket{\alpha}{\beta(t)} 
    =  \frac{e^{-\mathrm{i}t (hL+\mathcal{D})} (1+2(\alpha^*)^2\beta e^{2\mathrm{i}ht})^L}{(1+\abs{\alpha}^2)^L(1+4\abs{\beta}^2)^\frac{L}{2}},
    \quad
    \braket{\beta}{\alpha(t)} = \frac{e^{-\mathrm{i}t (hL+\mathcal{D})} (1+2\alpha^2\beta^* e^{2\mathrm{i}ht})^L}{(1+\abs{\alpha}^2)^L(1+4\abs{\beta}^2)^\frac{L}{2}}.
\end{align}
Plugging Eqs. (\ref{eq:overlap1}) and (\ref{eq:overlap2}) into Eq. (\ref{eq:fidelity_mix}) yields
\begin{equation}
\begin{split}
    \mathcal{F}(t) = & \frac{1}{Z^2}\left| \abs{u}^2\left(\frac{1+\abs{\alpha}^2e^{\mathrm{i}ht}}{1+\abs{\alpha}^2}\right)^{2L} + \abs{v}^2\left(\frac{1+4|\beta|^2e^{2\mathrm{i}ht}}{1+4\abs{\beta}^2}\right)^L \right. \\ & \left. + u^*v\frac{(1+2(\alpha^*)^2\beta e^{2\mathrm{i}ht})^L}{(1+\abs{\alpha}^2)^L(1+4\abs{\beta}^2)^\frac{L}{2}} + uv^*\frac{(1+2\alpha^2\beta^* e^{2\mathrm{i}ht})^L}{(1+\abs{\alpha}^2)^L(1+4\abs{\beta}^2)^\frac{L}{2}}  \right| ,
\end{split}
\end{equation}
with
\begin{equation}
    Z^2 = \abs{\abs{u}^2+\abs{v}^2+\frac{2\Re{u^*v(1+2(\alpha^*)^2\beta)}}{(1+\abs{\alpha}^2)^L(1+4\abs{\beta}^2)^\frac{L}{2}}}.
\end{equation}

Figure \ref{fig:a+b} shows $\mathcal{F}(t)$ for two different choices of $(\alpha, \beta)$. 
We can see that the fidelity shows revivals with period $2\pi/h$, which is the smallest common period of the two fidelity oscillations shown in Figs. \ref{fig:F(t)_b} and \ref{fig:fidelity_A}. Clearly, the trend of the curves is more complicated than the previous ones, with small peaks originating from the interference terms $\braket{\alpha}{\beta (t)}$ and $\braket{\beta}{\alpha (t)}$.
\begin{figure}[tbph]
    \centering
    \includegraphics[width=0.6\linewidth]{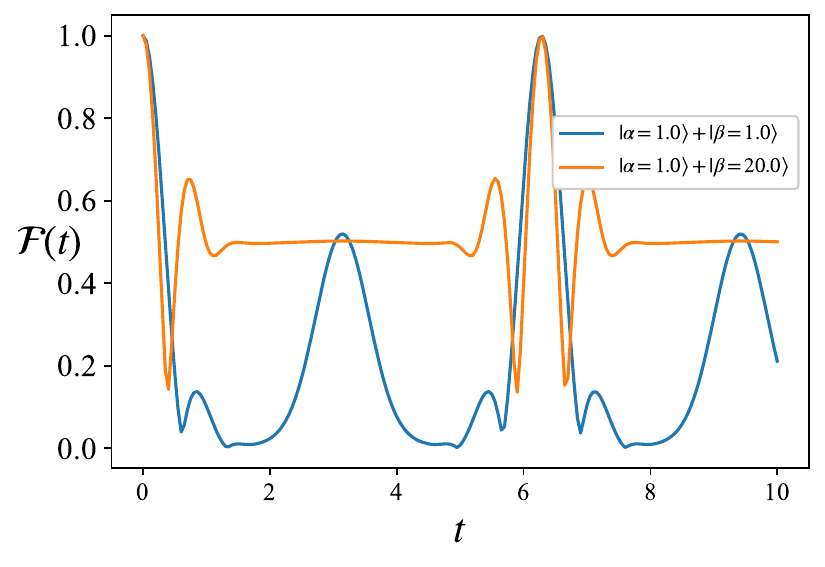}
    \caption{The dynamics of the fidelity of the superposition of two coherent states $\ket{\alpha}$ and $\ket{\beta}$ driven by $H_2$ (Eq. (\ref{perm_ham})) with $h = 1, L = 8$, and $D_j (j = 1, 2, \ldots, L)$ chosen randomly from $[-1, 1]$. The period of the revivals is $2\pi/h$. 
    }
    \label{fig:a+b}
\end{figure}

\end{widetext}

%\highlightReference{zhang2023many}
%\bibliographystyle{apsrev4-2}
\bibliography{sanada_scar}

%apsrev4-2.bst 2019-01-14 (MD) hand-edited version of apsrev4-1.bst
%Control: key (0)
%Control: author (8) initials jnrlst
%Control: editor formatted (1) identically to author
%Control: production of article title (0) allowed
%Control: page (0) single
%Control: year (1) truncated
%Control: production of eprint (0) enabled
\begin{thebibliography}{134}%
\makeatletter
\providecommand \@ifxundefined [1]{%
 \@ifx{#1\undefined}
}%
\providecommand \@ifnum [1]{%
 \ifnum #1\expandafter \@firstoftwo
 \else \expandafter \@secondoftwo
 \fi
}%
\providecommand \@ifx [1]{%
 \ifx #1\expandafter \@firstoftwo
 \else \expandafter \@secondoftwo
 \fi
}%
\providecommand \natexlab [1]{#1}%
\providecommand \enquote  [1]{``#1''}%
\providecommand \bibnamefont  [1]{#1}%
\providecommand \bibfnamefont [1]{#1}%
\providecommand \citenamefont [1]{#1}%
\providecommand \href@noop [0]{\@secondoftwo}%
\providecommand \href [0]{\begingroup \@sanitize@url \@href}%
\providecommand \@href[1]{\@@startlink{#1}\@@href}%
\providecommand \@@href[1]{\endgroup#1\@@endlink}%
\providecommand \@sanitize@url [0]{\catcode `\\12\catcode `\$12\catcode
  `\&12\catcode `\#12\catcode `\^12\catcode `\_12\catcode `\%12\relax}%
\providecommand \@@startlink[1]{}%
\providecommand \@@endlink[0]{}%
\providecommand \url  [0]{\begingroup\@sanitize@url \@url }%
\providecommand \@url [1]{\endgroup\@href {#1}{\urlprefix }}%
\providecommand \urlprefix  [0]{URL }%
\providecommand \Eprint [0]{\href }%
\providecommand \doibase [0]{https://doi.org/}%
\providecommand \selectlanguage [0]{\@gobble}%
\providecommand \bibinfo  [0]{\@secondoftwo}%
\providecommand \bibfield  [0]{\@secondoftwo}%
\providecommand \translation [1]{[#1]}%
\providecommand \BibitemOpen [0]{}%
\providecommand \bibitemStop [0]{}%
\providecommand \bibitemNoStop [0]{.\EOS\space}%
\providecommand \EOS [0]{\spacefactor3000\relax}%
\providecommand \BibitemShut  [1]{\csname bibitem#1\endcsname}%
\let\auto@bib@innerbib\@empty
%</preamble>
\bibitem [{\citenamefont {Trotzky}\ \emph {et~al.}(2012)\citenamefont
  {Trotzky}, \citenamefont {Chen}, \citenamefont {Flesch}, \citenamefont
  {McCulloch}, \citenamefont {Schollw{\"o}ck}, \citenamefont {Eisert},\ and\
  \citenamefont {Bloch}}]{ultracold}%
  \BibitemOpen
  \bibfield  {author} {\bibinfo {author} {\bibfnamefont {S.}~\bibnamefont
  {Trotzky}}, \bibinfo {author} {\bibfnamefont {Y.-A.}\ \bibnamefont {Chen}},
  \bibinfo {author} {\bibfnamefont {A.}~\bibnamefont {Flesch}}, \bibinfo
  {author} {\bibfnamefont {I.~P.}\ \bibnamefont {McCulloch}}, \bibinfo {author}
  {\bibfnamefont {U.}~\bibnamefont {Schollw{\"o}ck}}, \bibinfo {author}
  {\bibfnamefont {J.}~\bibnamefont {Eisert}},\ and\ \bibinfo {author}
  {\bibfnamefont {I.}~\bibnamefont {Bloch}},\ }\bibfield  {title} {\bibinfo
  {title} {{Probing the relaxation towards equilibrium in an isolated strongly
  correlated one-dimensional Bose gas}},\ }\href
  {https://doi.org/10.1038/nphys2232} {\bibfield  {journal} {\bibinfo
  {journal} {Nat. Phys.}\ }\textbf {\bibinfo {volume} {8}},\ \bibinfo {pages}
  {325} (\bibinfo {year} {2012})}\BibitemShut {NoStop}%
\bibitem [{\citenamefont {Neill}\ \emph {et~al.}(2016)\citenamefont {Neill},
  \citenamefont {Roushan}, \citenamefont {Fang}, \citenamefont {Chen},
  \citenamefont {Kolodrubetz}, \citenamefont {Chen}, \citenamefont {Megrant},
  \citenamefont {Barends}, \citenamefont {Campbell}, \citenamefont {Chiaro},
  \citenamefont {Dunsworth}, \citenamefont {Jeffrey}, \citenamefont {Kelly},
  \citenamefont {Mutus}, \citenamefont {O'Malley}, \citenamefont {Quintana},
  \citenamefont {Sank}, \citenamefont {Vainsencher}, \citenamefont {Wenner},
  \citenamefont {White}, \citenamefont {Polkovnikov},\ and\ \citenamefont
  {Martinis}}]{superconductive}%
  \BibitemOpen
  \bibfield  {author} {\bibinfo {author} {\bibfnamefont {C.}~\bibnamefont
  {Neill}}, \bibinfo {author} {\bibfnamefont {P.}~\bibnamefont {Roushan}},
  \bibinfo {author} {\bibfnamefont {M.}~\bibnamefont {Fang}}, \bibinfo {author}
  {\bibfnamefont {Y.}~\bibnamefont {Chen}}, \bibinfo {author} {\bibfnamefont
  {M.}~\bibnamefont {Kolodrubetz}}, \bibinfo {author} {\bibfnamefont
  {Z.}~\bibnamefont {Chen}}, \bibinfo {author} {\bibfnamefont {A.}~\bibnamefont
  {Megrant}}, \bibinfo {author} {\bibfnamefont {R.}~\bibnamefont {Barends}},
  \bibinfo {author} {\bibfnamefont {B.}~\bibnamefont {Campbell}}, \bibinfo
  {author} {\bibfnamefont {B.}~\bibnamefont {Chiaro}}, \bibinfo {author}
  {\bibfnamefont {A.}~\bibnamefont {Dunsworth}}, \bibinfo {author}
  {\bibfnamefont {E.}~\bibnamefont {Jeffrey}}, \bibinfo {author} {\bibfnamefont
  {J.}~\bibnamefont {Kelly}}, \bibinfo {author} {\bibfnamefont
  {J.}~\bibnamefont {Mutus}}, \bibinfo {author} {\bibfnamefont {P.~J.~J.}\
  \bibnamefont {O'Malley}}, \bibinfo {author} {\bibfnamefont {C.}~\bibnamefont
  {Quintana}}, \bibinfo {author} {\bibfnamefont {D.}~\bibnamefont {Sank}},
  \bibinfo {author} {\bibfnamefont {A.}~\bibnamefont {Vainsencher}}, \bibinfo
  {author} {\bibfnamefont {J.}~\bibnamefont {Wenner}}, \bibinfo {author}
  {\bibfnamefont {T.~C.}\ \bibnamefont {White}}, \bibinfo {author}
  {\bibfnamefont {A.}~\bibnamefont {Polkovnikov}},\ and\ \bibinfo {author}
  {\bibfnamefont {J.~M.}\ \bibnamefont {Martinis}},\ }\bibfield  {title}
  {\bibinfo {title} {{Ergodic dynamics and thermalization in an isolated
  quantum system}},\ }\href {https://doi.org/10.1038/nphys3830} {\bibfield
  {journal} {\bibinfo  {journal} {Nat. Phys.}\ }\textbf {\bibinfo {volume}
  {12}},\ \bibinfo {pages} {1037} (\bibinfo {year} {2016})}\BibitemShut
  {NoStop}%
\bibitem [{\citenamefont {Smith}\ \emph {et~al.}(2016)\citenamefont {Smith},
  \citenamefont {Lee}, \citenamefont {Richerme}, \citenamefont {Neyenhuis},
  \citenamefont {Hess}, \citenamefont {Hauke}, \citenamefont {Heyl},
  \citenamefont {Huse},\ and\ \citenamefont {Monroe}}]{trapped}%
  \BibitemOpen
  \bibfield  {author} {\bibinfo {author} {\bibfnamefont {J.}~\bibnamefont
  {Smith}}, \bibinfo {author} {\bibfnamefont {A.}~\bibnamefont {Lee}}, \bibinfo
  {author} {\bibfnamefont {P.}~\bibnamefont {Richerme}}, \bibinfo {author}
  {\bibfnamefont {B.}~\bibnamefont {Neyenhuis}}, \bibinfo {author}
  {\bibfnamefont {P.~W.}\ \bibnamefont {Hess}}, \bibinfo {author}
  {\bibfnamefont {P.}~\bibnamefont {Hauke}}, \bibinfo {author} {\bibfnamefont
  {M.}~\bibnamefont {Heyl}}, \bibinfo {author} {\bibfnamefont {D.~A.}\
  \bibnamefont {Huse}},\ and\ \bibinfo {author} {\bibfnamefont
  {C.}~\bibnamefont {Monroe}},\ }\bibfield  {title} {\bibinfo {title}
  {{Many-body localization in a quantum simulator with programmable random
  disorder}},\ }\href {https://doi.org/10.1038/nphys3783} {\bibfield  {journal}
  {\bibinfo  {journal} {Nat. Phys.}\ }\textbf {\bibinfo {volume} {12}},\
  \bibinfo {pages} {907} (\bibinfo {year} {2016})}\BibitemShut {NoStop}%
\bibitem [{\citenamefont {Bernien}\ \emph {et~al.}(2017)\citenamefont
  {Bernien}, \citenamefont {Schwartz}, \citenamefont {Keesling}, \citenamefont
  {Levine}, \citenamefont {Omran}, \citenamefont {Pichler}, \citenamefont
  {Choi}, \citenamefont {Zibrov}, \citenamefont {Endres}, \citenamefont
  {Greiner}, \citenamefont {Vuleti{\'c}},\ and\ \citenamefont
  {Lukin}}]{Bernien}%
  \BibitemOpen
  \bibfield  {author} {\bibinfo {author} {\bibfnamefont {H.}~\bibnamefont
  {Bernien}}, \bibinfo {author} {\bibfnamefont {S.}~\bibnamefont {Schwartz}},
  \bibinfo {author} {\bibfnamefont {A.}~\bibnamefont {Keesling}}, \bibinfo
  {author} {\bibfnamefont {H.}~\bibnamefont {Levine}}, \bibinfo {author}
  {\bibfnamefont {A.}~\bibnamefont {Omran}}, \bibinfo {author} {\bibfnamefont
  {H.}~\bibnamefont {Pichler}}, \bibinfo {author} {\bibfnamefont
  {S.}~\bibnamefont {Choi}}, \bibinfo {author} {\bibfnamefont {A.~S.}\
  \bibnamefont {Zibrov}}, \bibinfo {author} {\bibfnamefont {M.}~\bibnamefont
  {Endres}}, \bibinfo {author} {\bibfnamefont {M.}~\bibnamefont {Greiner}},
  \bibinfo {author} {\bibfnamefont {V.}~\bibnamefont {Vuleti{\'c}}},\ and\
  \bibinfo {author} {\bibfnamefont {M.~D.}\ \bibnamefont {Lukin}},\ }\bibfield
  {title} {\bibinfo {title} {{Probing many-body dynamics on a 51-atom quantum
  simulator}},\ }\href {https://doi.org/10.1038/nature24622} {\bibfield
  {journal} {\bibinfo  {journal} {Nature}\ }\textbf {\bibinfo {volume} {551}},\
  \bibinfo {pages} {579} (\bibinfo {year} {2017})}\BibitemShut {NoStop}%
\bibitem [{\citenamefont {von Neumann}(2010)}]{von2010proof}%
  \BibitemOpen
  \bibfield  {author} {\bibinfo {author} {\bibfnamefont {J.}~\bibnamefont {von
  Neumann}},\ }\bibfield  {title} {\bibinfo {title} {{Proof of the ergodic
  theorem and the H-theorem in quantum mechanics}},\ }\href
  {https://doi.org/10.1140/epjh/e2010-00008-5} {\bibfield  {journal} {\bibinfo
  {journal} {Eur. Phys. J. H}\ }\textbf {\bibinfo {volume} {35}},\ \bibinfo
  {pages} {201} (\bibinfo {year} {2010})},\ \bibinfo {note} {[English
  translation of (by R. Tumulka) \href{https://doi.org/10.1007/BF01339852}{Z.
  Phys. {\bf 57}, 30 (1929)}]}\BibitemShut {NoStop}%
\bibitem [{\citenamefont {Goldstein}\ \emph {et~al.}(2010)\citenamefont
  {Goldstein}, \citenamefont {Lebowitz}, \citenamefont {Tumulka},\ and\
  \citenamefont {Zanghi}}]{goldstein2010long}%
  \BibitemOpen
  \bibfield  {author} {\bibinfo {author} {\bibfnamefont {S.}~\bibnamefont
  {Goldstein}}, \bibinfo {author} {\bibfnamefont {J.~L.}\ \bibnamefont
  {Lebowitz}}, \bibinfo {author} {\bibfnamefont {R.}~\bibnamefont {Tumulka}},\
  and\ \bibinfo {author} {\bibfnamefont {N.}~\bibnamefont {Zanghi}},\
  }\bibfield  {title} {\bibinfo {title} {{Long-time behavior of macroscopic
  quantum systems: Commentary accompanying the English translation of John von
  Neumann’s 1929 article on the quantum ergodic theorem}},\ }\href
  {https://doi.org/10.1140/epjh/e2010-00007-7} {\bibfield  {journal} {\bibinfo
  {journal} {Eur. Phys. J. H}\ }\textbf {\bibinfo {volume} {35}},\ \bibinfo
  {pages} {173} (\bibinfo {year} {2010})}\BibitemShut {NoStop}%
\bibitem [{\citenamefont {Rigol}\ and\ \citenamefont
  {Srednicki}(2012)}]{rigol2012alternatives}%
  \BibitemOpen
  \bibfield  {author} {\bibinfo {author} {\bibfnamefont {M.}~\bibnamefont
  {Rigol}}\ and\ \bibinfo {author} {\bibfnamefont {M.}~\bibnamefont
  {Srednicki}},\ }\bibfield  {title} {\bibinfo {title} {{Alternatives to
  eigenstate thermalization}},\ }\href
  {https://doi.org/10.1103/PhysRevLett.108.110601} {\bibfield  {journal}
  {\bibinfo  {journal} {Phys. Rev. Lett.}\ }\textbf {\bibinfo {volume} {108}},\
  \bibinfo {pages} {110601} (\bibinfo {year} {2012})}\BibitemShut {NoStop}%
\bibitem [{\citenamefont {Tasaki}(1998)}]{tasaki1998quantum}%
  \BibitemOpen
  \bibfield  {author} {\bibinfo {author} {\bibfnamefont {H.}~\bibnamefont
  {Tasaki}},\ }\bibfield  {title} {\bibinfo {title} {{From quantum dynamics to
  the canonical distribution: general picture and a rigorous example}},\ }\href
  {https://doi.org/10.1103/PhysRevLett.80.1373} {\bibfield  {journal} {\bibinfo
   {journal} {Phys. Rev. Lett.}\ }\textbf {\bibinfo {volume} {80}},\ \bibinfo
  {pages} {1373} (\bibinfo {year} {1998})}\BibitemShut {NoStop}%
\bibitem [{\citenamefont {Deutsch}(1991)}]{deutsch1991quantum}%
  \BibitemOpen
  \bibfield  {author} {\bibinfo {author} {\bibfnamefont {J.~M.}\ \bibnamefont
  {Deutsch}},\ }\bibfield  {title} {\bibinfo {title} {{Quantum statistical
  mechanics in a closed system}},\ }\href
  {https://doi.org/10.1103/PhysRevA.43.2046} {\bibfield  {journal} {\bibinfo
  {journal} {Phys. Rev. A}\ }\textbf {\bibinfo {volume} {43}},\ \bibinfo
  {pages} {2046} (\bibinfo {year} {1991})}\BibitemShut {NoStop}%
\bibitem [{\citenamefont {Srednicki}(1994)}]{srednicki1994chaos}%
  \BibitemOpen
  \bibfield  {author} {\bibinfo {author} {\bibfnamefont {M.}~\bibnamefont
  {Srednicki}},\ }\bibfield  {title} {\bibinfo {title} {{Chaos and quantum
  thermalization}},\ }\href {https://doi.org/10.1103/PhysRevE.50.888}
  {\bibfield  {journal} {\bibinfo  {journal} {Phys. Rev. E}\ }\textbf {\bibinfo
  {volume} {50}},\ \bibinfo {pages} {888} (\bibinfo {year} {1994})}\BibitemShut
  {NoStop}%
\bibitem [{\citenamefont {Horoi}\ \emph {et~al.}(1995)\citenamefont {Horoi},
  \citenamefont {Zelevinsky},\ and\ \citenamefont {Brown}}]{horoi1995chaos}%
  \BibitemOpen
  \bibfield  {author} {\bibinfo {author} {\bibfnamefont {M.}~\bibnamefont
  {Horoi}}, \bibinfo {author} {\bibfnamefont {V.}~\bibnamefont {Zelevinsky}},\
  and\ \bibinfo {author} {\bibfnamefont {B.~A.}\ \bibnamefont {Brown}},\
  }\bibfield  {title} {\bibinfo {title} {{Chaos vs thermalization in the
  nuclear shell model}},\ }\href {https://doi.org/10.1103/PhysRevLett.74.5194}
  {\bibfield  {journal} {\bibinfo  {journal} {Phys. Rev. Lett.}\ }\textbf
  {\bibinfo {volume} {74}},\ \bibinfo {pages} {5194} (\bibinfo {year}
  {1995})}\BibitemShut {NoStop}%
\bibitem [{\citenamefont {Zelevinsky}\ \emph {et~al.}(1996)\citenamefont
  {Zelevinsky}, \citenamefont {Brown}, \citenamefont {Frazier},\ and\
  \citenamefont {Horoi}}]{zelevinsky1996nuclear}%
  \BibitemOpen
  \bibfield  {author} {\bibinfo {author} {\bibfnamefont {V.}~\bibnamefont
  {Zelevinsky}}, \bibinfo {author} {\bibfnamefont {B.~A.}\ \bibnamefont
  {Brown}}, \bibinfo {author} {\bibfnamefont {N.}~\bibnamefont {Frazier}},\
  and\ \bibinfo {author} {\bibfnamefont {M.}~\bibnamefont {Horoi}},\ }\bibfield
   {title} {\bibinfo {title} {{The nuclear shell model as a testing ground for
  many-body quantum chaos}},\ }\href
  {https://doi.org/10.1016/S0370-1573(96)00007-5} {\bibfield  {journal}
  {\bibinfo  {journal} {Phys. Rep.}\ }\textbf {\bibinfo {volume} {276}},\
  \bibinfo {pages} {85} (\bibinfo {year} {1996})}\BibitemShut {NoStop}%
\bibitem [{\citenamefont {Venuti}\ and\ \citenamefont
  {Liu}(2019)}]{venuti2019ergodicity}%
  \BibitemOpen
  \bibfield  {author} {\bibinfo {author} {\bibfnamefont {L.~C.}\ \bibnamefont
  {Venuti}}\ and\ \bibinfo {author} {\bibfnamefont {L.}~\bibnamefont {Liu}},\
  }\bibfield  {title} {\bibinfo {title} {{Ergodicity, eigenstate
  thermalization, and the foundations of statistical mechanics in quantum and
  classical systems}},\ }\href {https://doi.org/10.48550/arXiv.1904.02336}
  {\bibfield  {journal} {\bibinfo  {journal} {arXiv:1904.02336}\ } (\bibinfo
  {year} {2019})}\BibitemShut {NoStop}%
\bibitem [{\citenamefont {Deutsch}(2018)}]{deutsch2018eigenstate}%
  \BibitemOpen
  \bibfield  {author} {\bibinfo {author} {\bibfnamefont {J.~M.}\ \bibnamefont
  {Deutsch}},\ }\bibfield  {title} {\bibinfo {title} {{Eigenstate
  thermalization hypothesis}},\ }\href
  {https://doi.org/10.1088/1361-6633/aac9f1} {\bibfield  {journal} {\bibinfo
  {journal} {Rep. Prog. Phys.}\ }\textbf {\bibinfo {volume} {81}},\ \bibinfo
  {pages} {082001} (\bibinfo {year} {2018})}\BibitemShut {NoStop}%
\bibitem [{Note1()}]{Note1}%
  \BibitemOpen
  \bibinfo {note} {In contrast, the weak ETH claims that almost all energy
  eigenstates are thermal, which allows a small number of exceptional
  eigenstates called nonthermal states. Actually, the weak ETH has been proved
  in some cases~\cite {biroli2010effect, iyoda2017fluctuation}.}\BibitemShut
  {Stop}%
\bibitem [{\citenamefont {Rigol}\ \emph {et~al.}(2008)\citenamefont {Rigol},
  \citenamefont {Dunjko},\ and\ \citenamefont
  {Olshanii}}]{rigol2008thermalization}%
  \BibitemOpen
  \bibfield  {author} {\bibinfo {author} {\bibfnamefont {M.}~\bibnamefont
  {Rigol}}, \bibinfo {author} {\bibfnamefont {V.}~\bibnamefont {Dunjko}},\ and\
  \bibinfo {author} {\bibfnamefont {M.}~\bibnamefont {Olshanii}},\ }\bibfield
  {title} {\bibinfo {title} {{Thermalization and its mechanism for generic
  isolated quantum systems}},\ }\href {https://doi.org/10.1038/nature06838}
  {\bibfield  {journal} {\bibinfo  {journal} {Nature}\ }\textbf {\bibinfo
  {volume} {452}},\ \bibinfo {pages} {854} (\bibinfo {year}
  {2008})}\BibitemShut {NoStop}%
\bibitem [{\citenamefont {Polkovnikov}\ \emph {et~al.}(2011)\citenamefont
  {Polkovnikov}, \citenamefont {Sengupta}, \citenamefont {Silva},\ and\
  \citenamefont {Vengalattore}}]{polkovnikov2011colloquium}%
  \BibitemOpen
  \bibfield  {author} {\bibinfo {author} {\bibfnamefont {A.}~\bibnamefont
  {Polkovnikov}}, \bibinfo {author} {\bibfnamefont {K.}~\bibnamefont
  {Sengupta}}, \bibinfo {author} {\bibfnamefont {A.}~\bibnamefont {Silva}},\
  and\ \bibinfo {author} {\bibfnamefont {M.}~\bibnamefont {Vengalattore}},\
  }\bibfield  {title} {\bibinfo {title} {{Colloquium: Nonequilibrium dynamics
  of closed interacting quantum systems}},\ }\href
  {https://doi.org/10.1103/RevModPhys.83.863} {\bibfield  {journal} {\bibinfo
  {journal} {Rev. Mod. Phys.}\ }\textbf {\bibinfo {volume} {83}},\ \bibinfo
  {pages} {863} (\bibinfo {year} {2011})}\BibitemShut {NoStop}%
\bibitem [{\citenamefont {Nandkishore}\ and\ \citenamefont
  {Huse}(2015)}]{nandkishore2015many}%
  \BibitemOpen
  \bibfield  {author} {\bibinfo {author} {\bibfnamefont {R.}~\bibnamefont
  {Nandkishore}}\ and\ \bibinfo {author} {\bibfnamefont {D.~A.}\ \bibnamefont
  {Huse}},\ }\bibfield  {title} {\bibinfo {title} {{Many-body localization and
  thermalization in quantum statistical mechanics}},\ }\href
  {https://doi.org/10.1146/annurev-conmatphys-031214-014726} {\bibfield
  {journal} {\bibinfo  {journal} {Annu. Rev. Condens. Matter Phys.}\ }\textbf
  {\bibinfo {volume} {6}},\ \bibinfo {pages} {15} (\bibinfo {year}
  {2015})}\BibitemShut {NoStop}%
\bibitem [{Note2()}]{Note2}%
  \BibitemOpen
  \bibinfo {note} {Unfortunately, it was reported that there is no general
  theorem, algorithm, or systematic procedure to determine whether any given
  quantum many-body system thermalizes or not~\cite
  {shiraishi2021undecidability}.}\BibitemShut {Stop}%
\bibitem [{\citenamefont {Serbyn}\ \emph {et~al.}(2021)\citenamefont {Serbyn},
  \citenamefont {Abanin},\ and\ \citenamefont {Papi{\'c}}}]{serbyn2021quantum}%
  \BibitemOpen
  \bibfield  {author} {\bibinfo {author} {\bibfnamefont {M.}~\bibnamefont
  {Serbyn}}, \bibinfo {author} {\bibfnamefont {D.~A.}\ \bibnamefont {Abanin}},\
  and\ \bibinfo {author} {\bibfnamefont {Z.}~\bibnamefont {Papi{\'c}}},\
  }\bibfield  {title} {\bibinfo {title} {{Quantum many-body scars and weak
  breaking of ergodicity}},\ }\href
  {https://doi.org/10.1038/s41567-021-01230-2} {\bibfield  {journal} {\bibinfo
  {journal} {Nat. Phys.}\ }\textbf {\bibinfo {volume} {17}},\ \bibinfo {pages}
  {675} (\bibinfo {year} {2021})}\BibitemShut {NoStop}%
\bibitem [{\citenamefont {Moudgalya}\ \emph {et~al.}(2022)\citenamefont
  {Moudgalya}, \citenamefont {Regnault},\ and\ \citenamefont
  {Bernevig}}]{regnault2022quantum}%
  \BibitemOpen
  \bibfield  {author} {\bibinfo {author} {\bibfnamefont {S.}~\bibnamefont
  {Moudgalya}}, \bibinfo {author} {\bibfnamefont {N.}~\bibnamefont
  {Regnault}},\ and\ \bibinfo {author} {\bibfnamefont {B.~A.}\ \bibnamefont
  {Bernevig}},\ }\bibfield  {title} {\bibinfo {title} {{Quantum many-body scars
  and hilbert space fragmentation: a review of exact results}},\ }\href
  {https://doi.org/10.1088/1361-6633/ac73a0} {\bibfield  {journal} {\bibinfo
  {journal} {Rep. Prog. Phys.}\ }\textbf {\bibinfo {volume} {85}},\ \bibinfo
  {pages} {086501} (\bibinfo {year} {2022})}\BibitemShut {NoStop}%
\bibitem [{\citenamefont {Chandran}\ \emph {et~al.}(2023)\citenamefont
  {Chandran}, \citenamefont {Iadecola}, \citenamefont {Khemani},\ and\
  \citenamefont {Moessner}}]{chandran2023quantum}%
  \BibitemOpen
  \bibfield  {author} {\bibinfo {author} {\bibfnamefont {A.}~\bibnamefont
  {Chandran}}, \bibinfo {author} {\bibfnamefont {T.}~\bibnamefont {Iadecola}},
  \bibinfo {author} {\bibfnamefont {V.}~\bibnamefont {Khemani}},\ and\ \bibinfo
  {author} {\bibfnamefont {R.}~\bibnamefont {Moessner}},\ }\bibfield  {title}
  {\bibinfo {title} {{Quantum many-body scars: A quasiparticle perspective}},\
  }\href {https://doi.org/10.1146/annurev-conmatphys-031620-101617} {\bibfield
  {journal} {\bibinfo  {journal} {Annu. Rev. Condens. Matter Phys.}\ }\textbf
  {\bibinfo {volume} {14}},\ \bibinfo {pages} {443} (\bibinfo {year}
  {2023})}\BibitemShut {NoStop}%
\bibitem [{\citenamefont {Su}\ \emph {et~al.}(2023)\citenamefont {Su},
  \citenamefont {Sun}, \citenamefont {Hudomal}, \citenamefont {Desaules},
  \citenamefont {Zhou}, \citenamefont {Yang}, \citenamefont {Halimeh},
  \citenamefont {Yuan}, \citenamefont {Papi{\'c}},\ and\ \citenamefont
  {Pan}}]{su2023observation}%
  \BibitemOpen
  \bibfield  {author} {\bibinfo {author} {\bibfnamefont {G.-X.}\ \bibnamefont
  {Su}}, \bibinfo {author} {\bibfnamefont {H.}~\bibnamefont {Sun}}, \bibinfo
  {author} {\bibfnamefont {A.}~\bibnamefont {Hudomal}}, \bibinfo {author}
  {\bibfnamefont {J.-Y.}\ \bibnamefont {Desaules}}, \bibinfo {author}
  {\bibfnamefont {Z.-Y.}\ \bibnamefont {Zhou}}, \bibinfo {author}
  {\bibfnamefont {B.}~\bibnamefont {Yang}}, \bibinfo {author} {\bibfnamefont
  {J.~C.}\ \bibnamefont {Halimeh}}, \bibinfo {author} {\bibfnamefont {Z.-S.}\
  \bibnamefont {Yuan}}, \bibinfo {author} {\bibfnamefont {Z.}~\bibnamefont
  {Papi{\'c}}},\ and\ \bibinfo {author} {\bibfnamefont {J.-W.}\ \bibnamefont
  {Pan}},\ }\bibfield  {title} {\bibinfo {title} {{Observation of many-body
  scarring in a Bose-Hubbard quantum simulator}},\ }\href
  {https://doi.org/10.1103/PhysRevResearch.5.023010} {\bibfield  {journal}
  {\bibinfo  {journal} {Phys. Rev. Res.}\ }\textbf {\bibinfo {volume} {5}},\
  \bibinfo {pages} {023010} (\bibinfo {year} {2023})}\BibitemShut {NoStop}%
\bibitem [{\citenamefont {Zhang}\ \emph {et~al.}(2023)\citenamefont {Zhang},
  \citenamefont {Dong}, \citenamefont {Gao}, \citenamefont {Zhao},
  \citenamefont {Hao}, \citenamefont {Desaules}, \citenamefont {Guo},
  \citenamefont {Chen}, \citenamefont {Deng}, \citenamefont {Liu} \emph
  {et~al.}}]{zhang2023many}%
  \BibitemOpen
  \bibfield  {author} {\bibinfo {author} {\bibfnamefont {P.}~\bibnamefont
  {Zhang}}, \bibinfo {author} {\bibfnamefont {H.}~\bibnamefont {Dong}},
  \bibinfo {author} {\bibfnamefont {Y.}~\bibnamefont {Gao}}, \bibinfo {author}
  {\bibfnamefont {L.}~\bibnamefont {Zhao}}, \bibinfo {author} {\bibfnamefont
  {J.}~\bibnamefont {Hao}}, \bibinfo {author} {\bibfnamefont {J.-Y.}\
  \bibnamefont {Desaules}}, \bibinfo {author} {\bibfnamefont {Q.}~\bibnamefont
  {Guo}}, \bibinfo {author} {\bibfnamefont {J.}~\bibnamefont {Chen}}, \bibinfo
  {author} {\bibfnamefont {J.}~\bibnamefont {Deng}}, \bibinfo {author}
  {\bibfnamefont {B.}~\bibnamefont {Liu}}, \emph {et~al.},\ }\bibfield  {title}
  {\bibinfo {title} {{Many-body Hilbert space scarring on a superconducting
  processor}},\ }\href {https://doi.org/10.1038/s41567-022-01784-9} {\bibfield
  {journal} {\bibinfo  {journal} {Nat. Phys.}\ }\textbf {\bibinfo {volume}
  {19}},\ \bibinfo {pages} {120} (\bibinfo {year} {2023})}\BibitemShut
  {NoStop}%
\bibitem [{\citenamefont {Hudomal}\ \emph {et~al.}(2020)\citenamefont
  {Hudomal}, \citenamefont {Vasi{\'c}}, \citenamefont {Regnault},\ and\
  \citenamefont {Papi{\'c}}}]{hudomal2020quantum}%
  \BibitemOpen
  \bibfield  {author} {\bibinfo {author} {\bibfnamefont {A.}~\bibnamefont
  {Hudomal}}, \bibinfo {author} {\bibfnamefont {I.}~\bibnamefont {Vasi{\'c}}},
  \bibinfo {author} {\bibfnamefont {N.}~\bibnamefont {Regnault}},\ and\
  \bibinfo {author} {\bibfnamefont {Z.}~\bibnamefont {Papi{\'c}}},\ }\bibfield
  {title} {\bibinfo {title} {{Quantum scars of bosons with correlated
  hopping}},\ }\href {https://doi.org/10.1038/s42005-020-0364-9} {\bibfield
  {journal} {\bibinfo  {journal} {Commun. Phys.}\ }\textbf {\bibinfo {volume}
  {3}},\ \bibinfo {pages} {99} (\bibinfo {year} {2020})}\BibitemShut {NoStop}%
\bibitem [{\citenamefont {Zhao}\ \emph {et~al.}(2020)\citenamefont {Zhao},
  \citenamefont {Vovrosh}, \citenamefont {Mintert},\ and\ \citenamefont
  {Knolle}}]{zhao2020quantum}%
  \BibitemOpen
  \bibfield  {author} {\bibinfo {author} {\bibfnamefont {H.}~\bibnamefont
  {Zhao}}, \bibinfo {author} {\bibfnamefont {J.}~\bibnamefont {Vovrosh}},
  \bibinfo {author} {\bibfnamefont {F.}~\bibnamefont {Mintert}},\ and\ \bibinfo
  {author} {\bibfnamefont {J.}~\bibnamefont {Knolle}},\ }\bibfield  {title}
  {\bibinfo {title} {{Quantum many-body scars in optical lattices}},\ }\href
  {https://doi.org/10.1103/PhysRevLett.124.160604} {\bibfield  {journal}
  {\bibinfo  {journal} {Phys. Rev. Lett.}\ }\textbf {\bibinfo {volume} {124}},\
  \bibinfo {pages} {160604} (\bibinfo {year} {2020})}\BibitemShut {NoStop}%
\bibitem [{\citenamefont {Desaules}\ \emph {et~al.}(2021)\citenamefont
  {Desaules}, \citenamefont {Hudomal}, \citenamefont {Turner},\ and\
  \citenamefont {Papi{\'c}}}]{desaules2021proposal}%
  \BibitemOpen
  \bibfield  {author} {\bibinfo {author} {\bibfnamefont {J.-Y.}\ \bibnamefont
  {Desaules}}, \bibinfo {author} {\bibfnamefont {A.}~\bibnamefont {Hudomal}},
  \bibinfo {author} {\bibfnamefont {C.~J.}\ \bibnamefont {Turner}},\ and\
  \bibinfo {author} {\bibfnamefont {Z.}~\bibnamefont {Papi{\'c}}},\ }\bibfield
  {title} {\bibinfo {title} {{Proposal for realizing quantum scars in the
  tilted 1D Fermi-Hubbard model}},\ }\href
  {https://doi.org/10.1103/PhysRevLett.126.210601} {\bibfield  {journal}
  {\bibinfo  {journal} {Phys. Rev. Lett.}\ }\textbf {\bibinfo {volume} {126}},\
  \bibinfo {pages} {210601} (\bibinfo {year} {2021})}\BibitemShut {NoStop}%
\bibitem [{\citenamefont {Kunimi}\ \emph {et~al.}(2023)\citenamefont {Kunimi},
  \citenamefont {Tomita}, \citenamefont {Katsura},\ and\ \citenamefont
  {Kato}}]{kunimi2023proposal}%
  \BibitemOpen
  \bibfield  {author} {\bibinfo {author} {\bibfnamefont {M.}~\bibnamefont
  {Kunimi}}, \bibinfo {author} {\bibfnamefont {T.}~\bibnamefont {Tomita}},
  \bibinfo {author} {\bibfnamefont {H.}~\bibnamefont {Katsura}},\ and\ \bibinfo
  {author} {\bibfnamefont {Y.}~\bibnamefont {Kato}},\ }\bibfield  {title}
  {\bibinfo {title} {{Proposal for realizing quantum spin models with
  Dzyaloshinskii-Moriya interaction using Rydberg atoms}},\ }\href
  {https://doi.org/10.48550/arXiv.2306.05591} {\bibfield  {journal} {\bibinfo
  {journal} {arXiv:2306.05591}\ } (\bibinfo {year} {2023})}\BibitemShut
  {NoStop}%
\bibitem [{\citenamefont {Turner}\ \emph
  {et~al.}(2018{\natexlab{a}})\citenamefont {Turner}, \citenamefont
  {Michailidis}, \citenamefont {Abanin}, \citenamefont {Serbyn},\ and\
  \citenamefont {Papi{\'c}}}]{Turner}%
  \BibitemOpen
  \bibfield  {author} {\bibinfo {author} {\bibfnamefont {C.~J.}\ \bibnamefont
  {Turner}}, \bibinfo {author} {\bibfnamefont {A.~A.}\ \bibnamefont
  {Michailidis}}, \bibinfo {author} {\bibfnamefont {D.~A.}\ \bibnamefont
  {Abanin}}, \bibinfo {author} {\bibfnamefont {M.}~\bibnamefont {Serbyn}},\
  and\ \bibinfo {author} {\bibfnamefont {Z.}~\bibnamefont {Papi{\'c}}},\
  }\bibfield  {title} {\bibinfo {title} {{Weak ergodicity breaking from quantum
  many-body scars}},\ }\href {https://doi.org/10.1038/s41567-018-0137-5}
  {\bibfield  {journal} {\bibinfo  {journal} {Nat. Phys.}\ }\textbf {\bibinfo
  {volume} {14}},\ \bibinfo {pages} {745} (\bibinfo {year}
  {2018}{\natexlab{a}})}\BibitemShut {NoStop}%
\bibitem [{\citenamefont {Choi}\ \emph {et~al.}(2019)\citenamefont {Choi},
  \citenamefont {Turner}, \citenamefont {Pichler}, \citenamefont {Ho},
  \citenamefont {Michailidis}, \citenamefont {Papi{\'c}}, \citenamefont
  {Serbyn}, \citenamefont {Lukin},\ and\ \citenamefont
  {Abanin}}]{choi2019emergent}%
  \BibitemOpen
  \bibfield  {author} {\bibinfo {author} {\bibfnamefont {S.}~\bibnamefont
  {Choi}}, \bibinfo {author} {\bibfnamefont {C.~J.}\ \bibnamefont {Turner}},
  \bibinfo {author} {\bibfnamefont {H.}~\bibnamefont {Pichler}}, \bibinfo
  {author} {\bibfnamefont {W.~W.}\ \bibnamefont {Ho}}, \bibinfo {author}
  {\bibfnamefont {A.~A.}\ \bibnamefont {Michailidis}}, \bibinfo {author}
  {\bibfnamefont {Z.}~\bibnamefont {Papi{\'c}}}, \bibinfo {author}
  {\bibfnamefont {M.}~\bibnamefont {Serbyn}}, \bibinfo {author} {\bibfnamefont
  {M.~D.}\ \bibnamefont {Lukin}},\ and\ \bibinfo {author} {\bibfnamefont
  {D.~A.}\ \bibnamefont {Abanin}},\ }\bibfield  {title} {\bibinfo {title}
  {{Emergent SU (2) dynamics and perfect quantum many-body scars}},\ }\href
  {https://doi.org/10.1103/PhysRevLett.122.220603} {\bibfield  {journal}
  {\bibinfo  {journal} {Phys. Rev. Lett.}\ }\textbf {\bibinfo {volume} {122}},\
  \bibinfo {pages} {220603} (\bibinfo {year} {2019})}\BibitemShut {NoStop}%
\bibitem [{\citenamefont {Lin}\ and\ \citenamefont
  {Motrunich}(2019)}]{lin2019exact}%
  \BibitemOpen
  \bibfield  {author} {\bibinfo {author} {\bibfnamefont {C.-J.}\ \bibnamefont
  {Lin}}\ and\ \bibinfo {author} {\bibfnamefont {O.~I.}\ \bibnamefont
  {Motrunich}},\ }\bibfield  {title} {\bibinfo {title} {{Exact quantum
  many-body scar states in the Rydberg-blockaded atom chain}},\ }\href
  {https://doi.org/10.1103/PhysRevLett.122.173401} {\bibfield  {journal}
  {\bibinfo  {journal} {Phys. Rev. Lett.}\ }\textbf {\bibinfo {volume} {122}},\
  \bibinfo {pages} {173401} (\bibinfo {year} {2019})}\BibitemShut {NoStop}%
\bibitem [{\citenamefont {Shiraishi}(2019)}]{shiraishi2019connection}%
  \BibitemOpen
  \bibfield  {author} {\bibinfo {author} {\bibfnamefont {N.}~\bibnamefont
  {Shiraishi}},\ }\bibfield  {title} {\bibinfo {title} {{Connection between
  quantum-many-body scars and the Affleck--Kennedy--Lieb--Tasaki model from the
  viewpoint of embedded Hamiltonians}},\ }\href
  {https://doi.org/10.1088/1742-5468/ab342e} {\bibfield  {journal} {\bibinfo
  {journal} {J. Stat. Mech. Theory Exp.}\ }\textbf {\bibinfo {volume} {2019}},\
  \bibinfo {pages} {083103} (\bibinfo {year} {2019})}\BibitemShut {NoStop}%
\bibitem [{\citenamefont {Lin}\ \emph {et~al.}(2020{\natexlab{a}})\citenamefont
  {Lin}, \citenamefont {Calvera},\ and\ \citenamefont
  {Hsieh}}]{lin2020quantum}%
  \BibitemOpen
  \bibfield  {author} {\bibinfo {author} {\bibfnamefont {C.-J.}\ \bibnamefont
  {Lin}}, \bibinfo {author} {\bibfnamefont {V.}~\bibnamefont {Calvera}},\ and\
  \bibinfo {author} {\bibfnamefont {T.~H.}\ \bibnamefont {Hsieh}},\ }\bibfield
  {title} {\bibinfo {title} {{Quantum many-body scar states in two-dimensional
  Rydberg atom arrays}},\ }\href {https://doi.org/10.1103/PhysRevB.101.220304}
  {\bibfield  {journal} {\bibinfo  {journal} {Phys. Rev. B}\ }\textbf {\bibinfo
  {volume} {101}},\ \bibinfo {pages} {220304(R)} (\bibinfo {year}
  {2020}{\natexlab{a}})}\BibitemShut {NoStop}%
\bibitem [{\citenamefont {Moudgalya}\ \emph {et~al.}(2018)\citenamefont
  {Moudgalya}, \citenamefont {Regnault},\ and\ \citenamefont
  {Bernevig}}]{Moudgalya}%
  \BibitemOpen
  \bibfield  {author} {\bibinfo {author} {\bibfnamefont {S.}~\bibnamefont
  {Moudgalya}}, \bibinfo {author} {\bibfnamefont {N.}~\bibnamefont
  {Regnault}},\ and\ \bibinfo {author} {\bibfnamefont {B.~A.}\ \bibnamefont
  {Bernevig}},\ }\bibfield  {title} {\bibinfo {title} {{Entanglement of exact
  excited states of Affleck-Kennedy-Lieb-Tasaki models: Exact results,
  many-body scars, and violation of the strong eigenstate thermalization
  hypothesis}},\ }\href {https://doi.org/10.1103/PhysRevB.98.235156} {\bibfield
   {journal} {\bibinfo  {journal} {Phys. Rev. B}\ }\textbf {\bibinfo {volume}
  {98}},\ \bibinfo {pages} {235156} (\bibinfo {year} {2018})}\BibitemShut
  {NoStop}%
\bibitem [{\citenamefont {Mark}\ \emph {et~al.}(2020)\citenamefont {Mark},
  \citenamefont {Lin},\ and\ \citenamefont {Motrunich}}]{Mark}%
  \BibitemOpen
  \bibfield  {author} {\bibinfo {author} {\bibfnamefont {D.~K.}\ \bibnamefont
  {Mark}}, \bibinfo {author} {\bibfnamefont {C.-J.}\ \bibnamefont {Lin}},\ and\
  \bibinfo {author} {\bibfnamefont {O.~I.}\ \bibnamefont {Motrunich}},\
  }\bibfield  {title} {\bibinfo {title} {{Unified structure for exact towers of
  scar states in the Affleck-Kennedy-Lieb-Tasaki and other models}},\ }\href
  {https://doi.org/10.1103/PhysRevB.101.195131} {\bibfield  {journal} {\bibinfo
   {journal} {Phys. Rev. B}\ }\textbf {\bibinfo {volume} {101}},\ \bibinfo
  {pages} {195131} (\bibinfo {year} {2020})}\BibitemShut {NoStop}%
\bibitem [{\citenamefont {O'Dea}\ \emph {et~al.}(2020)\citenamefont {O'Dea},
  \citenamefont {Burnell}, \citenamefont {Chandran},\ and\ \citenamefont
  {Khemani}}]{o2020tunnels}%
  \BibitemOpen
  \bibfield  {author} {\bibinfo {author} {\bibfnamefont {N.}~\bibnamefont
  {O'Dea}}, \bibinfo {author} {\bibfnamefont {F.}~\bibnamefont {Burnell}},
  \bibinfo {author} {\bibfnamefont {A.}~\bibnamefont {Chandran}},\ and\
  \bibinfo {author} {\bibfnamefont {V.}~\bibnamefont {Khemani}},\ }\bibfield
  {title} {\bibinfo {title} {{From tunnels to towers: Quantum scars from Lie
  algebras and q-deformed Lie algebras}},\ }\href
  {https://doi.org/10.1103/PhysRevResearch.2.043305} {\bibfield  {journal}
  {\bibinfo  {journal} {Phys. Rev. Res.}\ }\textbf {\bibinfo {volume} {2}},\
  \bibinfo {pages} {043305} (\bibinfo {year} {2020})}\BibitemShut {NoStop}%
\bibitem [{\citenamefont {Schecter}\ and\ \citenamefont
  {Iadecola}(2019)}]{schecter2019weak}%
  \BibitemOpen
  \bibfield  {author} {\bibinfo {author} {\bibfnamefont {M.}~\bibnamefont
  {Schecter}}\ and\ \bibinfo {author} {\bibfnamefont {T.}~\bibnamefont
  {Iadecola}},\ }\bibfield  {title} {\bibinfo {title} {{Weak Ergodicity
  Breaking and Quantum Many-Body Scars in Spin-1 X Y Magnets}},\ }\href
  {https://doi.org/10.1103/PhysRevLett.123.147201} {\bibfield  {journal}
  {\bibinfo  {journal} {Phys. Rev. Lett.}\ }\textbf {\bibinfo {volume} {123}},\
  \bibinfo {pages} {147201} (\bibinfo {year} {2019})}\BibitemShut {NoStop}%
\bibitem [{\citenamefont {Iadecola}\ and\ \citenamefont
  {Schecter}(2020)}]{Iadecola}%
  \BibitemOpen
  \bibfield  {author} {\bibinfo {author} {\bibfnamefont {T.}~\bibnamefont
  {Iadecola}}\ and\ \bibinfo {author} {\bibfnamefont {M.}~\bibnamefont
  {Schecter}},\ }\bibfield  {title} {\bibinfo {title} {{Quantum many-body scar
  states with emergent kinetic constraints and finite-entanglement revivals}},\
  }\href {https://doi.org/10.1103/PhysRevB.101.024306} {\bibfield  {journal}
  {\bibinfo  {journal} {Phys. Rev. B}\ }\textbf {\bibinfo {volume} {101}},\
  \bibinfo {pages} {024306} (\bibinfo {year} {2020})}\BibitemShut {NoStop}%
\bibitem [{\citenamefont {Chattopadhyay}\ \emph {et~al.}(2020)\citenamefont
  {Chattopadhyay}, \citenamefont {Pichler}, \citenamefont {Lukin},\ and\
  \citenamefont {Ho}}]{Chattopadhyay}%
  \BibitemOpen
  \bibfield  {author} {\bibinfo {author} {\bibfnamefont {S.}~\bibnamefont
  {Chattopadhyay}}, \bibinfo {author} {\bibfnamefont {H.}~\bibnamefont
  {Pichler}}, \bibinfo {author} {\bibfnamefont {M.~D.}\ \bibnamefont {Lukin}},\
  and\ \bibinfo {author} {\bibfnamefont {W.~W.}\ \bibnamefont {Ho}},\
  }\bibfield  {title} {\bibinfo {title} {{Quantum many-body scars from virtual
  entangled pairs}},\ }\href {https://doi.org/10.1103/PhysRevB.101.174308}
  {\bibfield  {journal} {\bibinfo  {journal} {Phys. Rev. B}\ }\textbf {\bibinfo
  {volume} {101}},\ \bibinfo {pages} {174308} (\bibinfo {year}
  {2020})}\BibitemShut {NoStop}%
\bibitem [{\citenamefont {Moudgalya}\ \emph
  {et~al.}(2020{\natexlab{a}})\citenamefont {Moudgalya}, \citenamefont
  {O'Brien}, \citenamefont {Bernevig}, \citenamefont {Fendley},\ and\
  \citenamefont {Regnault}}]{moudgalya2020large}%
  \BibitemOpen
  \bibfield  {author} {\bibinfo {author} {\bibfnamefont {S.}~\bibnamefont
  {Moudgalya}}, \bibinfo {author} {\bibfnamefont {E.}~\bibnamefont {O'Brien}},
  \bibinfo {author} {\bibfnamefont {B.~A.}\ \bibnamefont {Bernevig}}, \bibinfo
  {author} {\bibfnamefont {P.}~\bibnamefont {Fendley}},\ and\ \bibinfo {author}
  {\bibfnamefont {N.}~\bibnamefont {Regnault}},\ }\bibfield  {title} {\bibinfo
  {title} {{Large classes of quantum scarred Hamiltonians from matrix product
  states}},\ }\href {https://doi.org/10.1103/PhysRevB.102.085120} {\bibfield
  {journal} {\bibinfo  {journal} {Phys. Rev. B}\ }\textbf {\bibinfo {volume}
  {102}},\ \bibinfo {pages} {085120} (\bibinfo {year}
  {2020}{\natexlab{a}})}\BibitemShut {NoStop}%
\bibitem [{\citenamefont {Shibata}\ \emph {et~al.}(2020)\citenamefont
  {Shibata}, \citenamefont {Yoshioka},\ and\ \citenamefont {Katsura}}]{SYK}%
  \BibitemOpen
  \bibfield  {author} {\bibinfo {author} {\bibfnamefont {N.}~\bibnamefont
  {Shibata}}, \bibinfo {author} {\bibfnamefont {N.}~\bibnamefont {Yoshioka}},\
  and\ \bibinfo {author} {\bibfnamefont {H.}~\bibnamefont {Katsura}},\
  }\bibfield  {title} {\bibinfo {title} {{Onsager's Scars in Disordered Spin
  Chains}},\ }\href {https://doi.org/10.1103/PhysRevLett.124.180604} {\bibfield
   {journal} {\bibinfo  {journal} {Phys. Rev. Lett.}\ }\textbf {\bibinfo
  {volume} {124}},\ \bibinfo {pages} {180604} (\bibinfo {year}
  {2020})}\BibitemShut {NoStop}%
\bibitem [{\citenamefont {Shiraishi}\ and\ \citenamefont
  {Mori}(2017)}]{Shiraishi-Mori}%
  \BibitemOpen
  \bibfield  {author} {\bibinfo {author} {\bibfnamefont {N.}~\bibnamefont
  {Shiraishi}}\ and\ \bibinfo {author} {\bibfnamefont {T.}~\bibnamefont
  {Mori}},\ }\bibfield  {title} {\bibinfo {title} {{Systematic Construction of
  Counterexamples to the Eigenstate Thermalization Hypothesis}},\ }\href
  {https://doi.org/10.1103/PhysRevLett.119.030601} {\bibfield  {journal}
  {\bibinfo  {journal} {Phys. Rev. Lett.}\ }\textbf {\bibinfo {volume} {119}},\
  \bibinfo {pages} {030601} (\bibinfo {year} {2017})}\BibitemShut {NoStop}%
\bibitem [{\citenamefont {McClarty}\ \emph {et~al.}(2020)\citenamefont
  {McClarty}, \citenamefont {Haque}, \citenamefont {Sen},\ and\ \citenamefont
  {Richter}}]{mcclarty2020disorder}%
  \BibitemOpen
  \bibfield  {author} {\bibinfo {author} {\bibfnamefont {P.~A.}\ \bibnamefont
  {McClarty}}, \bibinfo {author} {\bibfnamefont {M.}~\bibnamefont {Haque}},
  \bibinfo {author} {\bibfnamefont {A.}~\bibnamefont {Sen}},\ and\ \bibinfo
  {author} {\bibfnamefont {J.}~\bibnamefont {Richter}},\ }\bibfield  {title}
  {\bibinfo {title} {{Disorder-free localization and many-body quantum scars
  from magnetic frustration}},\ }\href
  {https://doi.org/10.1103/PhysRevB.102.224303} {\bibfield  {journal} {\bibinfo
   {journal} {Phys. Rev. B}\ }\textbf {\bibinfo {volume} {102}},\ \bibinfo
  {pages} {224303} (\bibinfo {year} {2020})}\BibitemShut {NoStop}%
\bibitem [{\citenamefont {Pakrouski}\ \emph {et~al.}(2020)\citenamefont
  {Pakrouski}, \citenamefont {Pallegar}, \citenamefont {Popov},\ and\
  \citenamefont {Klebanov}}]{pakrouski2020many}%
  \BibitemOpen
  \bibfield  {author} {\bibinfo {author} {\bibfnamefont {K.}~\bibnamefont
  {Pakrouski}}, \bibinfo {author} {\bibfnamefont {P.~N.}\ \bibnamefont
  {Pallegar}}, \bibinfo {author} {\bibfnamefont {F.~K.}\ \bibnamefont
  {Popov}},\ and\ \bibinfo {author} {\bibfnamefont {I.~R.}\ \bibnamefont
  {Klebanov}},\ }\bibfield  {title} {\bibinfo {title} {{Many-body scars as a
  group invariant sector of Hilbert space}},\ }\href
  {https://doi.org/10.1103/PhysRevLett.125.230602} {\bibfield  {journal}
  {\bibinfo  {journal} {Phys. Rev. Lett.}\ }\textbf {\bibinfo {volume} {125}},\
  \bibinfo {pages} {230602} (\bibinfo {year} {2020})}\BibitemShut {NoStop}%
\bibitem [{\citenamefont {Pakrouski}\ \emph {et~al.}(2021)\citenamefont
  {Pakrouski}, \citenamefont {Pallegar}, \citenamefont {Popov},\ and\
  \citenamefont {Klebanov}}]{pakrouski2021group}%
  \BibitemOpen
  \bibfield  {author} {\bibinfo {author} {\bibfnamefont {K.}~\bibnamefont
  {Pakrouski}}, \bibinfo {author} {\bibfnamefont {P.~N.}\ \bibnamefont
  {Pallegar}}, \bibinfo {author} {\bibfnamefont {F.~K.}\ \bibnamefont
  {Popov}},\ and\ \bibinfo {author} {\bibfnamefont {I.~R.}\ \bibnamefont
  {Klebanov}},\ }\bibfield  {title} {\bibinfo {title} {{Group theoretic
  approach to many-body scar states in fermionic lattice models}},\ }\href
  {https://doi.org/10.1103/PhysRevResearch.3.043156} {\bibfield  {journal}
  {\bibinfo  {journal} {Phys. Rev. Res.}\ }\textbf {\bibinfo {volume} {3}},\
  \bibinfo {pages} {043156} (\bibinfo {year} {2021})}\BibitemShut {NoStop}%
\bibitem [{\citenamefont {Ren}\ \emph {et~al.}(2021)\citenamefont {Ren},
  \citenamefont {Liang},\ and\ \citenamefont {Fang}}]{ren2021quasisymmetry}%
  \BibitemOpen
  \bibfield  {author} {\bibinfo {author} {\bibfnamefont {J.}~\bibnamefont
  {Ren}}, \bibinfo {author} {\bibfnamefont {C.}~\bibnamefont {Liang}},\ and\
  \bibinfo {author} {\bibfnamefont {C.}~\bibnamefont {Fang}},\ }\bibfield
  {title} {\bibinfo {title} {{Quasisymmetry groups and many-body scar
  dynamics}},\ }\href {https://doi.org/10.1103/PhysRevLett.126.120604}
  {\bibfield  {journal} {\bibinfo  {journal} {Phys. Rev. Lett.}\ }\textbf
  {\bibinfo {volume} {126}},\ \bibinfo {pages} {120604} (\bibinfo {year}
  {2021})}\BibitemShut {NoStop}%
\bibitem [{\citenamefont {Tang}\ \emph {et~al.}(2022)\citenamefont {Tang},
  \citenamefont {O'Dea},\ and\ \citenamefont {Chandran}}]{tang2021multi}%
  \BibitemOpen
  \bibfield  {author} {\bibinfo {author} {\bibfnamefont {L.-H.}\ \bibnamefont
  {Tang}}, \bibinfo {author} {\bibfnamefont {N.}~\bibnamefont {O'Dea}},\ and\
  \bibinfo {author} {\bibfnamefont {A.}~\bibnamefont {Chandran}},\ }\bibfield
  {title} {\bibinfo {title} {{Multimagnon quantum many-body scars from tensor
  operators}},\ }\href {https://doi.org/10.1103/PhysRevResearch.4.043006}
  {\bibfield  {journal} {\bibinfo  {journal} {Phys. Rev. Res.}\ }\textbf
  {\bibinfo {volume} {4}},\ \bibinfo {pages} {043006} (\bibinfo {year}
  {2022})}\BibitemShut {NoStop}%
\bibitem [{\citenamefont {Wildeboer}\ \emph {et~al.}(2022)\citenamefont
  {Wildeboer}, \citenamefont {Langlett}, \citenamefont {Yang}, \citenamefont
  {Gorshkov}, \citenamefont {Iadecola},\ and\ \citenamefont
  {Xu}}]{wildeboer2022quantum}%
  \BibitemOpen
  \bibfield  {author} {\bibinfo {author} {\bibfnamefont {J.}~\bibnamefont
  {Wildeboer}}, \bibinfo {author} {\bibfnamefont {C.~M.}\ \bibnamefont
  {Langlett}}, \bibinfo {author} {\bibfnamefont {Z.-C.}\ \bibnamefont {Yang}},
  \bibinfo {author} {\bibfnamefont {A.~V.}\ \bibnamefont {Gorshkov}}, \bibinfo
  {author} {\bibfnamefont {T.}~\bibnamefont {Iadecola}},\ and\ \bibinfo
  {author} {\bibfnamefont {S.}~\bibnamefont {Xu}},\ }\bibfield  {title}
  {\bibinfo {title} {{Quantum many-body scars from Einstein-Podolsky-Rosen
  states in bilayer systems}},\ }\href
  {https://doi.org/10.1103/PhysRevB.106.205142} {\bibfield  {journal} {\bibinfo
   {journal} {Phys. Rev. B}\ }\textbf {\bibinfo {volume} {106}},\ \bibinfo
  {pages} {205142} (\bibinfo {year} {2022})}\BibitemShut {NoStop}%
\bibitem [{\citenamefont {Ren}\ \emph {et~al.}(2022)\citenamefont {Ren},
  \citenamefont {Liang},\ and\ \citenamefont {Fang}}]{ren2022deformed}%
  \BibitemOpen
  \bibfield  {author} {\bibinfo {author} {\bibfnamefont {J.}~\bibnamefont
  {Ren}}, \bibinfo {author} {\bibfnamefont {C.}~\bibnamefont {Liang}},\ and\
  \bibinfo {author} {\bibfnamefont {C.}~\bibnamefont {Fang}},\ }\bibfield
  {title} {\bibinfo {title} {{Deformed symmetry structures and quantum
  many-body scar subspaces}},\ }\href
  {https://doi.org/10.1103/PhysRevResearch.4.013155} {\bibfield  {journal}
  {\bibinfo  {journal} {Phys. Rev. Res.}\ }\textbf {\bibinfo {volume} {4}},\
  \bibinfo {pages} {013155} (\bibinfo {year} {2022})}\BibitemShut {NoStop}%
\bibitem [{\citenamefont {Omiya}\ and\ \citenamefont
  {M{\"u}ller}(2023)}]{omiya2023fractionalization}%
  \BibitemOpen
  \bibfield  {author} {\bibinfo {author} {\bibfnamefont {K.}~\bibnamefont
  {Omiya}}\ and\ \bibinfo {author} {\bibfnamefont {M.}~\bibnamefont
  {M{\"u}ller}},\ }\bibfield  {title} {\bibinfo {title} {{Fractionalization
  paves the way to local projector embeddings of quantum many-body scars}},\
  }\href {https://doi.org/10.1103/PhysRevB.108.054412} {\bibfield  {journal}
  {\bibinfo  {journal} {Phys. Rev. B}\ }\textbf {\bibinfo {volume} {108}},\
  \bibinfo {pages} {054412} (\bibinfo {year} {2023})}\BibitemShut {NoStop}%
\bibitem [{\citenamefont {Turner}\ \emph
  {et~al.}(2018{\natexlab{b}})\citenamefont {Turner}, \citenamefont
  {Michailidis}, \citenamefont {Abanin}, \citenamefont {Serbyn},\ and\
  \citenamefont {Papi\ifmmode~\acute{c}\else \'{c}\fi{}}}]{Turner2}%
  \BibitemOpen
  \bibfield  {author} {\bibinfo {author} {\bibfnamefont {C.~J.}\ \bibnamefont
  {Turner}}, \bibinfo {author} {\bibfnamefont {A.~A.}\ \bibnamefont
  {Michailidis}}, \bibinfo {author} {\bibfnamefont {D.~A.}\ \bibnamefont
  {Abanin}}, \bibinfo {author} {\bibfnamefont {M.}~\bibnamefont {Serbyn}},\
  and\ \bibinfo {author} {\bibfnamefont {Z.}~\bibnamefont
  {Papi\ifmmode~\acute{c}\else \'{c}\fi{}}},\ }\bibfield  {title} {\bibinfo
  {title} {{Quantum scarred eigenstates in a Rydberg atom chain: Entanglement,
  breakdown of thermalization, and stability to perturbations}},\ }\href
  {https://doi.org/10.1103/PhysRevB.98.155134} {\bibfield  {journal} {\bibinfo
  {journal} {Phys. Rev. B}\ }\textbf {\bibinfo {volume} {98}},\ \bibinfo
  {pages} {155134} (\bibinfo {year} {2018}{\natexlab{b}})}\BibitemShut
  {NoStop}%
\bibitem [{\citenamefont {Lin}\ \emph {et~al.}(2020{\natexlab{b}})\citenamefont
  {Lin}, \citenamefont {Chandran},\ and\ \citenamefont
  {Motrunich}}]{lin2020slow}%
  \BibitemOpen
  \bibfield  {author} {\bibinfo {author} {\bibfnamefont {C.-J.}\ \bibnamefont
  {Lin}}, \bibinfo {author} {\bibfnamefont {A.}~\bibnamefont {Chandran}},\ and\
  \bibinfo {author} {\bibfnamefont {O.~I.}\ \bibnamefont {Motrunich}},\
  }\bibfield  {title} {\bibinfo {title} {{Slow thermalization of exact quantum
  many-body scar states under perturbations}},\ }\href
  {https://doi.org/10.1103/PhysRevResearch.2.033044} {\bibfield  {journal}
  {\bibinfo  {journal} {Phys. Rev. Res.}\ }\textbf {\bibinfo {volume} {2}},\
  \bibinfo {pages} {033044} (\bibinfo {year} {2020}{\natexlab{b}})}\BibitemShut
  {NoStop}%
\bibitem [{\citenamefont {Gotta}\ \emph {et~al.}(2023)\citenamefont {Gotta},
  \citenamefont {Moudgalya},\ and\ \citenamefont
  {Mazza}}]{gotta2023asymptotic}%
  \BibitemOpen
  \bibfield  {author} {\bibinfo {author} {\bibfnamefont {L.}~\bibnamefont
  {Gotta}}, \bibinfo {author} {\bibfnamefont {S.}~\bibnamefont {Moudgalya}},\
  and\ \bibinfo {author} {\bibfnamefont {L.}~\bibnamefont {Mazza}},\ }\bibfield
   {title} {\bibinfo {title} {{Asymptotic Quantum Many-Body Scars}},\ }\href
  {https://doi.org/10.48550/arXiv.2303.05407} {\bibfield  {journal} {\bibinfo
  {journal} {arXiv:2303.05407}\ } (\bibinfo {year} {2023})}\BibitemShut
  {NoStop}%
\bibitem [{\citenamefont {Moudgalya}\ and\ \citenamefont
  {Motrunich}(2022{\natexlab{a}})}]{moudgalya2022exhaustive}%
  \BibitemOpen
  \bibfield  {author} {\bibinfo {author} {\bibfnamefont {S.}~\bibnamefont
  {Moudgalya}}\ and\ \bibinfo {author} {\bibfnamefont {O.~I.}\ \bibnamefont
  {Motrunich}},\ }\bibfield  {title} {\bibinfo {title} {{Exhaustive
  Characterization of Quantum Many-Body Scars using Commutant Algebras}},\
  }\href {https://doi.org/10.48550/arXiv.2209.03377} {\bibfield  {journal}
  {\bibinfo  {journal} {arXiv:2209.03377}\ } (\bibinfo {year}
  {2022}{\natexlab{a}})}\BibitemShut {NoStop}%
\bibitem [{\citenamefont {Moudgalya}\ and\ \citenamefont
  {Motrunich}(2022{\natexlab{b}})}]{moudgalya2022hilbert}%
  \BibitemOpen
  \bibfield  {author} {\bibinfo {author} {\bibfnamefont {S.}~\bibnamefont
  {Moudgalya}}\ and\ \bibinfo {author} {\bibfnamefont {O.~I.}\ \bibnamefont
  {Motrunich}},\ }\bibfield  {title} {\bibinfo {title} {{Hilbert space
  fragmentation and commutant algebras}},\ }\href
  {https://doi.org/10.1103/PhysRevX.12.011050} {\bibfield  {journal} {\bibinfo
  {journal} {Phys. Rev. X}\ }\textbf {\bibinfo {volume} {12}},\ \bibinfo
  {pages} {011050} (\bibinfo {year} {2022}{\natexlab{b}})}\BibitemShut
  {NoStop}%
\bibitem [{\citenamefont {de~Leeuw}\ \emph {et~al.}(2015)\citenamefont
  {de~Leeuw}, \citenamefont {Kristjansen},\ and\ \citenamefont
  {Zarembo}}]{de2015one}%
  \BibitemOpen
  \bibfield  {author} {\bibinfo {author} {\bibfnamefont {M.}~\bibnamefont
  {de~Leeuw}}, \bibinfo {author} {\bibfnamefont {C.}~\bibnamefont
  {Kristjansen}},\ and\ \bibinfo {author} {\bibfnamefont {K.}~\bibnamefont
  {Zarembo}},\ }\bibfield  {title} {\bibinfo {title} {{One-point functions in
  defect CFT and integrability}},\ }\href
  {https://doi.org/10.1007/JHEP08(2015)098} {\bibfield  {journal} {\bibinfo
  {journal} {J. High Energ. Phys.}\ }\textbf {\bibinfo {volume} {2015}}\bibinfo
   {number} { (8)},\ \bibinfo {pages} {98}}\BibitemShut {NoStop}%
\bibitem [{\citenamefont {Piroli}\ \emph {et~al.}(2017)\citenamefont {Piroli},
  \citenamefont {Pozsgay},\ and\ \citenamefont {Vernier}}]{Piroli2017WhatIA}%
  \BibitemOpen
\bibfield  {number} {  }\bibfield  {author} {\bibinfo {author} {\bibfnamefont
  {L.}~\bibnamefont {Piroli}}, \bibinfo {author} {\bibfnamefont
  {B.}~\bibnamefont {Pozsgay}},\ and\ \bibinfo {author} {\bibfnamefont
  {E.}~\bibnamefont {Vernier}},\ }\bibfield  {title} {\bibinfo {title} {{What
  is an integrable quench}},\ }\href
  {https://doi.org/10.1016/j.nuclphysb.2017.10.012} {\bibfield  {journal}
  {\bibinfo  {journal} {Nucl. Phys.}\ }\textbf {\bibinfo {volume} {925}},\
  \bibinfo {pages} {362} (\bibinfo {year} {2017})}\BibitemShut {NoStop}%
\bibitem [{\citenamefont {De~Leeuw}\ \emph {et~al.}(2018)\citenamefont
  {De~Leeuw}, \citenamefont {Kristjansen},\ and\ \citenamefont
  {Linardopoulos}}]{de2018scalar}%
  \BibitemOpen
  \bibfield  {author} {\bibinfo {author} {\bibfnamefont {M.}~\bibnamefont
  {De~Leeuw}}, \bibinfo {author} {\bibfnamefont {C.}~\bibnamefont
  {Kristjansen}},\ and\ \bibinfo {author} {\bibfnamefont {G.}~\bibnamefont
  {Linardopoulos}},\ }\bibfield  {title} {\bibinfo {title} {{Scalar one-point
  functions and matrix product states of AdS/dCFT}},\ }\href
  {https://doi.org/10.1016/j.physletb.2018.03.083} {\bibfield  {journal}
  {\bibinfo  {journal} {Phys. Lett. B}\ }\textbf {\bibinfo {volume} {781}},\
  \bibinfo {pages} {238} (\bibinfo {year} {2018})}\BibitemShut {NoStop}%
\bibitem [{\citenamefont {Pozsgay}(2018)}]{Pozsgay}%
  \BibitemOpen
  \bibfield  {author} {\bibinfo {author} {\bibfnamefont {B.}~\bibnamefont
  {Pozsgay}},\ }\bibfield  {title} {\bibinfo {title} {{Overlaps with arbitrary
  two-site states in the XXZ spin chain}},\ }\href
  {https://doi.org/10.1088/1742-5468/aabbe1} {\bibfield  {journal} {\bibinfo
  {journal} {J. Stat. Mech.: Theory Exp.}\ }\textbf {\bibinfo {volume}
  {2018}}\bibinfo  {number} { (5)},\ \bibinfo {pages} {053103}}\BibitemShut
  {NoStop}%
\bibitem [{\citenamefont {Piroli}\ \emph {et~al.}(2019)\citenamefont {Piroli},
  \citenamefont {Vernier}, \citenamefont {Calabrese},\ and\ \citenamefont
  {Pozsgay}}]{piroli2019integrable}%
  \BibitemOpen
\bibfield  {number} {  }\bibfield  {author} {\bibinfo {author} {\bibfnamefont
  {L.}~\bibnamefont {Piroli}}, \bibinfo {author} {\bibfnamefont
  {E.}~\bibnamefont {Vernier}}, \bibinfo {author} {\bibfnamefont
  {P.}~\bibnamefont {Calabrese}},\ and\ \bibinfo {author} {\bibfnamefont
  {B.}~\bibnamefont {Pozsgay}},\ }\bibfield  {title} {\bibinfo {title}
  {{Integrable quenches in nested spin chains I: the exact steady states}},\
  }\href {https://doi.org/10.1088/1742-5468/ab1c51} {\bibfield  {journal}
  {\bibinfo  {journal} {J. Stat. Mech.: Theory Exp.}\ }\textbf {\bibinfo
  {volume} {2019}}\bibinfo  {number} { (6)},\ \bibinfo {pages}
  {063103}}\BibitemShut {NoStop}%
\bibitem [{\citenamefont {Pozsgay}\ \emph {et~al.}(2019)\citenamefont
  {Pozsgay}, \citenamefont {Piroli},\ and\ \citenamefont
  {Vernier}}]{pozsgay2019integrable}%
  \BibitemOpen
\bibfield  {number} {  }\bibfield  {author} {\bibinfo {author} {\bibfnamefont
  {B.}~\bibnamefont {Pozsgay}}, \bibinfo {author} {\bibfnamefont
  {L.}~\bibnamefont {Piroli}},\ and\ \bibinfo {author} {\bibfnamefont
  {E.}~\bibnamefont {Vernier}},\ }\bibfield  {title} {\bibinfo {title}
  {{Integrable matrix product states from boundary integrability}},\ }\href
  {https://doi.org/10.21468/SciPostPhys.6.5.062} {\bibfield  {journal}
  {\bibinfo  {journal} {SciPost Phys.}\ }\textbf {\bibinfo {volume} {6}},\
  \bibinfo {pages} {062} (\bibinfo {year} {2019})}\BibitemShut {NoStop}%
\bibitem [{\citenamefont {Moudgalya}\ \emph
  {et~al.}(2020{\natexlab{b}})\citenamefont {Moudgalya}, \citenamefont
  {Regnault},\ and\ \citenamefont {Bernevig}}]{moudgalya2020eta}%
  \BibitemOpen
  \bibfield  {author} {\bibinfo {author} {\bibfnamefont {S.}~\bibnamefont
  {Moudgalya}}, \bibinfo {author} {\bibfnamefont {N.}~\bibnamefont
  {Regnault}},\ and\ \bibinfo {author} {\bibfnamefont {B.~A.}\ \bibnamefont
  {Bernevig}},\ }\bibfield  {title} {\bibinfo {title} {{$\eta$-pairing in
  Hubbard models: From spectrum generating algebras to quantum many-body
  scars}},\ }\href {https://doi.org/10.1103/PhysRevB.102.085140} {\bibfield
  {journal} {\bibinfo  {journal} {Phys. Rev. B}\ }\textbf {\bibinfo {volume}
  {102}},\ \bibinfo {pages} {085140} (\bibinfo {year}
  {2020}{\natexlab{b}})}\BibitemShut {NoStop}%
\bibitem [{\citenamefont {Bu{\v{c}}a}\ \emph {et~al.}(2019)\citenamefont
  {Bu{\v{c}}a}, \citenamefont {Tindall},\ and\ \citenamefont
  {Jaksch}}]{buvca2019non}%
  \BibitemOpen
  \bibfield  {author} {\bibinfo {author} {\bibfnamefont {B.}~\bibnamefont
  {Bu{\v{c}}a}}, \bibinfo {author} {\bibfnamefont {J.}~\bibnamefont
  {Tindall}},\ and\ \bibinfo {author} {\bibfnamefont {D.}~\bibnamefont
  {Jaksch}},\ }\bibfield  {title} {\bibinfo {title} {{Non-stationary coherent
  quantum many-body dynamics through dissipation}},\ }\href
  {https://doi.org/10.1038/s41467-019-09757-y} {\bibfield  {journal} {\bibinfo
  {journal} {Nat. Commun.}\ }\textbf {\bibinfo {volume} {10}},\ \bibinfo
  {pages} {1730} (\bibinfo {year} {2019})}\BibitemShut {NoStop}%
\bibitem [{\citenamefont {Sklyanin}(1992)}]{sklyanin1992quantum}%
  \BibitemOpen
  \bibfield  {author} {\bibinfo {author} {\bibfnamefont {E.}~\bibnamefont
  {Sklyanin}},\ }\bibfield  {title} {\bibinfo {title} {{Quantum inverse
  scattering method. Selected topics}},\ }\href
  {https://arxiv.org/abs/hep-th/9211111} {\bibfield  {journal} {\bibinfo
  {journal} {arXiv:hep-th/9211111}\ } (\bibinfo {year} {1992})}\BibitemShut
  {NoStop}%
\bibitem [{\citenamefont {Grabowski}\ and\ \citenamefont
  {Mathieu}(1995{\natexlab{a}})}]{GRABOWSKI1995299}%
  \BibitemOpen
  \bibfield  {author} {\bibinfo {author} {\bibfnamefont {M.}~\bibnamefont
  {Grabowski}}\ and\ \bibinfo {author} {\bibfnamefont {P.}~\bibnamefont
  {Mathieu}},\ }\bibfield  {title} {\bibinfo {title} {{Structure of the
  Conservation Laws in Quantum Integrable Spin Chains with Short Range
  Interactions}},\ }\href {https://doi.org/10.1006/aphy.1995.1101} {\bibfield
  {journal} {\bibinfo  {journal} {Ann. Phys.}\ }\textbf {\bibinfo {volume}
  {243}},\ \bibinfo {pages} {299} (\bibinfo {year}
  {1995}{\natexlab{a}})}\BibitemShut {NoStop}%
\bibitem [{\citenamefont {De~Leeuw}\ \emph {et~al.}(2019)\citenamefont
  {De~Leeuw}, \citenamefont {Pribytok},\ and\ \citenamefont
  {Ryan}}]{de2019classifying}%
  \BibitemOpen
  \bibfield  {author} {\bibinfo {author} {\bibfnamefont {M.}~\bibnamefont
  {De~Leeuw}}, \bibinfo {author} {\bibfnamefont {A.}~\bibnamefont {Pribytok}},\
  and\ \bibinfo {author} {\bibfnamefont {P.}~\bibnamefont {Ryan}},\ }\bibfield
  {title} {\bibinfo {title} {{Classifying integrable spin-1/2 chains with
  nearest neighbour interactions}},\ }\href
  {https://doi.org/10.1088/1751-8121/ab529f} {\bibfield  {journal} {\bibinfo
  {journal} {J. Phys. A: Math. Theor.}\ }\textbf {\bibinfo {volume} {52}},\
  \bibinfo {pages} {505201} (\bibinfo {year} {2019})}\BibitemShut {NoStop}%
\bibitem [{Note3()}]{Note3}%
  \BibitemOpen
  \bibinfo {note} {The boost operator also works for non-difference-form R
  matrices, see Ref.~\cite {de2019classifying}.}\BibitemShut {Stop}%
\bibitem [{\citenamefont {Berry}\ and\ \citenamefont
  {Tabor}(1977)}]{Berry_1977}%
  \BibitemOpen
  \bibfield  {author} {\bibinfo {author} {\bibfnamefont {M.~V.}\ \bibnamefont
  {Berry}}\ and\ \bibinfo {author} {\bibfnamefont {M.}~\bibnamefont {Tabor}},\
  }\bibfield  {title} {\bibinfo {title} {{Level clustering in the regular
  spectrum}},\ }\href {https://doi.org/10.1098/rspa.1977.0140} {\bibfield
  {journal} {\bibinfo  {journal} {Proc. R. Soc. London A}\ }\textbf {\bibinfo
  {volume} {356}},\ \bibinfo {pages} {375} (\bibinfo {year}
  {1977})}\BibitemShut {NoStop}%
\bibitem [{\citenamefont {Berry}(1981)}]{Berry_1981}%
  \BibitemOpen
  \bibfield  {author} {\bibinfo {author} {\bibfnamefont {M.~V.}\ \bibnamefont
  {Berry}},\ }\bibfield  {title} {\bibinfo {title} {{Quantizing a classically
  ergodic system: Sinai's billiard and the KKR method}},\ }\href
  {https://doi.org/10.1016/0003-4916(81)90189-5} {\bibfield  {journal}
  {\bibinfo  {journal} {Ann. Phys.}\ }\textbf {\bibinfo {volume} {131}},\
  \bibinfo {pages} {163} (\bibinfo {year} {1981})}\BibitemShut {NoStop}%
\bibitem [{\citenamefont {Bohigas}\ \emph {et~al.}(1984)\citenamefont
  {Bohigas}, \citenamefont {Giannoni},\ and\ \citenamefont
  {Schmit}}]{Bohigas_1984}%
  \BibitemOpen
  \bibfield  {author} {\bibinfo {author} {\bibfnamefont {O.}~\bibnamefont
  {Bohigas}}, \bibinfo {author} {\bibfnamefont {M.~J.}\ \bibnamefont
  {Giannoni}},\ and\ \bibinfo {author} {\bibfnamefont {C.}~\bibnamefont
  {Schmit}},\ }\bibfield  {title} {\bibinfo {title} {{Characterization of
  Chaotic Quantum Spectra and Universality of Level Fluctuation Laws}},\ }\href
  {https://doi.org/10.1103/PhysRevLett.52.1} {\bibfield  {journal} {\bibinfo
  {journal} {Phys. Rev. Lett.}\ }\textbf {\bibinfo {volume} {52}},\ \bibinfo
  {pages} {1} (\bibinfo {year} {1984})}\BibitemShut {NoStop}%
\bibitem [{\citenamefont {Szász-Schagrin}\ \emph {et~al.}(2021)\citenamefont
  {Szász-Schagrin}, \citenamefont {Pozsgay},\ and\ \citenamefont
  {Takács}}]{Szasz-Schagrin}%
  \BibitemOpen
  \bibfield  {author} {\bibinfo {author} {\bibfnamefont {D.}~\bibnamefont
  {Szász-Schagrin}}, \bibinfo {author} {\bibfnamefont {B.}~\bibnamefont
  {Pozsgay}},\ and\ \bibinfo {author} {\bibfnamefont {G.}~\bibnamefont
  {Takács}},\ }\bibfield  {title} {\bibinfo {title} {{Weak integrability
  breaking and level spacing distribution}},\ }\href
  {https://doi.org/10.21468/SciPostPhys.11.2.037} {\bibfield  {journal}
  {\bibinfo  {journal} {SciPost Phys.}\ }\textbf {\bibinfo {volume} {11}},\
  \bibinfo {pages} {37} (\bibinfo {year} {2021})}\BibitemShut {NoStop}%
\bibitem [{\citenamefont {Gaudin}(1961)}]{GAUDIN1961447}%
  \BibitemOpen
  \bibfield  {author} {\bibinfo {author} {\bibfnamefont {M.}~\bibnamefont
  {Gaudin}},\ }\bibfield  {title} {\bibinfo {title} {{Sur la loi limite de
  l'espacement des valeurs propres d'une matrice ale´atoire}},\ }\href
  {https://doi.org/10.1016/0029-5582(61)90176-6} {\bibfield  {journal}
  {\bibinfo  {journal} {Nucl. Phys.}\ }\textbf {\bibinfo {volume} {25}},\
  \bibinfo {pages} {447} (\bibinfo {year} {1961})}\BibitemShut {NoStop}%
\bibitem [{\citenamefont {Oganesyan}\ and\ \citenamefont
  {Huse}(2007)}]{oganesyan2007localization}%
  \BibitemOpen
  \bibfield  {author} {\bibinfo {author} {\bibfnamefont {V.}~\bibnamefont
  {Oganesyan}}\ and\ \bibinfo {author} {\bibfnamefont {D.~A.}\ \bibnamefont
  {Huse}},\ }\bibfield  {title} {\bibinfo {title} {{Localization of interacting
  fermions at high temperature}},\ }\href
  {https://doi.org/10.1103/PhysRevB.75.155111} {\bibfield  {journal} {\bibinfo
  {journal} {Phys. Rev. B}\ }\textbf {\bibinfo {volume} {75}},\ \bibinfo
  {pages} {155111} (\bibinfo {year} {2007})}\BibitemShut {NoStop}%
\bibitem [{\citenamefont {Atas}\ \emph {et~al.}(2013)\citenamefont {Atas},
  \citenamefont {Bogomolny}, \citenamefont {Giraud},\ and\ \citenamefont
  {Roux}}]{atas2013distribution}%
  \BibitemOpen
  \bibfield  {author} {\bibinfo {author} {\bibfnamefont {Y.~Y.}\ \bibnamefont
  {Atas}}, \bibinfo {author} {\bibfnamefont {E.}~\bibnamefont {Bogomolny}},
  \bibinfo {author} {\bibfnamefont {O.}~\bibnamefont {Giraud}},\ and\ \bibinfo
  {author} {\bibfnamefont {G.}~\bibnamefont {Roux}},\ }\bibfield  {title}
  {\bibinfo {title} {{Distribution of the Ratio of Consecutive Level Spacings
  in Random Matrix Ensembles}},\ }\href
  {https://doi.org/10.1103/PhysRevLett.110.084101} {\bibfield  {journal}
  {\bibinfo  {journal} {Phys. Rev. Lett.}\ }\textbf {\bibinfo {volume} {110}},\
  \bibinfo {pages} {084101} (\bibinfo {year} {2013})}\BibitemShut {NoStop}%
\bibitem [{\citenamefont {Mori}\ \emph {et~al.}(2018)\citenamefont {Mori},
  \citenamefont {Ikeda}, \citenamefont {Kaminishi},\ and\ \citenamefont
  {Ueda}}]{mori2018thermalization}%
  \BibitemOpen
  \bibfield  {author} {\bibinfo {author} {\bibfnamefont {T.}~\bibnamefont
  {Mori}}, \bibinfo {author} {\bibfnamefont {T.~N.}\ \bibnamefont {Ikeda}},
  \bibinfo {author} {\bibfnamefont {E.}~\bibnamefont {Kaminishi}},\ and\
  \bibinfo {author} {\bibfnamefont {M.}~\bibnamefont {Ueda}},\ }\bibfield
  {title} {\bibinfo {title} {{Thermalization and prethermalization in isolated
  quantum systems: a theoretical overview}},\ }\href
  {https://doi.org/10.1088/1361-6455/aabcdf} {\bibfield  {journal} {\bibinfo
  {journal} {J. Phys. B: At. Mol. Opt. Phys.}\ }\textbf {\bibinfo {volume}
  {51}},\ \bibinfo {pages} {112001} (\bibinfo {year} {2018})}\BibitemShut
  {NoStop}%
\bibitem [{\citenamefont {Majumdar}\ and\ \citenamefont
  {Ghosh}(1969{\natexlab{a}})}]{Majumdar}%
  \BibitemOpen
  \bibfield  {author} {\bibinfo {author} {\bibfnamefont {C.~K.}\ \bibnamefont
  {Majumdar}}\ and\ \bibinfo {author} {\bibfnamefont {D.~K.}\ \bibnamefont
  {Ghosh}},\ }\bibfield  {title} {\bibinfo {title} {{On
  Next‐Nearest‐Neighbor Interaction in Linear Chain. I}},\ }\href
  {https://doi.org/10.1063/1.1664978} {\bibfield  {journal} {\bibinfo
  {journal} {J. Math. Phys.}\ }\textbf {\bibinfo {volume} {10}},\ \bibinfo
  {pages} {1388} (\bibinfo {year} {1969}{\natexlab{a}})}\BibitemShut {NoStop}%
\bibitem [{\citenamefont {Majumdar}\ and\ \citenamefont
  {Ghosh}(1969{\natexlab{b}})}]{Majumdar2}%
  \BibitemOpen
  \bibfield  {author} {\bibinfo {author} {\bibfnamefont {C.~K.}\ \bibnamefont
  {Majumdar}}\ and\ \bibinfo {author} {\bibfnamefont {D.~K.}\ \bibnamefont
  {Ghosh}},\ }\bibfield  {title} {\bibinfo {title} {{On
  Next‐Nearest‐Neighbor Interaction in Linear Chain. II}},\ }\href
  {https://doi.org/10.1063/1.1664979} {\bibfield  {journal} {\bibinfo
  {journal} {J. Math. Phys.}\ }\textbf {\bibinfo {volume} {10}},\ \bibinfo
  {pages} {1399} (\bibinfo {year} {1969}{\natexlab{b}})}\BibitemShut {NoStop}%
\bibitem [{\citenamefont {Caspers}\ \emph {et~al.}(1984)\citenamefont
  {Caspers}, \citenamefont {Emmett},\ and\ \citenamefont {Magnus}}]{Caspers}%
  \BibitemOpen
  \bibfield  {author} {\bibinfo {author} {\bibfnamefont {W.~J.}\ \bibnamefont
  {Caspers}}, \bibinfo {author} {\bibfnamefont {K.~M.}\ \bibnamefont
  {Emmett}},\ and\ \bibinfo {author} {\bibfnamefont {W.}~\bibnamefont
  {Magnus}},\ }\bibfield  {title} {\bibinfo {title} {{The Majumdar-Ghosh chain.
  Twofold ground state and elementary excitations}},\ }\href
  {https://doi.org/10.1088/0305-4470/17/13/020} {\bibfield  {journal} {\bibinfo
   {journal} {J. Phys. A: Math. Gen.}\ }\textbf {\bibinfo {volume} {17}},\
  \bibinfo {pages} {2687} (\bibinfo {year} {1984})}\BibitemShut {NoStop}%
\bibitem [{\citenamefont {Wen}\ \emph {et~al.}(1989)\citenamefont {Wen},
  \citenamefont {Wilczek},\ and\ \citenamefont {Zee}}]{Wen}%
  \BibitemOpen
  \bibfield  {author} {\bibinfo {author} {\bibfnamefont {X.~G.}\ \bibnamefont
  {Wen}}, \bibinfo {author} {\bibfnamefont {F.}~\bibnamefont {Wilczek}},\ and\
  \bibinfo {author} {\bibfnamefont {A.}~\bibnamefont {Zee}},\ }\bibfield
  {title} {\bibinfo {title} {{Chiral spin states and superconductivity}},\
  }\href {https://doi.org/10.1103/PhysRevB.39.11413} {\bibfield  {journal}
  {\bibinfo  {journal} {Phys. Rev. B}\ }\textbf {\bibinfo {volume} {39}},\
  \bibinfo {pages} {11413} (\bibinfo {year} {1989})}\BibitemShut {NoStop}%
\bibitem [{\citenamefont {Frahm}\ and\ \citenamefont
  {R{\"o}denbeck}(1997)}]{Frahm}%
  \BibitemOpen
  \bibfield  {author} {\bibinfo {author} {\bibfnamefont {H.}~\bibnamefont
  {Frahm}}\ and\ \bibinfo {author} {\bibfnamefont {C.}~\bibnamefont
  {R{\"o}denbeck}},\ }\bibfield  {title} {\bibinfo {title} {{Properties of the
  chiral spin liquid state in generalized spin ladders}},\ }\href
  {https://doi.org/10.1088/0305-4470/30/13/005} {\bibfield  {journal} {\bibinfo
   {journal} {J. Phys. A: Math. Gen.}\ }\textbf {\bibinfo {volume} {30}},\
  \bibinfo {pages} {4467} (\bibinfo {year} {1997})}\BibitemShut {NoStop}%
\bibitem [{\citenamefont {Sen}\ and\ \citenamefont {Chitra}(1995)}]{SU2scalar}%
  \BibitemOpen
  \bibfield  {author} {\bibinfo {author} {\bibfnamefont {D.}~\bibnamefont
  {Sen}}\ and\ \bibinfo {author} {\bibfnamefont {R.}~\bibnamefont {Chitra}},\
  }\bibfield  {title} {\bibinfo {title} {{Large-U limit of a Hubbard model in a
  magnetic field: Chiral spin interactions and paramagnetism}},\ }\href
  {https://doi.org/10.1103/PhysRevB.51.1922} {\bibfield  {journal} {\bibinfo
  {journal} {Phys. Rev. B}\ }\textbf {\bibinfo {volume} {51}},\ \bibinfo
  {pages} {1922} (\bibinfo {year} {1995})}\BibitemShut {NoStop}%
\bibitem [{\citenamefont {Kim}\ \emph {et~al.}(2023)\citenamefont {Kim},
  \citenamefont {Yang}, \citenamefont {M{\o}lmer},\ and\ \citenamefont
  {Ahn}}]{kim2023realization}%
  \BibitemOpen
  \bibfield  {author} {\bibinfo {author} {\bibfnamefont {K.}~\bibnamefont
  {Kim}}, \bibinfo {author} {\bibfnamefont {F.}~\bibnamefont {Yang}}, \bibinfo
  {author} {\bibfnamefont {K.}~\bibnamefont {M{\o}lmer}},\ and\ \bibinfo
  {author} {\bibfnamefont {J.}~\bibnamefont {Ahn}},\ }\bibfield  {title}
  {\bibinfo {title} {{Realization of an extremely anisotropic Heisenberg magnet
  in Rydberg atom arrays}},\ }\href {https://arxiv.org/abs/2307.04342}
  {\bibfield  {journal} {\bibinfo  {journal} {arXiv:2307.04342}\ } (\bibinfo
  {year} {2023})}\BibitemShut {NoStop}%
\bibitem [{Note4()}]{Note4}%
  \BibitemOpen
  \bibinfo {note} {The group $\{ 1, {\protect \cal F}\}$ is a discrete subgroup
  of ${\protect \rm SU}(2)$, where ${\protect \cal F}$ corresponds to a $\pi $
  rotation around the $x$ axis.}\BibitemShut {Stop}%
\bibitem [{\citenamefont {Ramkarthik}\ \emph {et~al.}(2013)\citenamefont
  {Ramkarthik}, \citenamefont {Chandra},\ and\ \citenamefont
  {Lakshminarayan}}]{Ramkarthik}%
  \BibitemOpen
  \bibfield  {author} {\bibinfo {author} {\bibfnamefont {M.~S.}\ \bibnamefont
  {Ramkarthik}}, \bibinfo {author} {\bibfnamefont {V.~R.}\ \bibnamefont
  {Chandra}},\ and\ \bibinfo {author} {\bibfnamefont {A.}~\bibnamefont
  {Lakshminarayan}},\ }\bibfield  {title} {\bibinfo {title} {{Entanglement
  signatures for the dimerization transition in the Majumdar-Ghosh model}},\
  }\href {https://doi.org/10.1103/PhysRevA.87.012302} {\bibfield  {journal}
  {\bibinfo  {journal} {Phys. Rev. A}\ }\textbf {\bibinfo {volume} {87}},\
  \bibinfo {pages} {012302} (\bibinfo {year} {2013})}\BibitemShut {NoStop}%
\bibitem [{\citenamefont {Popkov}\ and\ \citenamefont
  {Salerno}(2005)}]{popkov2005logarithmic}%
  \BibitemOpen
  \bibfield  {author} {\bibinfo {author} {\bibfnamefont {V.}~\bibnamefont
  {Popkov}}\ and\ \bibinfo {author} {\bibfnamefont {M.}~\bibnamefont
  {Salerno}},\ }\bibfield  {title} {\bibinfo {title} {{Logarithmic divergence
  of the block entanglement entropy for the ferromagnetic Heisenberg model}},\
  }\href {https://doi.org/10.1103/PhysRevA.71.012301} {\bibfield  {journal}
  {\bibinfo  {journal} {Phys. Rev. A}\ }\textbf {\bibinfo {volume} {71}},\
  \bibinfo {pages} {012301} (\bibinfo {year} {2005})}\BibitemShut {NoStop}%
\bibitem [{\citenamefont {Grabowski}\ and\ \citenamefont
  {Mathieu}(1994)}]{grabowski1994quantum}%
  \BibitemOpen
  \bibfield  {author} {\bibinfo {author} {\bibfnamefont {M.~P.}\ \bibnamefont
  {Grabowski}}\ and\ \bibinfo {author} {\bibfnamefont {P.}~\bibnamefont
  {Mathieu}},\ }\bibfield  {title} {\bibinfo {title} {{Quantum integrals of
  motion for the Heisenberg spin chain}},\ }\href
  {https://doi.org/10.1142/S0217732394002057} {\bibfield  {journal} {\bibinfo
  {journal} {Mod. Phys. Lett. A}\ }\textbf {\bibinfo {volume} {9}},\ \bibinfo
  {pages} {2197} (\bibinfo {year} {1994})}\BibitemShut {NoStop}%
\bibitem [{\citenamefont {Affleck}\ \emph {et~al.}(1987)\citenamefont
  {Affleck}, \citenamefont {Kennedy}, \citenamefont {Lieb},\ and\ \citenamefont
  {Tasaki}}]{Affleck}%
  \BibitemOpen
  \bibfield  {author} {\bibinfo {author} {\bibfnamefont {I.}~\bibnamefont
  {Affleck}}, \bibinfo {author} {\bibfnamefont {T.}~\bibnamefont {Kennedy}},
  \bibinfo {author} {\bibfnamefont {E.~H.}\ \bibnamefont {Lieb}},\ and\
  \bibinfo {author} {\bibfnamefont {H.}~\bibnamefont {Tasaki}},\ }\bibfield
  {title} {\bibinfo {title} {{Rigorous results on valence-bond ground states in
  antiferromagnets}},\ }\href {https://doi.org/10.1103/PhysRevLett.59.799}
  {\bibfield  {journal} {\bibinfo  {journal} {Phys. Rev. Lett.}\ }\textbf
  {\bibinfo {volume} {59}},\ \bibinfo {pages} {799} (\bibinfo {year}
  {1987})}\BibitemShut {NoStop}%
\bibitem [{\citenamefont {Affleck}\ \emph {et~al.}(1988)\citenamefont
  {Affleck}, \citenamefont {Kennedy}, \citenamefont {Lieb},\ and\ \citenamefont
  {Tasaki}}]{Affleck2}%
  \BibitemOpen
  \bibfield  {author} {\bibinfo {author} {\bibfnamefont {I.}~\bibnamefont
  {Affleck}}, \bibinfo {author} {\bibfnamefont {T.}~\bibnamefont {Kennedy}},
  \bibinfo {author} {\bibfnamefont {E.~H.}\ \bibnamefont {Lieb}},\ and\
  \bibinfo {author} {\bibfnamefont {H.}~\bibnamefont {Tasaki}},\ }\bibfield
  {title} {\bibinfo {title} {{Valence bond ground states in isotropic quantum
  antiferromagnets}},\ }\href {https://doi.org/cmp/1104161001} {\bibfield
  {journal} {\bibinfo  {journal} {Commun. Math. Phys.}\ }\textbf {\bibinfo
  {volume} {115}},\ \bibinfo {pages} {477 } (\bibinfo {year}
  {1988})}\BibitemShut {NoStop}%
\bibitem [{\citenamefont {Tasaki}(2020)}]{book_Tasaki}%
  \BibitemOpen
  \bibfield  {author} {\bibinfo {author} {\bibfnamefont {H.}~\bibnamefont
  {Tasaki}},\ }\href {https://doi.org/10.1007/978-3-030-41265-4} {\emph
  {\bibinfo {title} {{Physics and Mathematics of Quantum Many-Body Systems}}}}\
  (\bibinfo  {publisher} {Springer},\ \bibinfo {year} {2020})\BibitemShut
  {NoStop}%
\bibitem [{\citenamefont {Sutherland}(1975)}]{Sutherland}%
  \BibitemOpen
  \bibfield  {author} {\bibinfo {author} {\bibfnamefont {B.}~\bibnamefont
  {Sutherland}},\ }\bibfield  {title} {\bibinfo {title} {{Model for a
  multicomponent quantum system}},\ }\href
  {https://doi.org/10.1103/PhysRevB.12.3795} {\bibfield  {journal} {\bibinfo
  {journal} {Phys. Rev. B}\ }\textbf {\bibinfo {volume} {12}},\ \bibinfo
  {pages} {3795} (\bibinfo {year} {1975})}\BibitemShut {NoStop}%
\bibitem [{\citenamefont {Lai}(1974)}]{Lai_1974}%
  \BibitemOpen
  \bibfield  {author} {\bibinfo {author} {\bibfnamefont {C.~K.}\ \bibnamefont
  {Lai}},\ }\bibfield  {title} {\bibinfo {title} {{Lattice gas with
  nearest-neighbor interaction in one dimension with arbitrary statistics}},\
  }\href {https://doi.org/10.1063/1.1666522} {\bibfield  {journal} {\bibinfo
  {journal} {J. Math. Phys.}\ }\textbf {\bibinfo {volume} {15}},\ \bibinfo
  {pages} {1675} (\bibinfo {year} {1974})}\BibitemShut {NoStop}%
\bibitem [{\citenamefont {Uimin}(1970)}]{Uimin_1970}%
  \BibitemOpen
  \bibfield  {author} {\bibinfo {author} {\bibfnamefont {G.~V.}\ \bibnamefont
  {Uimin}},\ }\bibfield  {title} {\bibinfo {title} {{One-dimensional problem
  for S= 1 with modified antiferromagnetic Hamiltonian}},\ }\href@noop {}
  {\bibfield  {journal} {\bibinfo  {journal} {Zh. Eksp. Teor. Fiz. Pis. Red.}\
  }\textbf {\bibinfo {volume} {12}},\ \bibinfo {pages} {332} (\bibinfo {year}
  {1970})},\ \bibinfo {note}
  {[\href{http://jetpletters.ru/ps/1730/article_26296.shtml}{JETP Lett. {\bf
  12}, 225 (1970)}]}\BibitemShut {NoStop}%
\bibitem [{\citenamefont {Grabowski}\ and\ \citenamefont
  {Mathieu}(1995{\natexlab{b}})}]{Grabowski}%
  \BibitemOpen
  \bibfield  {author} {\bibinfo {author} {\bibfnamefont {M.~P.}\ \bibnamefont
  {Grabowski}}\ and\ \bibinfo {author} {\bibfnamefont {P.}~\bibnamefont
  {Mathieu}},\ }\bibfield  {title} {\bibinfo {title} {{Integrability test for
  spin chains}},\ }\href {https://doi.org/10.1088/0305-4470/28/17/013}
  {\bibfield  {journal} {\bibinfo  {journal} {J. Phys. A: Math. Gen.}\ }\textbf
  {\bibinfo {volume} {28}},\ \bibinfo {pages} {4777} (\bibinfo {year}
  {1995}{\natexlab{b}})}\BibitemShut {NoStop}%
\bibitem [{\citenamefont {Oh}\ \emph {et~al.}(2017)\citenamefont {Oh},
  \citenamefont {Katsura}, \citenamefont {Lee},\ and\ \citenamefont
  {Han}}]{oh2017proposal}%
  \BibitemOpen
  \bibfield  {author} {\bibinfo {author} {\bibfnamefont {Y.-T.}\ \bibnamefont
  {Oh}}, \bibinfo {author} {\bibfnamefont {H.}~\bibnamefont {Katsura}},
  \bibinfo {author} {\bibfnamefont {H.-Y.}\ \bibnamefont {Lee}},\ and\ \bibinfo
  {author} {\bibfnamefont {J.~H.}\ \bibnamefont {Han}},\ }\bibfield  {title}
  {\bibinfo {title} {{Proposal of a spin-one chain model with competing dimer
  and trimer interactions}},\ }\href
  {https://doi.org/10.1103/PhysRevB.96.165126} {\bibfield  {journal} {\bibinfo
  {journal} {Phys. Rev. B}\ }\textbf {\bibinfo {volume} {96}},\ \bibinfo
  {pages} {165126} (\bibinfo {year} {2017})}\BibitemShut {NoStop}%
\bibitem [{\citenamefont {Chen}\ \emph {et~al.}(2020)\citenamefont {Chen},
  \citenamefont {Capponi}, \citenamefont {Wietek}, \citenamefont {Mambrini},
  \citenamefont {Schuch},\ and\ \citenamefont {Poilblanc}}]{Pijk}%
  \BibitemOpen
  \bibfield  {author} {\bibinfo {author} {\bibfnamefont {J.-Y.}\ \bibnamefont
  {Chen}}, \bibinfo {author} {\bibfnamefont {S.}~\bibnamefont {Capponi}},
  \bibinfo {author} {\bibfnamefont {A.}~\bibnamefont {Wietek}}, \bibinfo
  {author} {\bibfnamefont {M.}~\bibnamefont {Mambrini}}, \bibinfo {author}
  {\bibfnamefont {N.}~\bibnamefont {Schuch}},\ and\ \bibinfo {author}
  {\bibfnamefont {D.}~\bibnamefont {Poilblanc}},\ }\bibfield  {title} {\bibinfo
  {title} {{$\mathrm{SU}(3{)}_{1}$ Chiral Spin Liquid on the Square Lattice: A
  View from Symmetric Projected Entangled Pair States}},\ }\href
  {https://doi.org/10.1103/PhysRevLett.125.017201} {\bibfield  {journal}
  {\bibinfo  {journal} {Phys. Rev. Lett.}\ }\textbf {\bibinfo {volume} {125}},\
  \bibinfo {pages} {017201} (\bibinfo {year} {2020})}\BibitemShut {NoStop}%
\bibitem [{\citenamefont {Schierenberg}\ \emph {et~al.}(2012)\citenamefont
  {Schierenberg}, \citenamefont {Bruckmann},\ and\ \citenamefont
  {Wettig}}]{schierenberg2012wigner}%
  \BibitemOpen
  \bibfield  {author} {\bibinfo {author} {\bibfnamefont {S.}~\bibnamefont
  {Schierenberg}}, \bibinfo {author} {\bibfnamefont {F.}~\bibnamefont
  {Bruckmann}},\ and\ \bibinfo {author} {\bibfnamefont {T.}~\bibnamefont
  {Wettig}},\ }\bibfield  {title} {\bibinfo {title} {{Wigner surmise for mixed
  symmetry classes in random matrix theory}},\ }\href
  {https://doi.org/10.1103/PhysRevE.85.061130} {\bibfield  {journal} {\bibinfo
  {journal} {Phys. Rev. E}\ }\textbf {\bibinfo {volume} {85}},\ \bibinfo
  {pages} {061130} (\bibinfo {year} {2012})}\BibitemShut {NoStop}%
\bibitem [{\citenamefont {Kundu}\ \emph {et~al.}(2023)\citenamefont {Kundu},
  \citenamefont {Kumar},\ and\ \citenamefont
  {Sen~Gupta}}]{kundu2023signatures}%
  \BibitemOpen
  \bibfield  {author} {\bibinfo {author} {\bibfnamefont {D.}~\bibnamefont
  {Kundu}}, \bibinfo {author} {\bibfnamefont {S.}~\bibnamefont {Kumar}},\ and\
  \bibinfo {author} {\bibfnamefont {S.}~\bibnamefont {Sen~Gupta}},\ }\bibfield
  {title} {\bibinfo {title} {{Signatures of spectral crossovers in the
  short-and long-range spectral correlations of a disordered spin-chain with
  Kramers degeneracy}},\ }\href {https://doi.org/10.1103/PhysRevB.107.094205}
  {\bibfield  {journal} {\bibinfo  {journal} {Phys. Rev. B}\ }\textbf {\bibinfo
  {volume} {107}},\ \bibinfo {pages} {094205} (\bibinfo {year}
  {2023})}\BibitemShut {NoStop}%
\bibitem [{\citenamefont {Hirano}\ and\ \citenamefont
  {Hatsugai}(2007)}]{hirano2007entanglement}%
  \BibitemOpen
  \bibfield  {author} {\bibinfo {author} {\bibfnamefont {T.}~\bibnamefont
  {Hirano}}\ and\ \bibinfo {author} {\bibfnamefont {Y.}~\bibnamefont
  {Hatsugai}},\ }\bibfield  {title} {\bibinfo {title} {{Entanglement entropy of
  one-dimensional gapped spin chains}},\ }\href
  {https://doi.org/10.1143/JPSJ.76.074603} {\bibfield  {journal} {\bibinfo
  {journal} {J. Phys. Soc. Jpn.}\ }\textbf {\bibinfo {volume} {76}},\ \bibinfo
  {pages} {074603} (\bibinfo {year} {2007})}\BibitemShut {NoStop}%
\bibitem [{\citenamefont {Katsura}\ \emph {et~al.}(2007)\citenamefont
  {Katsura}, \citenamefont {Hirano},\ and\ \citenamefont
  {Hatsugai}}]{katsura2007exact}%
  \BibitemOpen
  \bibfield  {author} {\bibinfo {author} {\bibfnamefont {H.}~\bibnamefont
  {Katsura}}, \bibinfo {author} {\bibfnamefont {T.}~\bibnamefont {Hirano}},\
  and\ \bibinfo {author} {\bibfnamefont {Y.}~\bibnamefont {Hatsugai}},\
  }\bibfield  {title} {\bibinfo {title} {{Exact analysis of entanglement in
  gapped quantum spin chains}},\ }\href
  {https://doi.org/10.1103/PhysRevB.76.012401} {\bibfield  {journal} {\bibinfo
  {journal} {Phys. Rev. B}\ }\textbf {\bibinfo {volume} {76}},\ \bibinfo
  {pages} {012401} (\bibinfo {year} {2007})}\BibitemShut {NoStop}%
\bibitem [{\citenamefont {Lange}\ \emph {et~al.}(1994)\citenamefont {Lange},
  \citenamefont {Kl{\"u}mper},\ and\ \citenamefont
  {Zittartz}}]{lange1994exact}%
  \BibitemOpen
  \bibfield  {author} {\bibinfo {author} {\bibfnamefont {C.}~\bibnamefont
  {Lange}}, \bibinfo {author} {\bibfnamefont {A.}~\bibnamefont {Kl{\"u}mper}},\
  and\ \bibinfo {author} {\bibfnamefont {J.}~\bibnamefont {Zittartz}},\
  }\bibfield  {title} {\bibinfo {title} {{Exact groundstates for
  antiferromagnetic spin-one chains with nearest and next-nearest neighbour
  interactions}},\ }\href {https://doi.org/10.1007/BF01313293} {\bibfield
  {journal} {\bibinfo  {journal} {Z. Phys. B}\ }\textbf {\bibinfo {volume}
  {96}},\ \bibinfo {pages} {267} (\bibinfo {year} {1994})}\BibitemShut
  {NoStop}%
\bibitem [{\citenamefont {Nakano}\ and\ \citenamefont
  {Takahashi}(1996)}]{nakano1996long}%
  \BibitemOpen
  \bibfield  {author} {\bibinfo {author} {\bibfnamefont {H.}~\bibnamefont
  {Nakano}}\ and\ \bibinfo {author} {\bibfnamefont {M.}~\bibnamefont
  {Takahashi}},\ }\bibfield  {title} {\bibinfo {title} {{Long-ranged
  interacting S= 1 spin chain with the exact valence-bond-solid state}},\
  }\href {https://doi.org/10.1103/PhysRevB.54.9000} {\bibfield  {journal}
  {\bibinfo  {journal} {Phys. Rev. B}\ }\textbf {\bibinfo {volume} {54}},\
  \bibinfo {pages} {9000} (\bibinfo {year} {1996})}\BibitemShut {NoStop}%
\bibitem [{\citenamefont {Scalapino}\ \emph {et~al.}(1998)\citenamefont
  {Scalapino}, \citenamefont {Zhang},\ and\ \citenamefont
  {Hanke}}]{scalapino1998so}%
  \BibitemOpen
  \bibfield  {author} {\bibinfo {author} {\bibfnamefont {D.}~\bibnamefont
  {Scalapino}}, \bibinfo {author} {\bibfnamefont {S.-C.}\ \bibnamefont
  {Zhang}},\ and\ \bibinfo {author} {\bibfnamefont {W.}~\bibnamefont {Hanke}},\
  }\bibfield  {title} {\bibinfo {title} {{SO (5) symmetric ladder}},\ }\href
  {https://doi.org/10.1103/PhysRevB.58.443} {\bibfield  {journal} {\bibinfo
  {journal} {Phys. Rev. B}\ }\textbf {\bibinfo {volume} {58}},\ \bibinfo
  {pages} {443} (\bibinfo {year} {1998})}\BibitemShut {NoStop}%
\bibitem [{\citenamefont {Frahm}\ and\ \citenamefont
  {Stahlsmeier}(2001)}]{frahm2001electronic}%
  \BibitemOpen
  \bibfield  {author} {\bibinfo {author} {\bibfnamefont {H.}~\bibnamefont
  {Frahm}}\ and\ \bibinfo {author} {\bibfnamefont {M.}~\bibnamefont
  {Stahlsmeier}},\ }\bibfield  {title} {\bibinfo {title} {{Electronic ladders
  with SO (5) symmetry: Phase diagrams and correlations at half filling}},\
  }\href {https://doi.org/10.1103/PhysRevB.63.125109} {\bibfield  {journal}
  {\bibinfo  {journal} {Phys. Rev. B}\ }\textbf {\bibinfo {volume} {63}},\
  \bibinfo {pages} {125109} (\bibinfo {year} {2001})}\BibitemShut {NoStop}%
\bibitem [{\citenamefont {Tu}\ \emph {et~al.}(2008)\citenamefont {Tu},
  \citenamefont {Zhang},\ and\ \citenamefont {Xiang}}]{tu2008class}%
  \BibitemOpen
  \bibfield  {author} {\bibinfo {author} {\bibfnamefont {H.-H.}\ \bibnamefont
  {Tu}}, \bibinfo {author} {\bibfnamefont {G.-M.}\ \bibnamefont {Zhang}},\ and\
  \bibinfo {author} {\bibfnamefont {T.}~\bibnamefont {Xiang}},\ }\bibfield
  {title} {\bibinfo {title} {{Class of exactly solvable S O (n) symmetric spin
  chains with matrix product ground states}},\ }\href
  {https://doi.org/10.1103/PhysRevB.78.094404} {\bibfield  {journal} {\bibinfo
  {journal} {Phys. Rev. B}\ }\textbf {\bibinfo {volume} {78}},\ \bibinfo
  {pages} {094404} (\bibinfo {year} {2008})}\BibitemShut {NoStop}%
\bibitem [{\citenamefont {Lee}\ \emph {et~al.}(2020)\citenamefont {Lee},
  \citenamefont {Melendrez}, \citenamefont {Pal},\ and\ \citenamefont
  {Changlani}}]{lee2020exact}%
  \BibitemOpen
  \bibfield  {author} {\bibinfo {author} {\bibfnamefont {K.}~\bibnamefont
  {Lee}}, \bibinfo {author} {\bibfnamefont {R.}~\bibnamefont {Melendrez}},
  \bibinfo {author} {\bibfnamefont {A.}~\bibnamefont {Pal}},\ and\ \bibinfo
  {author} {\bibfnamefont {H.~J.}\ \bibnamefont {Changlani}},\ }\bibfield
  {title} {\bibinfo {title} {{Exact three-colored quantum scars from geometric
  frustration}},\ }\href {https://doi.org/10.1103/PhysRevB.101.241111}
  {\bibfield  {journal} {\bibinfo  {journal} {Phys. Rev. B}\ }\textbf {\bibinfo
  {volume} {101}},\ \bibinfo {pages} {241111(R)} (\bibinfo {year}
  {2020})}\BibitemShut {NoStop}%
\bibitem [{\citenamefont {Lee}\ \emph {et~al.}(2021)\citenamefont {Lee},
  \citenamefont {Pal},\ and\ \citenamefont {Changlani}}]{lee2021frustration}%
  \BibitemOpen
  \bibfield  {author} {\bibinfo {author} {\bibfnamefont {K.}~\bibnamefont
  {Lee}}, \bibinfo {author} {\bibfnamefont {A.}~\bibnamefont {Pal}},\ and\
  \bibinfo {author} {\bibfnamefont {H.~J.}\ \bibnamefont {Changlani}},\
  }\bibfield  {title} {\bibinfo {title} {{Frustration-induced emergent hilbert
  space fragmentation}},\ }\href {https://doi.org/10.1103/PhysRevB.103.235133}
  {\bibfield  {journal} {\bibinfo  {journal} {Phys. Rev. B}\ }\textbf {\bibinfo
  {volume} {103}},\ \bibinfo {pages} {235133} (\bibinfo {year}
  {2021})}\BibitemShut {NoStop}%
\bibitem [{\citenamefont {Alhambra}\ \emph {et~al.}(2020)\citenamefont
  {Alhambra}, \citenamefont {Anshu},\ and\ \citenamefont
  {Wilming}}]{alhambra2020revivals}%
  \BibitemOpen
  \bibfield  {author} {\bibinfo {author} {\bibfnamefont {{\'A}.~M.}\
  \bibnamefont {Alhambra}}, \bibinfo {author} {\bibfnamefont {A.}~\bibnamefont
  {Anshu}},\ and\ \bibinfo {author} {\bibfnamefont {H.}~\bibnamefont
  {Wilming}},\ }\bibfield  {title} {\bibinfo {title} {{Revivals imply quantum
  many-body scars}},\ }\href {https://doi.org/10.1103/PhysRevB.101.205107}
  {\bibfield  {journal} {\bibinfo  {journal} {Phys. Rev. B}\ }\textbf {\bibinfo
  {volume} {101}},\ \bibinfo {pages} {205107} (\bibinfo {year}
  {2020})}\BibitemShut {NoStop}%
\bibitem [{\citenamefont {Shi}\ \emph {et~al.}(2006)\citenamefont {Shi},
  \citenamefont {Duan},\ and\ \citenamefont {Vidal}}]{shi2006classical}%
  \BibitemOpen
  \bibfield  {author} {\bibinfo {author} {\bibfnamefont {Y.-Y.}\ \bibnamefont
  {Shi}}, \bibinfo {author} {\bibfnamefont {L.-M.}\ \bibnamefont {Duan}},\ and\
  \bibinfo {author} {\bibfnamefont {G.}~\bibnamefont {Vidal}},\ }\bibfield
  {title} {\bibinfo {title} {{Classical simulation of quantum many-body systems
  with a tree tensor network}},\ }\href
  {https://doi.org/10.1103/PhysRevA.74.022320} {\bibfield  {journal} {\bibinfo
  {journal} {Phys. Rev. A}\ }\textbf {\bibinfo {volume} {74}},\ \bibinfo
  {pages} {022320} (\bibinfo {year} {2006})}\BibitemShut {NoStop}%
\bibitem [{\citenamefont {Katsura}\ \emph {et~al.}(2010)\citenamefont
  {Katsura}, \citenamefont {Kawashima}, \citenamefont {Kirillov}, \citenamefont
  {Korepin},\ and\ \citenamefont {Tanaka}}]{katsura2010entanglement}%
  \BibitemOpen
  \bibfield  {author} {\bibinfo {author} {\bibfnamefont {H.}~\bibnamefont
  {Katsura}}, \bibinfo {author} {\bibfnamefont {N.}~\bibnamefont {Kawashima}},
  \bibinfo {author} {\bibfnamefont {A.~N.}\ \bibnamefont {Kirillov}}, \bibinfo
  {author} {\bibfnamefont {V.~E.}\ \bibnamefont {Korepin}},\ and\ \bibinfo
  {author} {\bibfnamefont {S.}~\bibnamefont {Tanaka}},\ }\bibfield  {title}
  {\bibinfo {title} {{Entanglement in valence-bond-solid states on symmetric
  graphs}},\ }\href {https://doi.org/10.1088/1751-8113/43/25/255303} {\bibfield
   {journal} {\bibinfo  {journal} {J. Phys. A: Math. Theor.}\ }\textbf
  {\bibinfo {volume} {43}},\ \bibinfo {pages} {255303} (\bibinfo {year}
  {2010})}\BibitemShut {NoStop}%
\bibitem [{\citenamefont {Page}(1993)}]{page1993average}%
  \BibitemOpen
  \bibfield  {author} {\bibinfo {author} {\bibfnamefont {D.~N.}\ \bibnamefont
  {Page}},\ }\bibfield  {title} {\bibinfo {title} {{Average entropy of a
  subsystem}},\ }\href {https://doi.org/10.1103/PhysRevLett.71.1291} {\bibfield
   {journal} {\bibinfo  {journal} {Phys. Rev. Lett.}\ }\textbf {\bibinfo
  {volume} {71}},\ \bibinfo {pages} {1291} (\bibinfo {year}
  {1993})}\BibitemShut {NoStop}%
\bibitem [{\citenamefont {Ghoshal}\ and\ \citenamefont
  {Zamolodchikov}(1994)}]{ghoshal1994boundary}%
  \BibitemOpen
  \bibfield  {author} {\bibinfo {author} {\bibfnamefont {S.}~\bibnamefont
  {Ghoshal}}\ and\ \bibinfo {author} {\bibfnamefont {A.}~\bibnamefont
  {Zamolodchikov}},\ }\bibfield  {title} {\bibinfo {title} {{Boundary S matrix
  and boundary state in two-dimensional integrable quantum field theory}},\
  }\href {https://doi.org/10.1142/S0217751X94001552} {\bibfield  {journal}
  {\bibinfo  {journal} {Int. J. Mod. Phys. A}\ }\textbf {\bibinfo {volume}
  {9}},\ \bibinfo {pages} {3841} (\bibinfo {year} {1994})}\BibitemShut
  {NoStop}%
\bibitem [{\citenamefont {Schindler}\ \emph {et~al.}(2022)\citenamefont
  {Schindler}, \citenamefont {Regnault},\ and\ \citenamefont
  {Bernevig}}]{schindler2022exact}%
  \BibitemOpen
  \bibfield  {author} {\bibinfo {author} {\bibfnamefont {F.}~\bibnamefont
  {Schindler}}, \bibinfo {author} {\bibfnamefont {N.}~\bibnamefont
  {Regnault}},\ and\ \bibinfo {author} {\bibfnamefont {B.~A.}\ \bibnamefont
  {Bernevig}},\ }\bibfield  {title} {\bibinfo {title} {{Exact quantum scars in
  the chiral nonlinear Luttinger liquid}},\ }\href
  {https://doi.org/10.1103/PhysRevB.105.035146} {\bibfield  {journal} {\bibinfo
   {journal} {Phys. Rev. B}\ }\textbf {\bibinfo {volume} {105}},\ \bibinfo
  {pages} {035146} (\bibinfo {year} {2022})}\BibitemShut {NoStop}%
\bibitem [{\citenamefont {Martin}\ and\ \citenamefont
  {Matveev}(2022)}]{martin2022scar}%
  \BibitemOpen
  \bibfield  {author} {\bibinfo {author} {\bibfnamefont {I.}~\bibnamefont
  {Martin}}\ and\ \bibinfo {author} {\bibfnamefont {K.~A.}\ \bibnamefont
  {Matveev}},\ }\bibfield  {title} {\bibinfo {title} {{Scar states in a system
  of interacting chiral fermions}},\ }\href
  {https://doi.org/10.1103/PhysRevB.105.045119} {\bibfield  {journal} {\bibinfo
   {journal} {Physical Review B}\ }\textbf {\bibinfo {volume} {105}},\ \bibinfo
  {pages} {045119} (\bibinfo {year} {2022})}\BibitemShut {NoStop}%
\bibitem [{\citenamefont {Liska}\ \emph {et~al.}(2023)\citenamefont {Liska},
  \citenamefont {Gritsev}, \citenamefont {Vleeshouwers},\ and\ \citenamefont
  {Min{\'a}{\v{r}}}}]{liska2022holographic}%
  \BibitemOpen
  \bibfield  {author} {\bibinfo {author} {\bibfnamefont {D.}~\bibnamefont
  {Liska}}, \bibinfo {author} {\bibfnamefont {V.}~\bibnamefont {Gritsev}},
  \bibinfo {author} {\bibfnamefont {W.}~\bibnamefont {Vleeshouwers}},\ and\
  \bibinfo {author} {\bibfnamefont {J.}~\bibnamefont {Min{\'a}{\v{r}}}},\
  }\bibfield  {title} {\bibinfo {title} {{Holographic Quantum Scars}},\ }\href
  {https://doi.org/10.21468/SciPostPhys.15.3.106} {\bibfield  {journal}
  {\bibinfo  {journal} {SciPost Phys.}\ }\textbf {\bibinfo {volume} {15}},\
  \bibinfo {pages} {106} (\bibinfo {year} {2023})}\BibitemShut {NoStop}%
\bibitem [{\citenamefont {Cotler}\ and\ \citenamefont
  {Wei}(2023)}]{cotler2022}%
  \BibitemOpen
  \bibfield  {author} {\bibinfo {author} {\bibfnamefont {J.}~\bibnamefont
  {Cotler}}\ and\ \bibinfo {author} {\bibfnamefont {A.~Y.}\ \bibnamefont
  {Wei}},\ }\bibfield  {title} {\bibinfo {title} {{Quantum Scars in Quantum
  Field Theory}},\ }\href {https://doi.org/10.1103/PhysRevD.107.125005}
  {\bibfield  {journal} {\bibinfo  {journal} {Phys. Rev. D}\ }\textbf {\bibinfo
  {volume} {107}},\ \bibinfo {pages} {125005} (\bibinfo {year}
  {2023})}\BibitemShut {NoStop}%
\bibitem [{\citenamefont {Iemini}\ \emph {et~al.}(2018)\citenamefont {Iemini},
  \citenamefont {Russomanno}, \citenamefont {Keeling}, \citenamefont
  {Schir{\`o}}, \citenamefont {Dalmonte},\ and\ \citenamefont
  {Fazio}}]{iemini2018boundary}%
  \BibitemOpen
  \bibfield  {author} {\bibinfo {author} {\bibfnamefont {F.}~\bibnamefont
  {Iemini}}, \bibinfo {author} {\bibfnamefont {A.}~\bibnamefont {Russomanno}},
  \bibinfo {author} {\bibfnamefont {J.}~\bibnamefont {Keeling}}, \bibinfo
  {author} {\bibfnamefont {M.}~\bibnamefont {Schir{\`o}}}, \bibinfo {author}
  {\bibfnamefont {M.}~\bibnamefont {Dalmonte}},\ and\ \bibinfo {author}
  {\bibfnamefont {R.}~\bibnamefont {Fazio}},\ }\bibfield  {title} {\bibinfo
  {title} {{Boundary time crystals}},\ }\href
  {https://doi.org/10.1103/PhysRevLett.121.035301} {\bibfield  {journal}
  {\bibinfo  {journal} {Phys. Rev. Lett.}\ }\textbf {\bibinfo {volume} {121}},\
  \bibinfo {pages} {035301} (\bibinfo {year} {2018})}\BibitemShut {NoStop}%
\bibitem [{\citenamefont {Dutta}\ and\ \citenamefont
  {Cooper}(2021)}]{dutta2021out}%
  \BibitemOpen
  \bibfield  {author} {\bibinfo {author} {\bibfnamefont {S.}~\bibnamefont
  {Dutta}}\ and\ \bibinfo {author} {\bibfnamefont {N.~R.}\ \bibnamefont
  {Cooper}},\ }\bibfield  {title} {\bibinfo {title} {{Out-of-equilibrium steady
  states of a locally driven lossy qubit array}},\ }\href
  {https://doi.org/10.1103/PhysRevResearch.3.L012016} {\bibfield  {journal}
  {\bibinfo  {journal} {Phys. Rev. Res.}\ }\textbf {\bibinfo {volume} {3}},\
  \bibinfo {pages} {L012016} (\bibinfo {year} {2021})}\BibitemShut {NoStop}%
\bibitem [{\citenamefont {Tindall}\ \emph {et~al.}(2021)\citenamefont
  {Tindall}, \citenamefont {Schlawin}, \citenamefont {Sentef},\ and\
  \citenamefont {Jaksch}}]{tindall2021analytical}%
  \BibitemOpen
  \bibfield  {author} {\bibinfo {author} {\bibfnamefont {J.}~\bibnamefont
  {Tindall}}, \bibinfo {author} {\bibfnamefont {F.}~\bibnamefont {Schlawin}},
  \bibinfo {author} {\bibfnamefont {M.~A.}\ \bibnamefont {Sentef}},\ and\
  \bibinfo {author} {\bibfnamefont {D.}~\bibnamefont {Jaksch}},\ }\bibfield
  {title} {\bibinfo {title} {{Analytical solution for the steady states of the
  driven Hubbard model}},\ }\href {https://doi.org/10.1103/PhysRevB.103.035146}
  {\bibfield  {journal} {\bibinfo  {journal} {Phys. Rev. B}\ }\textbf {\bibinfo
  {volume} {103}},\ \bibinfo {pages} {035146} (\bibinfo {year}
  {2021})}\BibitemShut {NoStop}%
\bibitem [{\citenamefont {Bu{\v{c}}a}\ \emph {et~al.}(2022)\citenamefont
  {Bu{\v{c}}a}, \citenamefont {Booker},\ and\ \citenamefont
  {Jaksch}}]{buvca2022algebraic}%
  \BibitemOpen
  \bibfield  {author} {\bibinfo {author} {\bibfnamefont {B.}~\bibnamefont
  {Bu{\v{c}}a}}, \bibinfo {author} {\bibfnamefont {C.}~\bibnamefont {Booker}},\
  and\ \bibinfo {author} {\bibfnamefont {D.}~\bibnamefont {Jaksch}},\
  }\bibfield  {title} {\bibinfo {title} {{Algebraic theory of quantum
  synchronization and limit cycles under dissipation}},\ }\href
  {https://doi.org/10.21468/SciPostPhys.12.3.097} {\bibfield  {journal}
  {\bibinfo  {journal} {SciPost Phys.}\ }\textbf {\bibinfo {volume} {12}},\
  \bibinfo {pages} {097} (\bibinfo {year} {2022})}\BibitemShut {NoStop}%
\bibitem [{\citenamefont {Ziolkowska}\ and\ \citenamefont
  {Essler}(2020)}]{ziolkowska2020yang}%
  \BibitemOpen
  \bibfield  {author} {\bibinfo {author} {\bibfnamefont {A.~A.}\ \bibnamefont
  {Ziolkowska}}\ and\ \bibinfo {author} {\bibfnamefont {F.}~\bibnamefont
  {Essler}},\ }\bibfield  {title} {\bibinfo {title} {{Yang-baxter integrable
  lindblad equations}},\ }\href {https://doi.org/doi:
  10.21468/SciPostPhys.8.3.044} {\bibfield  {journal} {\bibinfo  {journal}
  {SciPost Phys.}\ }\textbf {\bibinfo {volume} {8}},\ \bibinfo {pages} {044}
  (\bibinfo {year} {2020})}\BibitemShut {NoStop}%
\bibitem [{\citenamefont {de~Leeuw}\ \emph {et~al.}(2021)\citenamefont
  {de~Leeuw}, \citenamefont {Paletta},\ and\ \citenamefont
  {Pozsgay}}]{de2021constructing}%
  \BibitemOpen
  \bibfield  {author} {\bibinfo {author} {\bibfnamefont {M.}~\bibnamefont
  {de~Leeuw}}, \bibinfo {author} {\bibfnamefont {C.}~\bibnamefont {Paletta}},\
  and\ \bibinfo {author} {\bibfnamefont {B.}~\bibnamefont {Pozsgay}},\
  }\bibfield  {title} {\bibinfo {title} {{Constructing integrable Lindblad
  superoperators}},\ }\href {https://doi.org/10.1103/PhysRevLett.126.240403}
  {\bibfield  {journal} {\bibinfo  {journal} {Phys. Rev. Lett.}\ }\textbf
  {\bibinfo {volume} {126}},\ \bibinfo {pages} {240403} (\bibinfo {year}
  {2021})}\BibitemShut {NoStop}%
\bibitem [{\citenamefont {de~Leeuw}\ \emph {et~al.}(2023)\citenamefont
  {de~Leeuw}, \citenamefont {Paletta}, \citenamefont {Pozsgay},\ and\
  \citenamefont {Vernier}}]{de2023hidden}%
  \BibitemOpen
  \bibfield  {author} {\bibinfo {author} {\bibfnamefont {M.}~\bibnamefont
  {de~Leeuw}}, \bibinfo {author} {\bibfnamefont {C.}~\bibnamefont {Paletta}},
  \bibinfo {author} {\bibfnamefont {B.}~\bibnamefont {Pozsgay}},\ and\ \bibinfo
  {author} {\bibfnamefont {E.}~\bibnamefont {Vernier}},\ }\bibfield  {title}
  {\bibinfo {title} {{Hidden quasi-local charges and Gibbs ensemble in a
  Lindblad system}},\ }\href {https://arxiv.org/abs/2305.01922} {\bibfield
  {journal} {\bibinfo  {journal} {arXiv:2305.01922}\ } (\bibinfo {year}
  {2023})}\BibitemShut {NoStop}%
\bibitem [{\citenamefont {Gritsev}\ and\ \citenamefont
  {Polkovnikov}(2017)}]{gritsev2017integrable}%
  \BibitemOpen
  \bibfield  {author} {\bibinfo {author} {\bibfnamefont {V.}~\bibnamefont
  {Gritsev}}\ and\ \bibinfo {author} {\bibfnamefont {A.}~\bibnamefont
  {Polkovnikov}},\ }\bibfield  {title} {\bibinfo {title} {{Integrable Floquet
  dynamics}},\ }\href {https://doi.org/10.21468/SciPostPhys.2.3.021} {\bibfield
   {journal} {\bibinfo  {journal} {SciPost Phys.}\ }\textbf {\bibinfo {volume}
  {2}},\ \bibinfo {pages} {021} (\bibinfo {year} {2017})}\BibitemShut {NoStop}%
\bibitem [{\citenamefont {Vanicat}\ \emph {et~al.}(2018)\citenamefont
  {Vanicat}, \citenamefont {Zadnik},\ and\ \citenamefont
  {Prosen}}]{vanicat2018integrable}%
  \BibitemOpen
  \bibfield  {author} {\bibinfo {author} {\bibfnamefont {M.}~\bibnamefont
  {Vanicat}}, \bibinfo {author} {\bibfnamefont {L.}~\bibnamefont {Zadnik}},\
  and\ \bibinfo {author} {\bibfnamefont {T.}~\bibnamefont {Prosen}},\
  }\bibfield  {title} {\bibinfo {title} {{Integrable Trotterization: Local
  Conservation Laws and Boundary Driving}},\ }\href
  {https://doi.org/10.1103/PhysRevLett.121.030606} {\bibfield  {journal}
  {\bibinfo  {journal} {Phys. Rev. Lett.}\ }\textbf {\bibinfo {volume} {121}},\
  \bibinfo {pages} {030606} (\bibinfo {year} {2018})}\BibitemShut {NoStop}%
\bibitem [{\citenamefont {Lotkov}\ \emph {et~al.}(2022)\citenamefont {Lotkov},
  \citenamefont {Gritsev}, \citenamefont {Fedorov},\ and\ \citenamefont
  {Kurlov}}]{lotkov2022floquet}%
  \BibitemOpen
  \bibfield  {author} {\bibinfo {author} {\bibfnamefont {A.}~\bibnamefont
  {Lotkov}}, \bibinfo {author} {\bibfnamefont {V.}~\bibnamefont {Gritsev}},
  \bibinfo {author} {\bibfnamefont {A.}~\bibnamefont {Fedorov}},\ and\ \bibinfo
  {author} {\bibfnamefont {D.}~\bibnamefont {Kurlov}},\ }\bibfield  {title}
  {\bibinfo {title} {{Floquet integrability and long-range entanglement
  generation in the one-dimensional quantum Potts model}},\ }\href
  {https://doi.org/10.1103/PhysRevB.105.144306} {\bibfield  {journal} {\bibinfo
   {journal} {Phys. Rev. B}\ }\textbf {\bibinfo {volume} {105}},\ \bibinfo
  {pages} {144306} (\bibinfo {year} {2022})}\BibitemShut {NoStop}%
\bibitem [{\citenamefont {Weinberg}\ and\ \citenamefont
  {Bukov}(2017)}]{quspin1}%
  \BibitemOpen
  \bibfield  {author} {\bibinfo {author} {\bibfnamefont {P.}~\bibnamefont
  {Weinberg}}\ and\ \bibinfo {author} {\bibfnamefont {M.}~\bibnamefont
  {Bukov}},\ }\bibfield  {title} {\bibinfo {title} {{QuSpin: a Python Package
  for Dynamics and Exact Diagonalisation of Quantum Many Body Systems part I:
  spin chains}},\ }\href {https://doi.org/10.21468/SciPostPhys.2.1.003}
  {\bibfield  {journal} {\bibinfo  {journal} {SciPost Phys.}\ }\textbf
  {\bibinfo {volume} {2}},\ \bibinfo {pages} {003} (\bibinfo {year}
  {2017})}\BibitemShut {NoStop}%
\bibitem [{\citenamefont {Weinberg}\ and\ \citenamefont
  {Bukov}(2019)}]{quspin2}%
  \BibitemOpen
  \bibfield  {author} {\bibinfo {author} {\bibfnamefont {P.}~\bibnamefont
  {Weinberg}}\ and\ \bibinfo {author} {\bibfnamefont {M.}~\bibnamefont
  {Bukov}},\ }\bibfield  {title} {\bibinfo {title} {{QuSpin: a Python Package
  for Dynamics and Exact Diagonalisation of Quantum Many Body Systems. Part II:
  bosons, fermions and higher spins}},\ }\href
  {https://doi.org/10.21468/SciPostPhys.7.2.020} {\bibfield  {journal}
  {\bibinfo  {journal} {SciPost Phys.}\ }\textbf {\bibinfo {volume} {7}},\
  \bibinfo {pages} {20} (\bibinfo {year} {2019})}\BibitemShut {NoStop}%
\bibitem [{\citenamefont {Eliot}(1990)}]{ELIOT1990109}%
  \BibitemOpen
  \bibfield  {author} {\bibinfo {author} {\bibfnamefont {C.}~\bibnamefont
  {Eliot}},\ }\bibfield  {title} {\bibinfo {title} {Chapter 3 - probability
  distributions},\ }in\ \href
  {https://doi.org/10.1016/B978-0-08-057105-8.50012-6} {\emph {\bibinfo
  {booktitle} {Probability, Statistics, and Queuing Theory with Computer
  Science Applications (Second Edition)}}},\ \bibinfo {series and number}
  {Computer Science and Scientific Computing},\ \bibinfo {editor} {edited by\
  \bibinfo {editor} {\bibfnamefont {A.~O.}\ \bibnamefont {Allen}}}\ (\bibinfo
  {publisher} {Academic Press},\ \bibinfo {address} {San Diego},\ \bibinfo
  {year} {1990})\ \bibinfo {edition} {second edition}\ ed.,\ pp.\ \bibinfo
  {pages} {109--198}\BibitemShut {NoStop}%
\bibitem [{\citenamefont {Uspensky}(1937)}]{uspensky1937introduction}%
  \BibitemOpen
  \bibfield  {author} {\bibinfo {author} {\bibfnamefont {J.~V.}\ \bibnamefont
  {Uspensky}},\ }\href@noop {} {\emph {\bibinfo {title} {{Introduction to
  mathematical probability}}}}\ (\bibinfo  {publisher} {McGraw-Hill Book
  Company, New York},\ \bibinfo {year} {1937})\BibitemShut {NoStop}%
\bibitem [{\citenamefont {Diener}\ and\ \citenamefont
  {Diener}(2005)}]{diener2005higher}%
  \BibitemOpen
  \bibfield  {author} {\bibinfo {author} {\bibfnamefont {F.}~\bibnamefont
  {Diener}}\ and\ \bibinfo {author} {\bibfnamefont {M.}~\bibnamefont
  {Diener}},\ }\bibfield  {title} {\bibinfo {title} {{Higher-order terms for
  the de Moivre-Laplace theorem}},\ }\href
  {https://doi.org/10.1090/conm/373/06920} {\bibfield  {journal} {\bibinfo
  {journal} {Contemp. Math.}\ }\textbf {\bibinfo {volume} {373}},\ \bibinfo
  {pages} {191} (\bibinfo {year} {2005})}\BibitemShut {NoStop}%
\bibitem [{\citenamefont {Georgi}(2000)}]{georgi2000lie}%
  \BibitemOpen
  \bibfield  {author} {\bibinfo {author} {\bibfnamefont {H.}~\bibnamefont
  {Georgi}},\ }\href {https://doi.org/10.1201/9780429499210} {\emph {\bibinfo
  {title} {{Lie algebras in particle physics: from isospin to unified
  theories}}}}\ (\bibinfo  {publisher} {Taylor \& Francis},\ \bibinfo {address}
  {London},\ \bibinfo {year} {2000})\BibitemShut {NoStop}%
\bibitem [{\citenamefont {Biroli}\ \emph {et~al.}(2010)\citenamefont {Biroli},
  \citenamefont {Kollath},\ and\ \citenamefont
  {L{\"a}uchli}}]{biroli2010effect}%
  \BibitemOpen
  \bibfield  {author} {\bibinfo {author} {\bibfnamefont {G.}~\bibnamefont
  {Biroli}}, \bibinfo {author} {\bibfnamefont {C.}~\bibnamefont {Kollath}},\
  and\ \bibinfo {author} {\bibfnamefont {A.~M.}\ \bibnamefont {L{\"a}uchli}},\
  }\bibfield  {title} {\bibinfo {title} {{Effect of rare fluctuations on the
  thermalization of isolated quantum systems}},\ }\href
  {https://doi.org/10.1103/PhysRevLett.105.250401} {\bibfield  {journal}
  {\bibinfo  {journal} {Phys. Rev. Lett.}\ }\textbf {\bibinfo {volume} {105}},\
  \bibinfo {pages} {250401} (\bibinfo {year} {2010})}\BibitemShut {NoStop}%
\bibitem [{\citenamefont {Iyoda}\ \emph {et~al.}(2017)\citenamefont {Iyoda},
  \citenamefont {Kaneko},\ and\ \citenamefont {Sagawa}}]{iyoda2017fluctuation}%
  \BibitemOpen
  \bibfield  {author} {\bibinfo {author} {\bibfnamefont {E.}~\bibnamefont
  {Iyoda}}, \bibinfo {author} {\bibfnamefont {K.}~\bibnamefont {Kaneko}},\ and\
  \bibinfo {author} {\bibfnamefont {T.}~\bibnamefont {Sagawa}},\ }\bibfield
  {title} {\bibinfo {title} {{Fluctuation theorem for many-body pure quantum
  states}},\ }\href {https://doi.org/10.1103/PhysRevLett.119.100601} {\bibfield
   {journal} {\bibinfo  {journal} {Phys. Rev. Lett.}\ }\textbf {\bibinfo
  {volume} {119}},\ \bibinfo {pages} {100601} (\bibinfo {year}
  {2017})}\BibitemShut {NoStop}%
\bibitem [{\citenamefont {Shiraishi}\ and\ \citenamefont
  {Matsumoto}(2021)}]{shiraishi2021undecidability}%
  \BibitemOpen
  \bibfield  {author} {\bibinfo {author} {\bibfnamefont {N.}~\bibnamefont
  {Shiraishi}}\ and\ \bibinfo {author} {\bibfnamefont {K.}~\bibnamefont
  {Matsumoto}},\ }\bibfield  {title} {\bibinfo {title} {{Undecidability in
  quantum thermalization}},\ }\href
  {https://doi.org/10.1038/s41467-021-25053-0} {\bibfield  {journal} {\bibinfo
  {journal} {Nat. Commun.}\ }\textbf {\bibinfo {volume} {12}},\ \bibinfo
  {pages} {5084} (\bibinfo {year} {2021})}\BibitemShut {NoStop}%
\end{thebibliography}%

\end{document}